\documentclass{article}
\usepackage{graphicx} 
\usepackage[margin=1in]{geometry}
\usepackage{amsmath}
\usepackage{algorithm}
\usepackage[noend]{algpseudocode}
\usepackage{amsfonts}
\usepackage{amssymb}
\usepackage{bm}
\usepackage[square, numbers]{natbib}
\usepackage{enumitem}
\usepackage{bbm}
\usepackage{comment}
\usepackage{stmaryrd}
\usepackage{theoremref}
\usepackage{mathrsfs}
\usepackage{amsthm}
\usepackage{booktabs}
\usepackage{color-edits}
\addauthor{jw}{blue}
\usepackage{tikz}
\usetikzlibrary{positioning,arrows.meta}  
\usepackage{subcaption}           
\usepackage[margin=1in]{geometry}

\DeclareMathOperator*{\esssup}{ess\,sup}

\newcommand\independent{\protect\mathpalette{\protect\independenT}{\perp}}
\def\independenT#1#2{\mathrel{\rlap{$#1#2$}\mkern2mu{#1#2}}}

\usepackage{xcolor}
\usepackage{authblk}

\renewcommand{\P}{\mathbb{P}}
\newcommand{\Unif}{\mathrm{Unif}}
\newcommand{\Bern}{\mathrm{Bern}}

\newcommand{\1}{\mathbbm{1}}

\newcommand{\E}{\mathbb{E}}
\newcommand{\R}{\mathbb{R}}
\newcommand{\G}{\mathbb{G}}
\newcommand{\bfphi}{\boldsymbol{\phi}}

\newcommand{\Rem}{\mathrm{Rem}}

\newcommand{\Bias}{\mathrm{Bias}}

\newcommand{\calF}{\mathcal{F}}
\newcommand{\sgn}{\mathrm{sgn}}
\newcommand{\calX}{\mathcal{X}}
\newcommand{\calV}{\mathcal{V}}

\newcommand{\calG}{\mathcal{G}}

\newcommand{\calI}{\mathcal{I}}

\newcommand{\calM}{\mathcal{M}}
\newcommand{\calZ}{\mathcal{Z}}
\newcommand{\calA}{\mathcal{A}}
\newcommand{\calN}{\mathcal{N}}

\newcommand{\calW}{\mathcal{W}}

\newcommand{\calO}{\mathcal{O}}

\newcommand{\fraks}{\mathfrak{s}}

\newcommand{\Cov}{\mathrm{Cov}}

\newcommand{\Var}{\mathrm{Var}}

\newcommand{\Cal}{\mathrm{Cal}}

\newcommand{\range}{\mathrm{range}}

\newcommand{\wh}[1]{\widehat{#1}}
\newcommand{\wt}[1]{\widetilde{#1}}
\newcommand{\wb}[1]{\overline{#1}}

\newcommand{\C}{\mathbb{C}}

\newcommand{\smax}{\mathbf{sm}}
\newcommand{\splus}{\mathbf{sp}}

\newcommand{\true}{\mathrm{true}}
\newcommand{\ass}{\mathrm{ass}}
\newcommand{\reg}{\mathrm{reg}}

\algrenewcommand\algorithmicrequire{\textbf{Input:}}
\algrenewcommand\algorithmicensure{\textbf{Output:}}

\newtheorem{theorem}{Theorem}
\numberwithin{theorem}{section}
\newtheorem{lemma}[theorem]{Lemma}
\newtheorem{informal}{Informal Theorem}

\newtheorem{prop}[theorem]{Proposition}
\newtheorem{corollary}[theorem]{Corollary}

\newtheorem{assumption}{Assumption}

\theoremstyle{definition}
\newtheorem{definition}[theorem]{Definition}

\theoremstyle{remark}
\newtheorem{rmk}[theorem]{Remark}

\usepackage{hyperref}
\hypersetup{
	colorlinks=false,
	bookmarks=true,
	breaklinks=true,
	hidelinks=true,
	pdfpagemode=empty,
}
\usepackage{nameref}
\usepackage[noabbrev, capitalize, nameinlink]{cleveref}
\title{Inference on Optimal Policy Values and Other Irregular Functionals via Softmax Smoothing\thanks{The present version of this article is a combination of the first version with an earlier work on inference on optimal policy values under semi-parametric restrictions~\citep{chen2023inference}.}}
\author[1]{Justin Whitehouse\thanks{Corresponding author: contact at \texttt{jwhiteho@stanford.edu}}}
\author[2]{Qizhao Chen}
\author[2]{Morgane Austern\thanks{MA was supported by NSF CAREER DMS 2441652.}}
\author[1]{Vasilis Syrgkanis\thanks{VS was supported by NSF Award IIS-2337916.}}

\affil[1]{Department of Management Science and Engineering, Stanford University}
\affil[2]{Department of Statistics, Harvard University}

\date{\today}

\begin{document}
\maketitle
\begin{abstract}
    Constructing confidence intervals for the value of an (unknown) optimal treatment policy is a fundamental problem in causal inference. Insight into the optimal policy value can guide the development of reward-maximizing, individualized treatment regimes. However, because the functional that defines the optimal value is non-differentiable, standard semi-parametric approaches for performing inference fail to be directly applicable.
    Many existing works circumvent non-differentiability by making the unrealistic assumption of zero probability of \textit{treatment non-response}, i.e.\ that every unit responds (either positively or negatively) to an assigned treatment.
    Further, works that don't circumvent this restriction rely on refitting nuisance models a number of times \textit{proportional to the sample size}.
    In this paper, we construct and analyze a simple, softmax smoothing-based estimator for the value of an optimal treatment policy. Our estimator applies in both static and dynamic treatment regimes, only requires fitting a constant number of nuisance models, and is statistically efficient when there is zero probability of non-response to treatment. Also, while our estimator does not require making semi-parametric restrictions, it can exploit them when they exist. We further show how our softmax smoothing approach can be used to estimate general parameters that are specified as a maximum of scores involving nuisance components, and look at conditional Balke and Pearl bounds and $L^1$ calibration error as salient examples. 
    
\end{abstract}
\newpage
\section{Introduction}
Individualized treatment policies, or mappings from a unit's observed characteristics to an assigned treatment, are deployed in a wide range of decision-making problems. In medicine, individualized policies may be used to maximize the expected outcome or well-being of a patient~\citep{qian2011performance, xu2022estimating}. In advertising, such policies may be employed for maximizing patron engagement or revenue~\citep{bottou2013counterfactual}. Motivated by such applications, a large literature has developed that focuses on using observational data to learn a treatment policy with a high \textit{value}, i.e.\ a large outcome averaged over a covariate distribution~\citep{athey2021policy,kitagawa2018should, kitagawa2022treatment, hirano2009asymptotics}.

Before investing resources in the construction and deployment of a highly individualized treatment rule,  one may first wish to quantify how much can be gained from policy learning. The natural estimand for assessing this is the \textit{value of the optimal treatment policy}. Let $Z = (X, A, Y)$ denote a observation, where $X$ represents a unit's covariates, $A \in [N]$ denotes a categorical treatment, $Y(a)$ a unit's potential outcome associated with treatment $a \in [N]$, and $Y = Y(A)$ the observed outcome.
Letting $\pi(X) \in [N]$ denote a arbitrary policy that assigns a unit to a treatment, the value of the optimal treatment policy is given by
\[
V^\ast := \max_\pi V(\pi), \;\; \text{where} \;\; V(\pi) := \E[Y(\pi(X))].
\]
Inference on $V^\ast$ can be used to determine the value of developing individualized policies in several ways. For instance, if an advertiser wanted to determine if some baseline treatment policy $\pi_0$ (perhaps a rule-based policy) could be improved through the use of machine learning (ML) algorithms, he could estimate  the policy value gap $V^\ast - V(\pi_0)$. Likewise, if a teacher wanted to determine if a single lesson plan was a good fit for all students, she could estimate the value of personalization, i.e.\ the quantity $V^\ast - \max_{a \in [N]}V(\pi_a)$, where $\pi_a(X) = a$ denotes a constant treatment policy. 
These examples motivate the need for simple, computationally efficient inferential procedures for $V^\ast$.

At first glance, because $V^\ast$ is just a scalar value, it may appear that estimating $V^\ast$ should be easier than learning a good individualized treatment policy $\pi(X)$. However, because the target estimand $V^\ast$ is \textit{irregular}, standard approaches for constructing confidence intervals fail to give valid coverage. In more detail, under standard causal assumptions, $V^\ast$ can be re-written in terms of the outcome regression (i.e.\ Q-function) $Q^\ast(a, x) = \E[Y \mid X = x, A = a]$ as
\begin{equation}
\label{eq:id_Q_func}
V^\ast = \E\left[\max_{\ell \in [N]}\{Q^\ast(\ell, X)\}\right]. 
\end{equation}
When the argument maximizing set $\arg\max_{\ell \in [N]}Q^\ast(\ell, X)$ is non-unique with positive probability, the functional $Q \mapsto \E\left[\max_{\ell} Q(\ell, X)\right]$ is not Gateaux/path-wise differentiable at $Q^\ast$. 
Classical doubly-robust and targeted maximum likelihood estimators require the Gateaux differentiability of the underlying functional, and are thus inapplicable when multiple treatments offer the same expected reward. Restricting ourselves to the case of binary treatments $A \in \{0, 1\}$ and letting $\tau^\ast(X) := Q^\ast(1, X) - Q^\ast(0, X)$ denote the conditional average treatment effect (i.e.\ CATE), these classical estimation approaches are inapplicable precisely when $\P(\tau^\ast(X) = 0) > 0$, i.e.\ when there is a positive probability treatment non-response in the population.\footnote{Throughout this paper, when we say ``non-response to treatment,'' we mean that multiple treatments yield the same expected outcome. A stronger notion of non-response would be that the potential outcomes are identical, i.e.\ that $Y(1) =  Y(0)$ with positive probability in the binary treatment setting.}

This irregularity has forced researchers to develop estimators that are tailor-made for estimating $V^\ast$. Some works assume that the optimal action is almost surely unique \citep{semenova2023aggregated} and others make parametric assumptions on the outcome regression (e.g.\ linearity) to enable inference~\citep{laber2014dynamic, goldberg2014comment, chakraborty2010inference}. While these assumptions may be reasonable in some settings, they can be violated in domains such as medicine or advertising where treatments can be entirely ineffective and response curves may be highly non-parametric. The most flexible estimation approaches, due to \citet{luedtke2016statistical} and \citet{shi2020breaking}, allow for a positive probability of non-response to treatment and also permit non-parametric outcome regressions. This flexibility comes at the cost of refitting a number of nuisance models that grows with the sample size. Refitting nuisances many times may be possible when nuisance models are parametric or datasets are small. However, when datasets are large or black-box learners such as neural networks are used, training a growing number of nuisance models may not be feasible. 
We discuss these approaches amongst others in greater detail in the related work section (Section~\ref{sec:intro:related}). There is thus a need for computationally efficient estimators that allow non-parametric outcome regressions and  non-response to treatment.


\subsection{Our Contributions and Overview}

In this paper, we show that careful smoothing can be used to construct asymptotically normal estimates for both the value of the optimal treatment policy and also a broader class of  maximum-type estimands. Defining the softmax function $\smax^\beta : \R^N \rightarrow \R$ by
\begin{equation}
\label{eq:softmax}
\smax^\beta_\ell\{u_\ell\} \equiv \smax^\beta\{u_1, \dots, u_N\} := \frac{\sum_{i = 1}^N u_i \exp\{\beta u_i\}}{\sum_{j = 1}^N \exp\{\beta u_j\}},
\end{equation}
our approach is to replace the target $V^\ast$ outlined in Equation~\eqref{eq:id_Q_func} with a smoothed surrogate:
\[
\underbrace{V^\ast := \E\left[\max_\ell\{Q^\ast(\ell, X)\}\right]}_{\text{Non-differentiable functional}} \quad\Longrightarrow\quad \underbrace{V^\beta := \E\left[\smax^\beta_\ell\{Q^\ast(\ell, X)\}\right]}_{\text{Smoothed approximation}}.
\]
The functional $Q \mapsto \E\left[\smax^\beta_\ell Q(\ell, X)\right]$ is twice Gateaux differentiable, thus enabling the use of semi-parametric de-biasing techniques for performing inference on $V^\beta$. 
We enumerate our contributions in greater detail below.
\begin{itemize}
\item In Section~\ref{sec:static}, we describe a de-biased estimator for $V^\ast$. Our estimator is based on a smoothed Neyman orthogonal score $\Psi^\beta(Z; Q, \alpha)$ that satisfies $V^\beta = \E\left[\Psi^\beta(Z; Q^\ast, \alpha^{\beta})\right]$, where $\alpha^\beta(a, x)$ is a nuisance function to be defined in the sequel. Under a mild assumption on the density of the sub-optimality gaps $\Delta_k := \max_\ell\{Q^\ast(\ell, X)\} - Q^\ast(k, X)$ near (but excluding) zero, we show that it is possible to obtain root-$n$ inference for the \textit{un-smoothed} value $V^\ast$ (Theorem~\ref{thm:normal_static}, described informally below).
\begin{informal}
\label{thm:informal}
Let $Z_1, \dots, Z_n$ be i.i.d.\ and let $\wh{Q}, \wh{\alpha}$ be nuisance estimates independent from the sample. Suppose each $\Delta_k$ has density bounded by $C t^{\delta - 1}$ for all $0 < t \leq t_0$ and some $\delta > 0$. Then, if $\beta_n = \omega\left(n^{\frac{1}{2(1 + \delta)}}\right)$ and $\|\wh{Q} - Q^\ast\|_{L^2} = o_\P(\beta_n^{-1/2}n^{-1/4})$ (alongside other standard assumptions), the estimator $\wh{V}_n := \frac{1}{n}\sum_{i = 1}^n \Psi^\beta(Z_i; \wh{Q}, \wh{\alpha})$ is asymptotically linear
\[
\sqrt{n}(\wh{V}_n - V^\ast) = \frac{1}{\sqrt{n}}\sum_{i = 1}\rho_V(Z_i) + o_\P(1),
\]
Further, $\rho_V(Z)$ is the efficient influence function when the optimal treatment is unique almost surely.
\end{informal}
Our estimator is computationally efficient, only requiring a \textit{constant} number of nuisance fits. Further, while our main results are stated in terms of sample-splitting, they can be easily extended to cross-fitting (see Remark~\ref{rmk:cross-fit} for a general discussion).
We also present analogous results under semi-parametric restrictions on the blip effects (Section~\ref{sec:static:param}) and for dynamic treatment regimes (Appendix~\ref{app:dynamic}). We empirically evalute the coverage of our estimator in Section~\ref{sec:experiments}.

\item In Section~\ref{sec:irregular}, we generalize our smoothing estimator to irregular estimands of the form
\[
V^\ast := \E\left[\max_{\ell }\psi_\ell(X, g_\ell^\ast)\right],
\]
where $\psi_\ell(x; g_\ell)$ are general affine scores and the true nuisance $g_\ell^\ast$ is a regression.
In this setting, we similarly show how to construct Neyman orthogonal moments and choose smoothing parameters. As examples, we revisit the problem of performing inference on conditional Balke and Pearl bounds on the average treatment effect~\citep{levis2023covariate} and provide an estimator for the $L^1$ calibration error of a regression model \citep{gupta2022post}. 

\end{itemize}


Two theoretical insights drive our analysis. Our first insight is that the particular choice of smoothing function is critically important in the estimation of irregular parameters. Several works have historically leveraged the \textit{softplus} function $\splus^\beta\{u_1, \dots, u_N\} := \frac{1}{\beta}\log\left(\sum_{i = 1}^N e^{\beta u_i}\right)$ to construct smooth surrogates for non-differentiable functions~\citep{levis2023covariate, goldberg2014comment}. This smoother possesses a notable defect: when $u_1, \dots, u_N$ are equal, the soft-plus function differs from the true maximum by exactly $\log(N)/\beta$. In general, the magnitude of this bias is too large to permit asymptotically normal inference around the un-smoothed policy value $V^\ast$.
In contrast, the softmax exhibits \textit{zero bias} in presence of ties, which will enable inference even when argument maximizing set contains multiple elements.

Our second insight is that the softmax function exhibits fast bias decay under mild distributional assumptions on the sub-optimality gaps $
\Delta_k$.
If these gaps admit a density bounded in the sense of Informal Theorem~\ref{thm:informal}, the bias of softmax smoothing will decay as $|V^\ast - V^\beta| = O\left(\beta^{-(1 + \delta)}\right)$ (Lemma~\ref{lem:margin}). 
The provides a  window in which one can select $\beta_n$ to obtain valid inference.
This density assumption is analogous to the \textit{soft-margin} condition considered in \citet{luedtke2016statistical} and \citet{shi2020breaking}, and thus allows for an arbitrary large probability of treatment non-response.\footnote{The classical soft-margin condition states that $\P(0 < \Delta_k \leq t) \lesssim t^{\delta}$ for $t \in (0, t_0]$ for some $\delta, t_0 > 0$.} When $\delta = 1$ and there are binary treatments, this is equivalent to the CATE have bounded density near zero.

Lastly, we note that the softmax function described above is commonly used in the machine learning literature for tasks such as soft Q-learning~\citep{schulman2017equivalence, nachum2017bridging, haarnoja2017reinforcement,haarnoja2018soft}. However, the goal in this literature is not inference, but rather \textit{regret minimization}, for which a naive bias bound of $O(1/\beta)$ actually suffices. For inferential tasks, naive bias bounds generally lead to pessimistic confidence intervals for smoothed estimates~\citep{levis2023covariate, zhang2024winners}. Our work indicates that, with appropriate care, tools that are commonly leveraged in machine learning tasks may be applicable to broad classes of statistical problems.

\subsection{Related Work}
\label{sec:intro:related}
\paragraph{Inference on Optimal Policy Values:} Many works have focused on performing inference on optimal policy values under parametric restrictions on outcome regressions/Q-functions. \citet{laber2014dynamic} assume linear Q-functions and propose estimating the optimal policy value in a dynamic regime using Q-learning. 
Instead of smoothing, they  construct pessimistic confidence intervals in regions of non-regularity (i.e.\ regions of non-responders) via a union bound and bootstrapping methods. 
\citet{goldberg2014comment} consider the same parametric restrictions and show how to use a softplus function to approximate $\max\{a, b\}$. These authors also require conservative estimation when the probability of non-response is positive. We emphasize that our estimator is asymptotically normal even when non-response to treatment occurs, and that our results hold in both non-parametric and semi-parametric settings. \citet{chakraborty2013inference} propose a method based on $m$-out-of-$n$ bootstrapping to build confidence intervals on structural parameters, but assume Q-functions are linear and achieve width $m$ confidence intervals instead of width $n$ (here, $m = o(n)$).

Other authors have considered settings where the Q-functions/conditional average treatment effects (CATEs) are allowed to be fully non-parametric. For instance, \citet{semenova2023aggregated} analyze a first-order de-biased estimator estimator for the optimal policy value. However, this estimator requires that nuisances be estimated at $o(n^{-1/4})$ rates in the $L^\infty$ norm, which may be restrictive when off-the-shelf ML estimators like gradient boosted trees or neural networks are used. Further, they require the reward maximizing action to be unique almost surely. We also mention the work of \citet{park2024debiased}\footnote{This cited paper was released after the first draft of this paper~\citep{chen2023inference}}, who use smoothing but obtain slower, $\sqrt{n/\beta_n^2}$ rates of convergence, where $\beta_n$ grows to infinity (Proposition 1 of their work). In contrast, our Theorem~\ref{thm:normal_static} establishes $\sqrt{n}$ asymptotic linearity of the softmax smoothed estimator, which results in tighter confidence intervals.
Our results also apply to other irregular functionals of interest.

The two works most closely related to the present paper in are those of \citet{shi2020breaking} and \citet{luedtke2016statistical}, which are discussed above.
These works are very general, describing estimators that allow for non-parametric outcome regressions, dynamic treatment regimes, and positive probabilities of treatment non-response. However, both approaches require training a number of nuisance models that \textit{grows} with the number of samples. Our estimator requires only fitting a fixed number of nuisances, but at the cost of selecting a smoothing parameter. There is thus a fundamental tradeoff between our estimator and the two mentioned above.
Because these methods fit a growing number of nuisance models, they require stronger nuisance error control than our estimator. In our setting, it suffices that the $L^2$ nuisance errors vanish in probability, which is the standard in semi-parametric literature~\citep{chernozhukov2018double}. On the other hand, these works require either $L^2$ convergence of the nuisance estimates or high probability bounds on the $L^2$ error (see \citet{bibaut2020sufficient} for a discussion of this in the context of the estimator of \citet{luedtke2016statistical}).

\paragraph{Policy Learning:}

In the causal inference and reinforcement learning (RL) literature, a large body of work has developed related to learning an optimal or near-optimal individualized treatment policy. These include approaches based on Q-learning~\citep{zhang2012robust, watkins1992q, watkins1989learning, qian2011performance,zhao2009reinforcement, zhao2011reinforcement, moodie2012q}, A-learning (or advantage learning)~\citep{shi2024statistically, murphy2003optimal, moodie2007demystifying, robins2004optimal}, empirical risk minimization~\citep{athey2021policy, luedtke2020performance, kallus2018confounding}, and more. We particularly emphasize works relating to empirical risk minimization~\citep{athey2021policy, kitagawa2018should, kitagawa2022treatment}, which describe how to learn policies with low \textit{regret}, which is defined as the difference in value between the best policy in the class and the learned policy~\citep{manski2004statistical, hirano2009asymptotics, luedtke2020performance}. We also note works on optimal policy learning in partially-identified settings~\citep{olea2023decision, kitagawa2023treatment, pu2021estimating}, which may be of particular importance when one does not believe commonly-made conditional exogeneity assumptions. We emphasize that the algorithms in the works listed above do not discuss inference on the value of the optimal treatment policy. Also relevant to our work are algorithms from the RL literature based on soft-Q learning~\citep{garg2021iq, nachum2017bridging, schulman2017equivalence}, which use entropy regularization to learn a near-optimal Boltzmann policy. Again, these algorithms do not provide a means for a learner to perform inference on optimal policy values.

\paragraph{Causal Inference:}

Lastly, we touch on tools from the causal inference and semi-parametric statistics that are relevant to our work, including work on semi-parametric efficiency and estimation~\citep{kosorok2008introduction, newey1994asymptotic, van2000asymptotic,robins1995analysis,robins1995semiparametric,bickel1993efficient, levit1976efficiency}, missing data and doubly-robust estimation~\citep{robins2004optimal,bang2005doubly,tsiatis2006semiparametric}, and targeted maximum likelihood estimation~\citep{van2006targeted, van2011targeted}. Most closely related to our work is the literature of double/de-biased machine learning~\citep{bang2005doubly, chernozhukov2018double}, which advocates performing generic first-order corrections and cross-fitting to reduce bias in estimation problems involving nuisance components. These approaches are commonly leveraged in both causal inference~\citep{wager2018estimation, athey2019generalized,runge2023causal}, causal machine learning~\citep{foster2023orthogonal, whitehouse2024orthogonal, lan2025meta, van2007super}, and automatic de-biased learning methods~\citep{hirshberg2021augmented, chernozhukov2022automatic, chernozhukov2022riesznet, chernozhukov2022nested}. 

\subsection{Notation}
\label{sec:back:notation}


We let $P_V$ be distribution of a arbitrary random variable $V$.  We let $\E[\cdots]$ and $\P(\cdots)$ denote expectation and probability respectively over all sources of randomness.  For arbitrary independent random variables $U$ and $V$ with distributions $P_U$ and $P_V$ and an arbitrary function $m(u, v)$, we let $\E_V[m(U, V)] := \int m(U, v)P_V(d v) = \E_{U,V}[m(U, V) \mid U]$. In this case $\E[m(U, V)] = \E_{U, V}[m(U, V)] = \int\int m(u, v) P_U(du) P_V(dv)$. Given samples $Z_1, \dots, Z_n$, we let $\P_n := \frac{1}{n}\sum_{m = 1}^n \delta_{Z_m}$ denote the corresponding empirical distribution, $\P_n f(Z) := \frac{1}{n}\sum_m f(Z_m)$, and $\G_n := \sqrt{n}(\P_n - \E_Z)$.

For a given integer $N > 0$, we let $[N] := \{1, 2, \dots, N\}$. For $p \in [1, \infty)$, a random variable $V \in \calV$ with distribution $P_V$, and a function $f : \calV \rightarrow \R^d$, we let $\|f\|_{L^p(P_V)} := \left(\E_V\|f(V)\|_p^p\right)^{1/p}$, where $\|a\|_p := \left(\sum_i |a_i|^p\right)^{1/p}$ for $a \in \R^d$.  For $p = \infty$, we let $\|f\|_{L^\infty(P_V)} := \esssup |f(V)|$. We define the $L^p$ space of functions by $L^p(P_V) := \left\{f : \calV \rightarrow \R^d \text{ s.t. } \|f\|_{L^p(P_V)} < \infty\right\}$. The dimension $d$ of the range of functions will either be clear from context or made explicit via the notation $L^p(P_V; \R^d)$. 

Given a Banach space $B$,  $T : B \rightarrow \R^d$, and  $f^\ast, g \in B$, we define the Gateaux derivative of $T$ at $f^\ast$ in the direction $g$ as $D_f T(f^\ast)(g) := \frac{\partial}{\partial t}T(f^\ast + tg) \vert_{t = 0}$. We  define the second Gateaux derivative as $D_{f}^2T(f^\ast)(g) := \frac{\partial^2}{\partial t^2}T(f^\ast + tg)\vert_{t = 0}$. If we have  a bi-variate map $T : B\times B \rightarrow \R^d$ and $f_1^\ast, f_2^\ast, g_1, g_2 \in B$, we let the cross Gateaux derivative be defined as $D_{f_1, f_2}T(f_1^\ast, f_2^\ast)(g_1, g_2) := \frac{\partial^2}{\partial s \partial t}T(f_1^\ast + tg_1, f_2^\ast + s g_2)\vert{t = s = 0}$. For $f : \R^d \rightarrow \R$, we let $\partial_{x_i} f(x)$ denote the $i$th partial derivative, $\nabla_x f(x)$ the gradient, and $\nabla_x^2 f(x)$ the Hessian, when said objects exist. If $f : \R^d \rightarrow \R^m$, we will let $\partial_x f(x)$ denote the Jacobian. When denote the derivative of composed functions with respect to the argument of the outer function, we use the notation $\partial_i f(g(x)) := \partial_{u_i} f(u) \vert_{u = g(x)}$, $\nabla_u f(g(x)) := \nabla_u f(u)\vert_{u = g(x)}$, and $\nabla_u^2 f(g(x)) = \nabla_u^2 f(u)\vert_{u = g(x)}$ for convenience unless otherwise noted. Lastly, for functions $f, g : \calV \rightarrow \R^{d}$, we define the bracket between $f$ and $g$ as  $[f, g] := \{h : \calZ \rightarrow \R : \exists \lambda : \calX \rightarrow [0, 1] \text{ s.t. } h(z) = \lambda(z)f(z) + (1 - \lambda(z))g(z)\}$.

\section{Problem Formulation and Preliminaries}
\label{sec:back}

In this section, we formalize the problem of performing inference on the value of the optimal treatment policy.  We present our assumptions on the underlying data-generating process and identify the optimal value both in terms of outcome regressions and blip effects. The latter identification will be used for performing inference under semi-parametric restrictions (see Section~\ref{sec:static:param}). In the main body of the paper we just focus on static treatment regimes, i.e.\ settings where there is a single round of treatment. We  generalize this discussion to dynamic (namely, two-stage) treatment regimes in Appendix~\ref{app:dynamic}. We conclude by discussing softmax smoothing, the main density assumption made throughout this work, and the bias introduced by smoothing. To streamline our presentation, we defer discussion of the case of general irregular functionals to Section~\ref{sec:irregular}.

\subsection{Assumptions and Identification for Static Treatment Regimes}
\label{sec:back:static}

We assume observations are of the form $Z = (X, A, Y)$, where $X \in \calX$ represents an individual's  covariates, $A \in [N] := \{1, 2, \dots, N\}$ represents a treatment coming from a fixed set of options, and $Y \in \R$ represents an observed outcome. We make the following assumptions on the data generating process.

\begin{assumption}
\label{ass:static}
We assume the learner has access to i.i.d.\ samples $Z = (X, A, Y)$ drawn from some distribution $P_Z$ satisfying the following conditions:
\begin{enumerate}
    \item \textit{(Consistency)} There are potential outcomes $(Y(a) : a \in [N])$ such that $Y = Y(A)$.
    \item \textit{(Conditional Ignorability)} We assume that the potential outcomes $(Y(a) : a \in [N])$ are conditionally independent of treatment $A$ given covariates $X$, i.e.\ $Y(a) \independent A \mid X$ for each $a \in [N]$.
    \item \textit{(Strong Positivity)} Let $p^\ast(a \mid x) := \P(A = a \mid X = x)$ represent the propensity score. We assume $p^\ast$ is strictly bounded away from $0$ --- i.e.\ there is $\eta \in (0, 1/2]$ such that, for any action $\ell \in [N]$ and $x \in \calX$,
    \[
    \eta \leq p^\ast(\ell \mid x).
    \]
\end{enumerate}
\end{assumption}

In what follows, we do not assume that the propensity $p^\ast(a \mid x)$ is known. By a treatment policy, we mean a  mapping $\pi : \calX \rightarrow [N]$ taking an individual's covariates to a categorical treatment. Under the above assumption, straightforward manipulation yields that $V^\ast := \max_{\pi}\E[Y(\pi(X))]$ is identified as $V^\ast = \E\left[\max_\ell\{Q^\ast(\ell, X)\}\right]$, which is precisely the characterization provided in Equation~\eqref{eq:id_Q_func}.

We can also characterize the value of the optimal treatment policy via \textit{blip effects}. The blip effect $\gamma^\ast(a, x) := Q^\ast(a, x)  - Q^\ast(a^\ast, x)$ gives the expected value of an action $a \in [N]$ for an individual with covariates $x \in \calX$ relative to some baseline or control action $a^\ast \in [N]$. Given that one can write the value of a policy $\pi$ relative to the observational policy as
\begin{align*}
\E[Y(\pi(X)) - Y(A)]  = \E[Y(\pi(X)) - Y(a^\ast) - (Y(A) - Y(a^\ast))] = \E[\gamma^\ast(\pi(X), X) - \gamma^\ast(A, X)],
\end{align*}
it follows that the value of an optimal treatment policy can also be characterized as
\begin{equation}
\label{eq:id_blip_eff}
V^\ast = \E\left[\max_\ell \gamma^\ast(\ell, X) - \gamma^\ast(A, X) + Y\right].
\end{equation}
We use this identification in Section~\ref{sec:static:param}, in which we discuss the estimation of the optimal policy value under semi-parametric restrictions on the blip effects.


\subsection{Density Assumption and Smoothing}
\label{sec:back:smooth}
As noted in the introduction, the mapping $Q \mapsto \E\left[\max_{\ell}Q(\ell, X)\right]$ is not generally Gateaux differentiable at $Q^\ast$ when there are non-responders to treatment, i.e.\ when
\[
\P\left(\left|\arg\max_{\ell \in [N]}Q^\ast(\ell, X)\right| > 1\right) > 0.
\]
In this setting, such laws are generally termed to be \textit{irregular} or \textit{nonregular}~\citep{shi2020breaking, luedtke2016statistical}. When the law is irregular, \citet{hirano2012impossibility} show that regular asymptotically linear (RAL) estimators fail to exist in general. Consequently, standard notions of semi-parametric efficiency may fail to be applicable. Likewise, de-biased or one-step estimators~\citep{chernozhukov2018double, van2006targeted} rely on the Gateaux differentiability of the underlying functional, and thus cannot be used.

To work around this irregularity, we defined the smoothed surrogate $V^\beta := \E\left[\smax^\beta_\ell\{Q^\ast(\ell, X)\}\right]$. Since the softmax function $\smax^\beta\{u\}$ is twice continuously differentiable with respect to $u$, the above surrogate is twice Gateaux differentiable with respect to $Q$ for any fixed $\beta > 0$. This enables the estimation of $V^\beta$ via standard de-biased estimators~\citep{chernozhukov2018double, van2006targeted}.
We ultimately want to obtain root-$n$ inference for $V^\ast$ through a de-biased estimator for $V^\beta$. To do this, we need to be able to select the smoothing parameter $\beta$ as a function of the sample size in a way to ensure the bias decays like $|V^\beta - V^\ast| = o(n^{-1/2})$. The bias of softmax smoothing is governed by the distribution of the sub-optimality gaps 
\[
\Delta_k := \max_\ell\{Q^\ast(\ell, X)\} - Q^\ast(k, X),
\]
which are small when action $k$ is close to optimal, and large otherwise. Without further assumptions on these gaps, the worst-case bias of smoothing decays as $|V^\beta - V^\ast| = \Theta\left(\tfrac{1}{\beta}\right)$. In this regime, to ensure $|V^\beta - V^\ast| = o(n^{-1/2})$, one would need to take $\beta = \omega(n^{1/2})$. In the sequel, we will see that taking $\beta$ to be this large would require the learner to estimate the unknown regression $Q^\ast$ at \textit{super-parametric} $o_\P(n^{-1/2})$ rates, which is generally impossible.

Our insight is that, under a mild assumption of the density of the gaps $\Delta_k$, the bias will shrink rapidly as the smoothing parameter grows. Our main assumption formally codifies this assumption on the sub-optimality gaps, and is analogous to the assumptions considered in \citet{luedtke2016statistical} and \citet{shi2020breaking} (which are instead stated in terms of the CDF of the sub-optimality gaps).

\begin{assumption}
\label{ass:margin}
We say a non-negative random variable $\Delta \in \R_{\geq 0}$ satisfies the $\delta$-polynomial density assumption if there exists $c, H > 0$ such that the random variable $\Delta\mathbbm{1}\{\Delta \in (0, c)\}$ admits a Lebesgue density $f$ satisfying $f(t) \leq Ht^{\delta - 1}$.
\end{assumption}

\begin{figure}[htbp]
    \centering
    
    \begin{subfigure}[t]{0.32\textwidth}
        \centering
        \includegraphics[width=\linewidth]{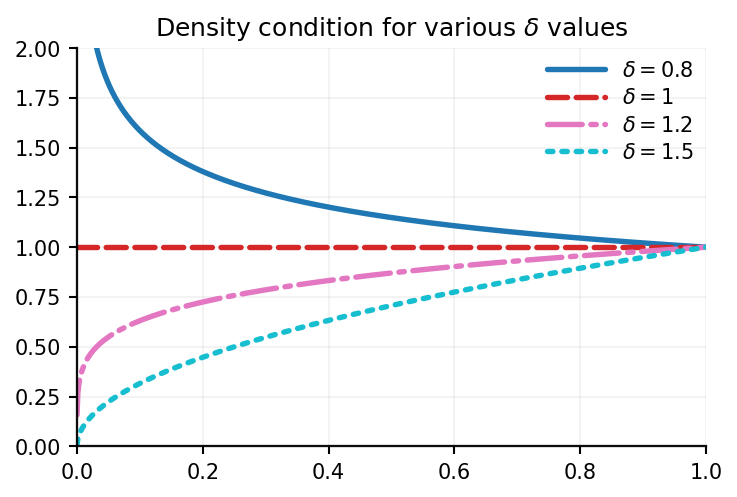}
        \caption{Visualization of Assumption~\ref{ass:margin} for various $\delta$.}
        \label{fig:combined:margin}
    \end{subfigure}
    \hfill
    \begin{subfigure}[t]{0.32\textwidth}
        \centering
        \includegraphics[width=\linewidth]{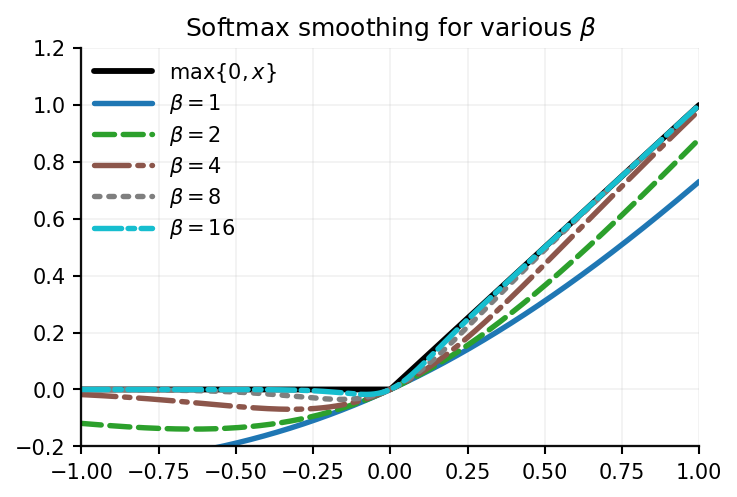}
        \caption{Softmax smoothing.}
        \label{fig:combined:softmax}
    \end{subfigure}
    \hfill
    \begin{subfigure}[t]{0.32\textwidth}
        \centering
        \includegraphics[width=\linewidth]{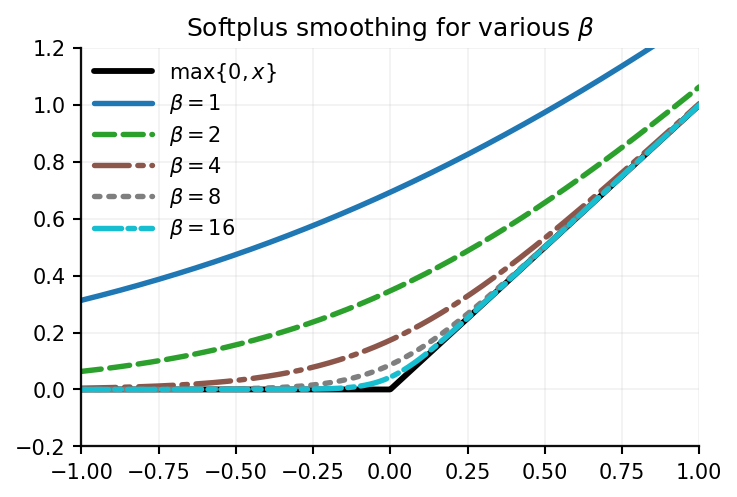}
        \caption{Softplus smoothing.}
        \label{fig:combined:softplus}
    \end{subfigure}
    
    \caption{An illustration of the density condition outlined in Assumption~\ref{ass:margin} alongside a comparison of softmax and softplus approximations for the map $x \mapsto \max\{0,x\}$. Panel (a) illustrates that, when $\delta < 1$, the density may diverge near zero. Panels (b) and (c) illustrate the bias of softplus smoothing when $x = 0$.}
    \label{fig:margin-and-smoothing}
\end{figure}

When Assumption~\ref{ass:margin} is applied to the $\Delta_k$, it allows the maximizing index of $Q^\ast(X) = (Q^\ast(1, X), \dots, Q^\ast(N, X))$ to be non-unique with positive probability. The parameter $\delta > 0$ governs the density of near ties: $\delta > 1$ indicates the density of $\Delta_k$ is small near zero, $\delta = 1$ corresponds to a bounded density, and $\delta \in (0, 1)$ indicates the $\Delta_k$ may be heavy-tailed. These cases are illustrated in Figure~\ref{fig:margin-and-smoothing}. 
Under this assumption, we can prove a bound on the bias of softmax smoothing, which we prove in Appendix~\ref{app:softmax}.
\begin{lemma}
\label{lem:margin}
Let $\beta > 0$ be a temperature and let $U_1, \dots, U_N$ be random variables. Define the random differences $\Delta_k := \max_\ell U_\ell - U_k \geq 0$. First, one always has
\[
\E\left[\max_\ell U_\ell - \smax_\ell^\beta U_\ell\right] \leq \sum_{\ell = 1}^N \E\left[\Delta_\ell \exp\left\{-\beta \Delta_\ell\right\}\right].
\]
Second, if each $\Delta_k$ satisfies Assumption~\ref{ass:margin} with constants $c, H, \delta > 0$, then there exist absolute constants $C, \beta_\ast > 0$ only depending on $c, H,$ and $\delta$ such that
\[
\E\left[\Delta_k \exp\left\{-\beta \Delta_k\right\}\right] \leq C\left(\frac{1}{\beta}\right)^{1 + \delta}\;\;\text{for all } \beta \geq \beta_\ast.
\]
\end{lemma}

To gain visual intuition for the above lemma, we can look at Panel (b) of Figure~\ref{fig:margin-and-smoothing}. This panel illustrates a softmax smoothing approximation for the function $\max\{0, x\}$. From the figure, it is clear that the bias from smoothing is most pronounced near (but not at) zero. Further, as $\beta$ grows, the region in which the bias is pronounced shrinks. Thus, if the gaps are unlikely to be very close to the origin (in the sense of Assumption~\ref{ass:margin}), we should be able to get faster than the naive $O\left(\tfrac{1}{\beta}\right)$ bias control.

\section{Inference on the Value of the Optimal Treatment Policy}
\label{sec:static}

In this section, we propose a de-biased estimator for the value of the optimal treatment policy based on softmax smoothing. Since $Q \mapsto \E\left[\max_\ell\{Q(\ell, X)\}\right]$ was not generally Gateaux differentiable, we proposed replacing it by a smoothed alternative $V^\beta := \E\big[\smax^\beta_\ell\{Q^\ast(\ell, X)\}\big]$. Our first goal is to provide a \textit{de-biased} or \textit{Neyman orthogonal} score for $V^\beta$~\citep{chernozhukov2018double}.

\begin{definition}
\label{def:neyman}
A score $m(Z; g, \theta)$ depending on an observation $Z$, nuisance function $g$, and finite-dimensional parameter $\theta$ is said to be \textit{Neyman orthogonal} at a pair $(\theta_0, g^\ast)$ if 
\[
D_g \E\left[m(Z; \theta_0, g^\ast)\right](\omega) \;\; \text{for all } \omega \in \calG,
\]
 where $\calG$ is some space of perturbation functions.
\end{definition}

Typically, $\theta_0$ is the unique solution to the moment equation $0 = \E\left[m(Z; \theta_0, g^\ast)\right]$, and one regularly has $g^\ast, \omega \in \calG := L^2(P_W; \R^p)$, where $p \geq 1$ and $W \subset Z$ denotes a subset of features. We will always assume this to be the case unless otherwise stated. If we are interested in a parameter $\theta_0 = \E[m(Z; g^\ast)]$ instead, orthogonality reduces to the condition that $D_g\E[m(Z; g^\ast)](\omega) = 0$. Heuristically, a score is Neyman orthogonal if any errors in nuisance estimation are only second order in nature. A standard example of a Neyman orthogonal score is the doubly-robust pseudo-outcome leveraged in the estimation of average treatment effects~\citep{kennedy2016semiparametric, kennedy2024semiparametric}. The following proposition describes a Neyman orthogonal score for $V^\beta$ for any fixed $\beta > 0$. We provide a proof in Appendix~\ref{app:treatment}.

\begin{prop}
\label{prop:ortho_static}
For any $\beta > 0$, consider the score $\Psi^\beta$ defined as
\[
\Psi^\beta(Z; Q, \alpha) := \smax^\beta_{\ell}Q(\ell, X) + \sum_{k = 1}^N\alpha_k(A, X)(Y - Q(A, X)),
\]
where the true nuisances are $Q^\ast(a, x) = \E[Y \mid X = x, A =a]$ and $\alpha^\beta(a, x) = (\alpha^\beta_1(a, x), \dots, \alpha^\beta_N(a, x))$. Here, $\alpha^\beta$ has $k$th coordinate
\[
\alpha^\beta_k(a, x) := \partial_k\smax^\beta_\ell\{Q^\ast(\ell, x)\}\frac{\mathbbm{1}\{a = k\}}{p^\ast(k \mid x)},
\]
where the partial derivative of the softmax function is given by
\begin{equation}
\label{eq:softmax_partial}
\partial_k \smax^\beta_\ell\{u\} := \frac{e^{\beta u_k}}{\sum_j e^{\beta u_j}}\left[1 + \beta \left\{u_k - \smax^\beta_\ell u\right\}\right]
\end{equation}
for any $u \in \R^N$. Then, $\Psi^\beta$ is Neyman orthogonal at $(Q^\ast, \alpha^\beta)$, and $V^\beta = \E\left[\Psi^\beta(Z; Q^\ast, \alpha^\beta)\right]$.
\end{prop}


In general, the solution $V^\beta$  to the smoothed, de-biased score is not the same as the optimal policy value $V^\ast$. However, by appropriately selecting an increasing sequence of smoothing parameters $(\beta_n)_{n \geq 1}$ and exploiting the bias properties outlined in Lemma~\ref{lem:margin}, one can use the above score to obtain asymptotically linear (and hence normal) estimates of $V^\ast$. The following theorem describes an estimator for $V^\ast$ in which nuisance estimates are independent of the data. This sample-splitting approach is just for ease of exposition, and one can show that a natural cross-fitting variant of Theorem~\ref{thm:normal_static} is true as well (see Remark~\ref{rmk:cross-fit} below). 

\begin{theorem}
\label{thm:normal_static}
Let $\beta_n$ be a  smoothing parameter. Let $Z_1, \dots, Z_n$ be i.i.d.\ random variables satisfying Assumption~\ref{ass:static}, and let $\wh{Q}_n$ and $\wh{\alpha}_n$ be nuisance estimates for $Q^\ast$ and $\alpha^{\beta_n}$ that are independent of $Z_1, \dots, Z_n$. Suppose the following assumptions hold.
\begin{enumerate}
    \item \textit{(Controlled Density Near Zero)} For each $k \in [N]$, the random variable $\Delta_k := \max_\ell\{Q^\ast(\ell, X)\} - Q^\ast(k, X)$ satisfies Assumption~\ref{ass:margin} with constants $c, H, \delta > 0$.
    \item \textit{(Growing Smoothing Parameters)} The smoothing parameter $\beta_n$ satisfies $\beta_n = \omega\left(n^{\frac{1}{2(1 + \delta)}}\right)$.
    \item \textit{(Nuisance Convergence)} The following nuisance estimation rates hold:
    \[
    \|\wh{Q}_n - Q^\ast\|_{L^2(P_W)} = o_\P(n^{-1/4}\beta_n^{-1/2}) \quad \text{and} \quad \|\wh{\alpha}_{n} - \alpha^{\beta_n}\|_{L^2(P_W)}\times\|\wh{Q}_{n} - Q^\ast\|_{L^2(P_W)} = o_\P(n^{-1/2}).
    \]
    Further, we assume all nuisance estimates are consistent in $L^2(P_W)$, i.e.\ $\|\wh{Q}_n - Q^\ast\|_{L^2(P_W)} = o_\P(1)$ and $\|\wh{\alpha}_n - \alpha^{\beta_n}\|_{L^2(P_W)} = o_\P(1)$.
    \item \textit{(Boundedness)} There is some universal constant $G > 0$ such that:
    \[
    \|Q^\ast\|_{L^\infty(P_W)}, \|\wh{Q}_n\|_{L^\infty(P_W)}, \|\alpha^{\beta_n}\|_{L^\infty(P_W)}, \|\wh{\alpha}_n\|_{L^\infty(P_W)}, \|Y\|_{L^\infty(P_Z)} \leq G.
    \]
\end{enumerate}
Then, the estimate $\wh{V}_n := \P_n \Psi^{\beta_n}(Z; \wh{Q}_n, \wh{\alpha}_n)$ is asymptotically linear, i.e.\
\[
\sqrt{n}(\wh{V}_n  - V^\ast) := \frac{1}{\sqrt{n}}\sum_{i =1 }^n \rho_V(Z_i) + o_\P(1),
\]
where the influence function $\rho_V$ is given by
\begin{equation}
\label{eq:influence-function}
\rho_V(Z) := \max_{\ell}\{Q^\ast(\ell, X)\} + \sum_{k = 1}^N \fraks_k(Q^\ast(X))\frac{\1\{A = k\}}{p^\ast(k \mid X)}\left\{Y - Q^\ast(A, X)\right\} - V^\ast
\end{equation}
and $\fraks_k(u) \equiv \fraks_k(u_1, \dots, u_N) := \frac{\mathbbm{1}\big\{k \in \arg\max_\ell u\big\}}{|\arg\max_\ell(u)|}$. Consequently, we have $
\sqrt{n}(\wh{V}_n - V^\ast) \Rightarrow \calN\left(0, \Sigma^\ast\right)$ 
when $\Sigma^\ast := \Var[\rho_V(Z)] > 0$.
\end{theorem}

The above theorem provides an asymptotically valid range in which the learner can select a smoothing parameter. On one hand, we must smooth sufficiently aggressively to make the bias negligible. This is captured through the requirement that $\beta_n = \omega\left(n^{\frac{1}{2(1 + \delta)}}\right)$. On the other hand, $\beta_n$ must grow sufficiently slowly to ensure the estimating the unknown regression is  feasible. This emerges through the condition that $\|\wh{Q} - Q^\ast\|_{L^2(P_W)}  = o_\P(\beta_n^{-1/2}n^{-1/4})$. This places an implicit upper bound of $\beta_n = o(n^{1/2})$ on the smoothing parameter: if $\beta_n = \Omega(n^{1/2})$, then we would need $\|\wh{Q} - Q^\ast\|_{L^2(P_W)} = o_\P(n^{-1/2})$, which is generally not possible. We note that such ``faster than $n^{-1/4}$'' rates on the regression $Q^\ast$ are also required in \citet{luedtke2016statistical} and \citet{shi2020breaking} --- see Remark~\ref{rmk:nuisance_rates} below.

We can interpret the asymptotic behavior of our estimator in the context of the results of \citet{luedtke2016statistical} and \citet{van2014targeted}. When the maximizing coordinate of $Q^\ast(X) := (Q^\ast(1, X), \dots, Q^\ast(N, X))$ is unique almost surely, \citet{van2014targeted} remarkably show that performing inference on $V^\ast$ is asymptotically equivalent to being directly given the (almost surely unique) optimal policy $\pi^\ast(x) := \arg\max_\ell Q^\ast(\ell, X)$ and performing policy evaluation. More formally, they show that the efficient influence function for $V^\ast$ is given by
\[
\rho^{\reg}_V(Z) := \underbrace{\max_\ell\{Q^\ast(\ell, X)\}}_{= Q^\ast(\pi^\ast(X), X)} + \frac{\1\{A = \pi^\ast(X)\}}{p^\ast(\pi^\ast(X) \mid X)} \left\{Y - Q^\ast(A, X)\right\} - V^\ast.
\]
This exactly aligns with our influence function $\rho_V(Z)$ defined in Equation~\eqref{eq:influence-function}, since in the regular setting $\fraks_k(Q^\ast(X)) = 1$ if and only if $k = \arg\max_\ell Q^\ast(\ell, X)$. In particular, this shows our estimator is semi-parametric efficient in this setting.

When the maximizing coordinate of $Q^\ast = (Q^\ast(1, X), \dots, Q^\ast(N, X))$ is non-unique with positive probability, efficiency results in the standard semi-parametric sense are not possible~\citep{hirano2012impossibility}. Further, there is no longer an almost surely unique optimal treatment policy. Any (potentially randomized) treatment policy $\pi$ satisfying $\pi(X) \in \arg\max_\ell Q^\ast(\ell, X)$ almost surely will attain the optimal policy value. Note that for any fixed $u \in \R^N$, the vector $(\fraks_1(u), \dots, \fraks_N(u))$ defines a probability distribution on $[N]$ that samples an action from the maximizing set $\arg\max_\ell\{u_1, \dots, u_N\}$ uniformly at random. Thus, the influence function $\rho_V(Z)$ corresponds to doing policy evaluation on the randomized policy that selects an optimal action uniformly at random when there are ties.


We do not prove Theorem~\ref{thm:normal_static} in the main body of the paper. We provide our proof in Appendix~\ref{app:treatment}, which follows as a special case of our result on more general irregular functionals (Section~\ref{sec:irregular}). At a high level, the proof follows from a careful balancing of the smoothing bias and the second order nuisance estimation errors. In our arguments, we leverage both Lemma~\ref{lem:margin} and various analytical properties of the softmax function, which are noted in Lemma~\ref{lem:softmax}. 
We now provide a very brief proof sketch illustrating how one should select $\beta_n$ to ensure asymptotic normality.

\begin{proof}[Proof Sketch for Theorem~\ref{thm:normal_static}]For any $\beta > 0$, let $\Bias(\beta) := \left|V^\beta - V^\ast\right|$ denote the bias incurred through smoothing and let $\Rem(\beta)$ denote the second order error obtained from Taylor expanding $\E\left[\smax^\beta_\ell\{\wh{Q}_n(\ell, X)\}\right]$ around $V^\beta = \E\left[\smax^{\beta}_\ell\{Q^\ast(\ell, X)\}\right]$, which is formally given by
\[
\Rem(\beta) := \frac{1}{2}D_Q^2 \E\left[ \smax^\beta_\ell\{\wb{Q}(\ell, X)\}\right](\wh{Q}_n - Q^\ast).
\]
Here, $\wb{Q} = \wb{t}Q^\ast + (1 - \wb{t})\wh{Q}_n$ for some $\wb{t} \in [0, 1]$.
After some manipulation of terms and under standard $o_\P(n^{-1/2})$ nuisance estimation rates on the product $\|\wh{Q} - Q^\ast\|_{L^2(P_W)}\|\wh{\alpha} - \alpha^{\beta_n}\|_{L^2(P_W)}$, one arrives at the decomposition
\begin{align*}
\sqrt{n}(\wh{V} - V^\ast) &= \sqrt{n}\Bias(\beta_n) + \underbrace{\frac{1}{\sqrt{n}}\sum_{m = 1}^n \left\{\Psi^{\beta_n}(Z_m; Q^\ast, \alpha^{\beta_n}) - V^{\beta_n}\right\}}_{\text{Asymptotically Normal Term}} + \sqrt{n}\Rem(\beta_n) + o_\P(1).
\end{align*}

The ``asymptotically normal term'', as the name would suggest, will converge to a normal distribution via Theorem~\ref{thm:smooth_clt} in Appendix~\ref{app:clt}, a central limit theorem for generic, smoothed functionals. To obtain overall asymptotic linearity/normality, one just needs to show that $\Bias(\beta_n) = o(n^{-1/2})$ and $\Rem(\beta_n) = o_\P(n^{-1/2})$. Leveraging bounds on the operator norm of the softmax Hessian (namely, Lemma~\ref{lem:softmax} Point~\ref{pt:softmax_hess}), one can show $\Rem(\beta_n) \lesssim \beta_n \|\wh{Q} - Q^\ast\|_{L^2(P_W)}^2$. Likewise, under the Assumption~\ref{ass:margin}, Lemma~\ref{lem:margin} yields that $\Bias(\beta_n) = O\left(\beta^{-(1 + \delta}\right)$. Rearranging, one finds that, so long as $\beta_n = \omega\left(n^{\frac{1}{2(1 + \delta)}}\right)$ and $\|\wh{Q} - Q^\ast\|_{L^2(P_W)} = o_\P(n^{-1/4}\beta_n^{-1/2})$, asymptotic linearity (and hence normality) is obtained. 

\end{proof}

The following corollary allows us to build confidence intervals using the quantiles of the standard normal distribution. Again, we prove said corollary in Appendix~\ref{app:treatment}. The proof largely boils down to showing the consistency of the plug-in variance estimate and then applying the continuous mapping theorem.

\begin{corollary}
\label{cor:normal_static}
Under the conditions of Theorem~\ref{thm:normal_static}, letting $\wh{\Sigma}_n := \frac{1}{n}\sum_{m = 1}^n\left\{ \Psi^{\beta_n}(Z_m; \wh{Q}_n, \wh{\alpha}_n) - \wh{V}_n\right\}^2$, we have 
\[
\sqrt{n}\wh{\Sigma}_n^{-1/2}(\wh{V}_n - V^\ast) \Rightarrow \calN(0, 1).
\]
Consequently, for any $\delta \in (0, 1)$, the set
\[
C_{1 -\delta} := \left[\wh{V}_n - \frac{\wh{\Sigma}_n^{1/2}}{\sqrt{n}}z_{1 -\delta/2}, \wh{V}_n + \frac{\wh{\Sigma}_n^{1/2}}{\sqrt{n}}z_{1- \delta/2}\right]
\]
forms an asymptotic $1 - \delta$ confidence interval for $V^\ast$, where $z_{1-\delta/2}$ is the $(1 - \delta/2)$th quantile of the standard normal distribution.
\end{corollary}



\begin{rmk}
\label{rmk:nuisance_rates}
We now parse the density assumption in the context of work by \citet{luedtke2016statistical}. First, suppose we were able to exactly take $\beta_n = \Theta\left(n^{\frac{1}{2(1 + \delta)}}\right)$. Then, we would need a regression estimation rate of $\|\wh{Q}_n - Q^\ast\|_{L^2(P_W)} = o_\P\left(n^{-\frac{1}{4}\frac{2 + \delta}{1 + \delta}}\right)$, which exhibits the same sample size dependence as noted in Theorem 8 of \citet{luedtke2016statistical}. We consider several different density assumptions on $\tau^\ast$ below.
\begin{enumerate}
    \item \textit{(Bounded Density)} The case $\delta = 1$ corresponds to the density of the sub-optimality gaps $\Delta_i$ being bounded near, but not at zero. In this case, the learner must estimate the Q-functions  at a $o_\P(n^{-3/8})$ rate in the $L^2(P_W)$ norm. This departs from the traditional $o_\P(n^{-1/4})$ rates required on nuisance estimates when the score is twice differentiable and has bounded Hessian. It is still slower than the super-parametric rate of $o_\P(n^{-1/2})$ that would be required from non-orthogonal scores.
    \item \textit{(Unbounded Density near Zero)} When $\delta = 1/2$, the density of the gaps approaches infinity near zero. In this case, we need a rate of $o_\P(n^{-5/12})$ for regression estimates, which is naturally faster than what was required in the setting of a bounded density. 
    \item \textit{(Quadratic Density near Zero)} As a final example, when $\delta = 2$, the density of $\tau^\ast(X)$ puts very little mass near the origin. Because of this, a learner can get away with slower $o_\P(n^{-1/3})$, which is slower than the case of bounded density. More generally, as $\delta \rightarrow \infty$, the closer we get to being able to assume $o_\P(n^{-1/4})$ estimation rates on the regression.
\end{enumerate}

One subtlety in our setting is that we cannot actually take $\beta_n$ to be on the order of $n^{\frac{1}{2(1 + \delta)}}$ --- rather we must have $\beta_n = \omega\left(n^{\frac{1}{2(1 + \delta)}}\right)$. Thus, we will need to estimate nuisances at slightly faster rates than those described above. However, we view this as minor, as we can simply take an arbitrarily slowly growing function $f(n)$ (say $f(n) := \log^\ast(n)$)\footnote{The iterated logarithm function $\log^\ast(n)$, which commonly appears in the analysis of algorithms in the theoretical computer science literature, is given by the recurrence $\log^\ast(n) = 0$ for $n \leq 1$, and $\log^\ast(n) := 1 + \log^\ast(\log(n))$ otherwise.} and set $\beta_n := n^{\frac{1}{2(1 + \delta)}}f(n)$. Then, we would need a nuisance estimation rate of $\|\wh{Q} - Q^\ast\|_{L^2(P_W)} = o_\P\left(n^{-\frac{1}{4}\frac{2 + \delta}{1 + \delta}}f(n)^{-1/2}\right)$. In our experiments in Section~\ref{sec:experiments}, we take $f(n) = \log\log(n)$, for instance. We note that this same interpretation of necessary regression estimation rates holds in Section~\ref{sec:irregular} below as well.
\end{rmk}

\subsection{Inference Under Semi-parametric Restrictions}
\label{sec:static:param}

In the previous subsection, we considered performing inference on the value of the optimal treatment policy in a setting where we made no structural assumptions on the outcome regressions/Q-functions. As noted in Remark~\ref{rmk:nuisance_rates}, there is a price for smoothing: the learner needed to estimate the regressions at slightly faster than the typical $o_\P(n^{-1/4})$ required in causal inference. While existing works such as \citet{shi2020breaking} and \citet{luedtke2016statistical} also suffer from this problem, this requirement may be unpalatable when one believes the regressions are particularly hard to estimate. 

To allow for more flexible nuisance estimation rates, we now consider a setting where we place semi-parametric restrictions on the underlying outcome regressions, namely through the assumption of linearity of the \textit{blip effects} in a known feature embedding. Importantly, this is strictly weaker than assuming the Q-function $Q^\ast(a, x)$ is itself linear (which was the assumption made in works such as \citet{chakraborty2010inference, goldberg2014comment, laber2014dynamic}). This modeling allows the expected baseline response to some fixed treatment to be arbitrarily complicated. 

\begin{assumption}
\label{ass:static_param}
We assume that the blip effect $\gamma^\ast(a, x) := Q^\ast(a, x) - Q^\ast(a^\ast, x)$ is linear in a known feature embedding of $(a, x)$, i.e.\ that
\[
\gamma^\ast(A, X) = \theta_0^\top \phi(A, X) \textit{ for some (unknown) } \theta_0 \in \R^d
\]
where $a^\ast \in [N]$ is a known baseline treatment (e.g.\ control) and $\phi : [N] \times \calX \rightarrow \R^d$ is known.
\end{assumption}

Recalling the identification formula for $V^\ast$ in terms of blip effects from Section~\ref{sec:back:static}, we see that under Assumption~\ref{ass:static_param} the value of the optimal treatment policy can be written as
\begin{equation}
\label{eq:opt_val_param}
V^\ast = \E\left[\Phi^\ast(Z; \theta_0)\right] \quad \text{where } \Phi^\ast(Z; \theta_0) := \max_\ell\{\theta_0^\top \phi(\ell, X)\} - \theta_0^\top \phi(A, X) + Y.
\end{equation}
If we are able to obtain a $\sqrt{n}$-consistent estimate $\wh{\theta}$ of the structural parameter $\theta_0$, we can follow the softmax smoothing strategy from the previous section to perform inference on $V^\ast$ via a plug-in estimate.
Thus, we start by by identifying $\theta_0$ in terms of a Neyman orthogonal score. 




\begin{prop}
\label{prop:ortho_static_param}
Consider the score $\Psi$ defined as
\[
\Psi(Z; \theta, \mu, r) := \left\{(Y - \mu(X)) - \theta^\top(\phi(A, X) - r(X))\right\}\left(\phi(A, X) - r(X)\right),
\]
where the true nuisances are $\mu^\ast(x) := \E[Y \mid X = x]$ and $r^\ast(x) := \E[\phi(A, X) \mid X = x]$. Then, $\Psi$ is Neyman orthogonal at $(\theta_0, \mu^\ast, r^\ast)$. Further, when $\E[\Cov(\phi(A, X) \mid X)] \succ 0$, $\theta_0$ is the unique solution to the estimating equation $0 = \E[\Psi(Z; \theta_0, \mu^\ast, r^\ast)]$.
\end{prop}

Under standard convergence and regularity assumptions, we can use the score $\Psi$ to perform inference on the structural parameter $\theta_0$. Notably, we will be able to obtain asymptotically linear estimates for $\theta_0$ only assuming the regression $\wh{\mu}$ converges to to $\mu^\ast$ in $L^2$ at $o_\P(n^{-1/4})$ rates. This is uniformly slower than the rates required on the outcome regression in the previous section. The proof of the following theorem can be found in Appendix~\ref{app:treatment}

\begin{theorem}[Normality of Structural Parameter]
\label{thm:normal_static_param}
Let $Z_1, \dots, Z_n$ be i.i.d.\ random variables satisfying Assumptions~\ref{ass:static} and \ref{ass:static_param}, and let $\wh{\mu}_n$ and $\wh{r}_n$ be nuisance estimates for $\mu^\ast$ and $r^\ast$ that are independent of $Z_1, \dots, Z_n$. Suppose the following assumptions hold.
\begin{enumerate}
    \item \textit{(Nuisance Convergence)} We assume the nuisance estimates converge at sufficiently fast rates, in the particular that
    \[
    \|\wh{r}_{n} - r^\ast\|_{L^2(P_X)} = o_\P(n^{-1/4}) \qquad \text{and} \qquad \|\wh{r}_{n} - r^\ast\|_{L^2(P_X)}\times\|\wh{\mu}_n - \mu^\ast\|_{L^2(P_X)} = o_\P(n^{-1/2}).
    \]
    \item \textit{(Boundedness)} There is some universal constant $G > 0$ such that:
    \[
    \max_{k}\|\phi(k, X)\|_{L^\infty(P_X)}, \|\mu^\ast\|_{L^\infty(P_X)}, \|\wh{\mu}_n\|_{L^\infty(P_X)}, \|r^\ast\|_{L^\infty(P_X)}, \|\wh{r}_n\|_{L^\infty(P_X)}, \|Y\|_{L^\infty(P_Z)} \leq G.
    \]
    Additionally, the structural parameter $\theta_0$ satisfies $\|\theta_0\|_2 \leq G$.
    \item \textit{(Invertible Covariance)} The expected conditional covariance of the feature embedding is invertible, i.e.\ 
    \[
    \E[\Cov(\phi(A, X) \mid X)] \succ 0.
    \]
\end{enumerate}
Then, letting $\wh{\theta}_n$ be the solution to the empirical estimating equation $0 := \P_n \Psi(Z; \wh{\theta}_n,\wh{\mu}_n, \wh{r}_n)$, we have the following asymptotic linearity result:
\[
\sqrt{n}(\wh{\theta}_n  - \theta_0) := \frac{1}{\sqrt{n}}\sum_{i =1 }^n \rho_\theta(Z_i)+ o_\P(1),
\]
where $\rho_\theta(z) := \E[\Cov(\phi(A, X) \mid X)]^{-1}\Psi(Z; \theta_0, \mu^\ast, r^\ast)$. Consequently, asymptotic normality also holds: 
\[
\sqrt{n}(\wh{\theta}_n - \theta_0) \Rightarrow \calN\left(0, \Sigma^\ast\right)
\]
whenever $\Sigma^\ast := \Cov(\rho_\theta(Z_i)) \succ 0$.
\end{theorem}

Using Theorem~\ref{thm:normal_static_param}, we can directly obtain asymptotically linear estimates for the value $V^\ast$ as well. In particular, because the estimates for $\theta_0$ provided by the aforementioned theorem are asymptotically linear, we do not need to construct a Neyman orthogonal score for the value of the optimal treatment policy. We can simply plug-in our estimates and appropriately account for the additional introduced variance. This is detailed in the following corollary.

\begin{corollary}[Normality of Policy Value Estimate]
\label{cor:value_param_static}
For each $\beta > 0$, define the score $\Phi^\beta(Z; \theta)$ by
\[
\Phi^\beta(Z; \theta) := \smax^\beta_\ell\theta^\top \phi(\ell, X) - \theta^\top \phi(A, X) + Y.
\]
Let $\beta_n$ be a smoothing parameter and let $Z_1, \dots, Z_n$ and $\wh{\theta}_n$ be as in Theorem~\ref{thm:normal_static_param}.
Suppose the following conditions hold.
\begin{enumerate}
    \item \textit{(Controlled Density Near Zero)} For each $k \in [N]$, the random variable $\Delta_k := \max_\ell \theta_0^\top \phi(\ell, X)  - \theta_0^\top \phi(k, X)$ satisfies Assumption~\ref{ass:margin} with constants $c, H, \delta > 0$.
    \item \textit{(Growing Smoothing Parameters)} The smoothing parameter $\beta_n$ satisfies $\beta_n = \omega\left(n^{\frac{1}{2(1 + \delta)}}\right)$ and $\beta_n = o(n^{1/2})$.
    \item \textit{(Asymptotic Linearity of Structural Parameter)} The nuisance  $\wh{\theta}_n$ is asymptotically linear about $\theta_0$ with influence function $\rho_\theta(z)$, i.e.
    \[
    \sqrt{n}(\wh{\theta}_n - \theta_0) = \frac{1}{\sqrt{n}}\sum_{i = 1}^n \rho_\theta(Z_i) + o_\P(1).
    \]
    \item \textit{(Boundedness)} There is some universal constant $G > 0$ such that:
    \[
    \|Y\|_{L^\infty(P_Z)},, \|\wh{\theta}_n\|_2, \|\theta_0\|_2 \leq G.
    \]
\end{enumerate}
Then, the estimate $\wh{V}_n := \P_n \Phi^{\beta_n}(Z; \wh{\theta}_n)$ is asymptotically linear:
\[
\sqrt{n}(\wh{V}_n - V^\ast) = \frac{1}{\sqrt{n}}\sum_{i = 1}^n \rho_V(Z_i) + o_\P(1).
\]
Here, $\rho_V(z) := \Phi^\ast(z; \theta_0) - V^\ast + \E[\phi^\infty(X) - \phi(A, X)]^\top \rho_\theta(z)$ and $\phi^\infty(x) := \sum_{k = 1}^N \fraks_k\big(\theta_0^\top \phi(k, X)\big)\phi(k, X)$. This further implies $\sqrt{n}(\wh{V}_n - V^\ast) \Rightarrow \calN(0, \Sigma^\ast)$ whenever $\Sigma^\ast = \Var[\rho_V(Z)] > 0$.
\end{corollary}


We prove Corollary~\ref{cor:value_param_static} in Appendix~\ref{app:treatment} by checking the conditions of Theorem~\ref{thm:smooth_clt_generic}, our smoothed central limit theorem proved in Appendix~\ref{app:clt}. Further, we extend our estimator of the value of the optimal treatment policy under semi-parametric assumptions on the blip effect to dynamic settings in Appendix~\ref{app:dynamic}.

\section{Softmax Smoothing for Irregular Functionals}
\label{sec:irregular}

Up to this point, we have focused on performing inference on the value of an optimal treatment policy. We now show that softmax smoothing can be applied more broadly to classes of irregular functionals that are specified as the expected point-wise maximum of a collection of sufficiently regular scores. This framework will not only subsume the results in the previous section, but will also be applicable in many other settings as well. Throughout this section, we assume that observations take the generic form $Z$ and that we have $X \subset W \subset Z$ for a generalized set of covariates $X$ and an extended set of features $W$. Often, we will have $W = (X, A)$. Here, by  $X \subset W \subset Z$, we mean that $X$ consists of a subset of the coordinates of a vector $W$, and likewise $W$ consists of a sbuset of the coordinates of $Z$. We let $\calX$, $\calW$, and $\calZ$ denote the measurable spaces in which $X$, $W$, and $Z$ take their values. The goal is to perform inference on the maximum of the scores: 
\begin{equation}
\label{eq:gen_score}
V^\ast := \E\left[\max_{\ell}\psi_\ell(X; g_\ell^\ast)\right]
\end{equation}
Here, for each $k \in [N]$, $\psi_k$ is a score depending on the data vector $X$ and a nuisance component $g_k : \calW \rightarrow \R^{d_k}$, and $g_k^\ast(w) \in L^2(P_{W}; \R^{d_k})$ denotes some true, unknown nuisance function. We let $\psi(X; g) := (\psi_1(X; g_1), \dots, \psi_N(X; g_N))$ for convenience throughout, where $g = (g_1, \dots, g_N)$.
In the example of the value of an optimal treatment policy, for each $k \in [N]$, we have $W = (A, X)$, $g_k^\ast(a, x) = Q^\ast(a, x)= \E[Y \mid X =x, A = a]$, and $\psi_k(x; g_k) = g_k(k, x)$. We describe how to instantiate this framework with respect to two other examples, conditional Balke and Pearl bounds~\citep{levis2023covariate} and $L^1$ calibration error \citep{gupta2022post}, in the sequel. Throughout this section, we make the following assumption on the constituent scores $\psi_1(x; g_1), \dots, \psi_N(x; g_N)$.

\begin{assumption}
\label{ass:constituent_scores}
For any $k \in [N]$, we assume that the score $\psi_k(x; g_k)$ satisfies the following.
\begin{enumerate}
    \item There is a $Z$-measurable, square-integrable random variable $U_k \in \R^{d_k}$ such that $g_k^\ast(w) = \E[U_k \mid W = w]$.
    \item The score $\psi_k$ is affine in $g_k$. Further, for any $\nu \in L^2(P_W; \R^{d_k})$, defining $\ell_{X, \nu}^k : \R \rightarrow \R$ by $\ell_{X, \nu}^k(\epsilon) := \psi_k(X; g_k^\ast + \epsilon \nu)$, there is $F_\nu^k : \calX \rightarrow \R_{\geq 0}$ such that
    \[
    \sup_{|\epsilon| \leq 1}\left|\partial_\epsilon\ell_{\nu, X}^k(\epsilon)\right| \leq F_\nu^k(X)\;\;\text{a.s., and } F_{\nu}^k(X) \in L^1(P_X).
    \]
    \item The first Gateaux derivative of $\psi_k$ admits a conditional Riesz representer, i.e.\ there is a square-integrable function $\zeta_k^\ast : \calW \rightarrow \R^{d_k}$ such that, for all $\omega_k \in L^2(P_{W}; \R^{d_k})$,
    \[
    D_{g_k}\psi_k(X; g_k^\ast)(\omega_k) = \E[\zeta_k^\ast(W)^\top \omega_k(W)\mid X] \;\;\text{almost surely}
    \]

\end{enumerate}
\end{assumption}

In words, the first assumption just states that each $g_k^\ast$ is a regression of some random variables $U_k$ onto the extended set of covariates $W$. As noted above, this assumption is trivially satisfied for the value of the optimal treatment policy. The first part of the second assumption says that each score is affine in the nuisance component $g_k$. Given $\psi_k(x; Q) = Q(k, x)$, this is trivially satisfied for the value of the optimal treatment policy. The second part ensures that each constituent score is sufficiently regular to exchange derivatives with expectations. In particular, it will hold whenever 
\begin{equation}
\label{eq:affine-score}
\psi_k(X; g_k) = c_{k, 0} + \sum_{p = 1}^{P_k}c_{k, p}g_{k, i_p}(a_p, X) 
\end{equation}
for an integer $P_k$, constants $c_{k, 0}, \dots, c_{k, P_k} \in \R$, indices $i_1, \dots, i_p \in [d_k]$, and values $a_1, \dots, a_{P_k} \in \calA$. This is the case for the optimal policy value and all examples in the sequel. Finally, the third assumption states that the scores are sufficiently regular to admit a Riesz representer conditional on $X$. We have $\zeta_k^\ast(a, x) = \frac{\mathbbm{1}\{a = k\}}{p^\ast(k \mid x)}$ for the optimal policy value. We derive the representers for conditional Balke and Pearl bounds and the $L^1$ calibration error in the sequel. 

While estimands such as the one outlined in Equation~\eqref{eq:gen_score} have been considered by \citet{semenova2023aggregated}, the existing results have several key limitations. First, all convergence rates of nuisance estimates must be $o(n^{-1/4})$ in the $L^\infty(P_Z)$ norm, which can be a strong requirement when ML learners are used.
Second, the argument maximizer is assumed to be \textit{unique almost surely}, e.g.\ in the case covered in the previous sections the probability of non-response to treatment is assumed to be zero. In this section, we use softmax smoothing to construct confidence intervals (a) only assuming in probability rates of convergence for nuisance estimates and (b) allowing for an arbitrarily large probability of the argument maximizer being non-unique.  
Once again, the first step in our recipe is to construct a softmax approximation to the main objective. To this end, we define
\[
V^\beta := \E\left[\smax^\beta_\ell \psi_\ell(X; g_\ell^\ast)\right].
\]
As before, the quantity $\E[\smax^\beta_\ell \psi_\ell(X; g_\ell)]$ is, without modification, highly sensitive to misestimation in the nuisance components $g_1, \dots, g_N$. Thus, we should subtract a first order correction/de-biasing term to arrive at a Neyman orthogonal score. The following proposition accomplishes precisely this --- we provide a proof in Appendix~\ref{app:irregular}.
\begin{prop}
\label{prop:general_score}
For any $\beta > 0$, consider the score 
\begin{equation}
\label{eq:smooth_score_irregular}
\Psi^\beta(Z; g, \alpha) := \smax^\beta_\ell \psi_\ell(X; g_\ell) + \sum_{\ell = 1}^N \alpha_\ell(W)^\top(U_\ell - g_\ell(W))
\end{equation}
where the $k$th additional nuisance $\alpha_k^\beta : \calW \rightarrow \R^{d_k}$ is given by
\[
\alpha_k^\beta(w) := \left(\partial_k \smax^\beta_\ell\psi_\ell(x; g_\ell^\ast)\right)  \zeta_k^\ast(w),
\]
and $\partial_k\smax_\ell^\beta\{u\}$ is as given in Equation~\eqref{eq:softmax_partial}. Then, $\Psi^\beta$ is Neyman orthogonal, and we have $V^\beta = \E\left[\Psi^\beta(Z; g^\ast, \alpha^\beta)\right]$.


\end{prop}

Using the above Neyman orthogonal score, we can arrive at a general asymptotic linearity/normality result that allows one to construct confidence intervals for broad classes of irregular statistical parameters. Of particular importance will be the following, limiting score
\[
\Psi^\ast(Z; g, \alpha) = \max_\ell\{\psi_\ell(X; g_\ell)\}  + \sum_{\ell = 1}^N \alpha_\ell(W)^\top(U_\ell - g_\ell(W)),
\]
where the limiting $k$th nuisance $\alpha^\ast_k$ is given as the point-wise limit $\alpha^\ast_k(W) := \lim_{\beta \rightarrow \infty}\alpha^\beta_k(W)$. This limit is precisely the quantity
\[
\alpha^\ast_k(W) = \fraks_k(\psi_\ell(X; g_\ell^\ast))\zeta^\ast_k(W) 
\]
where, as before, $\fraks_k(u) := \frac{\mathbbm{1}\{k \in \arg\max_\ell\{u_1, \dots, u_N\}\}}{|\arg\max_\ell\{u_1, \dots, u_N\}|}$. This score is reminiscent of the influence function discussed following  Theorem~\ref{thm:normal_static} (up to the centering at $V^\ast$), just generalized to a broader class of irregular parameters. In particular, the same asymptotic interpretation holds for $\Psi^\ast$  --- in the limit of large $\beta$, the estimator breaks ties via uniform randomization. We now state our main result, which we prove in Appendix~\ref{app:irregular}

\begin{theorem}
\label{thm:normal_irregular}
Let $\beta_n > 0$ be a smoothing parameter and define $d := d_1 + \cdots + d_N$. Let $Z_1, \dots, Z_n$ be i.i.d.\ random variables, and let $\wh{g}_n, \wh{\alpha}_n : \calW \rightarrow \R^{d}$  be estimates for $g^\ast := (g_\ell^\ast : \ell \in [N])$ and $\alpha^{\beta_n} := (\alpha_\ell^{\beta_n} : \ell \in [N])$ that are independent of $Z_1, \dots, Z_n$. Suppose the following conditions hold.
\begin{enumerate}
    \item \textit{(Controlled Density Near Zero)} Let $M := \max_\ell \psi_\ell(X; g_\ell^\ast)$ and for each $k$ let $\Delta_k := M - \psi_k(X; g_k^\ast)$. We assume for each $k$ that $\Delta_k$ satisfies Assumption~\ref{ass:margin} with constants $c, H, \delta > 0$.
    \item \textit{(Growing Smoothing Parameters)} The smoothing parameter $\beta_n$ satisfies $\beta_n = \omega\left(n^{\frac{1}{2(1 + \delta)}}\right)$.
    \item \textit{(Score Continuity)} Each constituent score $\psi_k(Z; g_k)$ is Lipschitz continuous in $g_k$, i.e.
    \[
    \|\psi_k(X; g_k) - \psi_k(X; g'_k)\|_{L^2(P_X)} \lesssim \|g_k - g_k'\|_{L^2(P_{W})}.
    \]
    \item \textit{(Nuisance Rates)} We have for each $k \in [N]$ that 
    \[
    \|\wh{g}_{n, k} - g^\ast_k\|_{L^2(P_W)} = o_\P(n^{-1/4}\beta_n^{-1/2}) \quad \text{and} \quad \|\wh{\alpha}_{n, k} - \alpha^{\beta_n}_k\|_{L^2(P_W)}\times\|\wh{g}_{n, k} - g^\ast_k\|_{L^2(P_W)} = o_\P(n^{-1/2}).
    \]
    Further, we assume nuisance consistency in $L^2$, i.e.\ $\|\wh{g}_n - g^\ast\|_{L^2(P_W)} = o_\P(1)$ and $\|\wh{\alpha}_n - \alpha^{\beta_n}\|_{L^2(P_W)} = o_\P(1)$.\footnote{Here, we are viewing $g^\ast$ and $\wh{g}_n$ are functions in $L^2(P_W; \R^d)$.}
    \item \textit{(Boundedness)} There is a constant $G > 0$ such that
    \[
    \|g^\ast\|_{L^\infty(P_W)}, \|\wh{g}_n\|_{L^\infty(P_W)}, \|\alpha^{\beta_n}\|_{L^\infty(P_W)}, \|\wh{\alpha}_n\|_{L^\infty(P_W)}, \|U_k\|_{L^\infty(P_Z)} \leq G.
    \]
    \end{enumerate}
Then, the estimate $\wh{V}_n := \P_n \Psi^{\beta_n}(Z; \wh{g}_n, \wh{\alpha}_n)$ is asymptotically linear, i.e.
\[
\sqrt{n}(\wh{V}_n - V^\ast) = \frac{1}{\sqrt{n}}\sum_{i = 1}^n\rho_V(Z_i) + o_\P(1),
\]
where $\rho_V(Z) := \Psi^\ast(Z; g^\ast, \alpha^\ast) - V^\ast$ and $\Psi^\ast$ and $\alpha^\ast$ are as defined above. Consequently, asymptotic normality also holds
\[
\sqrt{n}(\wh{V}_n - V^\ast)\Rightarrow \calN(0, \Sigma^\ast),
\]
whenever $\Sigma^\ast := \Var[\rho_V(Z)] > 0$.

\end{theorem}

Most of the assumptions above are commonly made in the causal inference literature. The most peculiar of the assumptions is that of score continuity (Assumption 3), as the quantity on the left-hand side only involves the random covariates $X$, whereas the quantity of the right involves the extended set of features $W$. In the context of Equation~\eqref{eq:affine-score}, this assumption ensures that each value $a_{k, p} \in \calA$ in  the score $\psi_k$ occurs with conditional probability bounded away from zero given $X$. In the example of the value of the optimal treatment policy, we have $\psi_k(X; Q) = Q(k, X)$, and so mean-squared continuity is implied by the assumption of strong positivity. A similar positivity assumption will ensure mean-squared continuity in the example of conditional Balke and Pearl bounds discussed below. For the other example, that of $L^1$ calibration error, we will have $W = X$, and so the assumption will follow trivially. 

In the above, we have chosen to assume the $\psi_k(X; g_k)$ are affine for ease of exposition. In principle, one could generalize our result to non-affine scores that are sufficiently regular. 
More specifically, one would need to assume the $\ell_{\nu, x}^k(\epsilon)$ as defined above are twice continuously differentiable and that there is another integrable function $G_\nu^{k}(X)$ such that $\sup_{|\epsilon|\leq 1}\|\partial_\epsilon^2 \ell^k_{\nu, X}(\epsilon)\|_2 \leq G_{\nu}^k(X)$. The first assumption allows for twice Gateaux differentiability, and the second envelope enables the exchange of second derivatives with expectations. One would also need to assume conditions analogous to mean-squared continuity for the first and second Gateaux derivatives of $\psi_k(X; g)$. Given that the examples discussed in the paper all fall into the setting of affine scores, we avoid introducing additional, technical assumptions.

As before, the above theorem yields a corollary that allows for the practical construction of asymptotically valid confidence intervals using a plug-in estimate of the variance. 
\begin{corollary}
\label{cor:normal_irregular}
Under the conditions of Theorem~\ref{thm:normal_irregular}, letting $\wh{\Sigma}_n := \frac{1}{n}\sum_{m = 1}^n\left\{ \Psi^{\beta_n}(Z_m; \wh{g}_n, \wh{\alpha}_n) - \wh{V}_n\right\}^2$, we have 
\[
\sqrt{n}\wh{\Sigma}_n^{-1/2}(\wh{V}_n - V^\ast) \Rightarrow \calN(0, 1).
\]
Consequently, for any $\delta \in (0, 1)$, the set
\[
C_{1-\delta} := \left[\wh{V}_n - \frac{\wh{\Sigma}_n^{1/2}}{\sqrt{n}}z_{1 -\delta/2}, \wh{V}_n + \frac{\wh{\Sigma}_n^{1/2}}{\sqrt{n}}z_{1 -\delta/2}\right]
\]
forms an asymptotic $1 - \delta$ confidence interval for $V^\ast$, where again $z_{1 -\delta/2}$ is the $(1 - \delta/2)$th quantile of the standard normal distribution.
\end{corollary}

\begin{rmk}
\label{rmk:cross-fit}
In the statements of both of our main theorems (Theorem~\ref{thm:normal_static} and Theorem~\ref{thm:normal_irregular}), we have assumed that the nuisance estimates are constructed on an independent sample of data. In practice, one will typically want to use cross-fitting \citep{chernozhukov2018double} to produce the estimate $\wh{V}_n$, as this will make the most efficient use of data. In short, cross-fitting works as follows. Let $K$ be  a fixed number of folds, which we assume for simplicity divides $n$.\footnote{If the sample size is not divisible by $k$, one can still guarantee the size of all folds are within one sample of each other, which will not impact asymptotic analysis.} 
\begin{enumerate}
    \item Given a sample $Z_1, \dots, Z_n$, splits the indices $[n]$ into $K$ disjoint sets (or ``folds'') $\calI_1, \dots, \calI_K$ of size $n/K$. 
    \item Letting $\calI_{-k} := [n]\setminus \calI_k$, we fit nuisance estimates $\wh{g}^{(-k)}_n, \wh{\alpha}^{(-k)}_n$ of $g^\ast$ and $\alpha^{\beta_n}$ using ML methods on samples with indices in $\calI_{-k}$.
    \item We return the estimate $\wh{V}_n$ given by
    \[
    \wh{V}_n := \frac{1}{n}\sum_{k = 1}^K\sum_{i \in \calI_k}\Psi^{\beta_n}\left(Z_i; \wh{g}^{(-k)}_n, \wh{\alpha}^{(-k)}_n\right).
    \]
\end{enumerate}
If, for each $k \in [K]$, the nuisance estimates $\wh{g}^{(-k)}_n$ and $\wh{\alpha}^{(-k)}_n$ satisfy the assumptions of Theorem~\ref{thm:normal_irregular}, then it is simple to show that
\[
\sqrt{n}(\wh{V}_n - V^\ast) = \frac{1}{\sqrt{n}}\sum_{m = 1}^n \rho_V(Z_m) + o_\P(1).
\]
Further, one will have
\[
\sqrt{n}\wh{\Sigma}_n^{-1/2}(\wh{V}_n - V^\ast) \Rightarrow \calN(0, 1) \quad \text{where} \quad \wh{\Sigma}_n := \frac{1}{n}\sum_{k = 1}^K \sum_{i \in \calI_k}\left\{\Psi^{\beta_n}(Z_i; \wh{g}^{(-k)}_n, \wh{\alpha}^{(-k)}_n) - \wh{V}_n\right\}^2,
\]
a fact with which one can construct confidence intervals in the same manner as outlined in Corollary~\ref{cor:normal_irregular}. We provide a brief proof of this in Appendix~\ref{app:irregular} in Corollary~\ref{cor:cross_fit}.

\end{rmk}

\subsection{An Application to Conditional Balke and Pearl Bounds}
\label{sec:irregular:balke}
Throughout this work, we have focused on $V^\ast := \E[\max_\ell\{Q^\ast(\ell,X)\}]$ as an example of an irregular parameter. We now discuss another example focused on the partial identification of the average treatment effect (ATE) in instrumental variable settings. We largely follow the exposition given in \citet{levis2023covariate}, and recommend \citet{balke1994counterfactual, balke1995probabilistic, balke1997bounds} for further exposition and history. In general, one can consider both upper and lower bounds on the average treatment effect --- we solely discuss lower bounds as results for upper bounds are exactly analogous. The main contributions of this application in comparison to the results of \citet{levis2023covariate} are that (a) our main theorem allows one to select $(\beta_n)_{n \geq 1}$  to guarantee asymptotic normality around the un-smoothed, population Balke and Pearl bounds (defined below) as opposed to smoothed analogues, and (b) the theorem holds even when the argument maximizer used in defining the Balke and Pearl bounds is non-unique with arbitrarily large probability. 

For this example, observations are of the form $Z = (X, D, V, Y)$, $X \in \calX$ represent covariates, $D \in \{0, 1\}$ represents a binary treatment, $V \in \{0, 1\}$ a binary instrument, and $Y \in \{0, 1\}$ a binary outcome. For simplicity, we let $W = (X, V)$ denote the tuple consisting of covariates and instrument. We make the following assumptions on the vector of observations $Z$.

\begin{assumption}
\label{ass:balke}
We assume the following hold for the observations $Z = (X, D, V, Y)$ outlined above.
\begin{enumerate}
    \item \textit{(Consistency)} There exist potential outcomes $D(v)$ and $Y(d, v)$ for $d, v \in \{0, 1\}$ such that $D = D(V)$ and $Y = Y(D, V)$.
    \item \textit{(Exclusion Restriction)} The potential outcomes only depend on $v$ through $d$, i.e.\ we have $Y(0, d) = Y(1, d)$. 
    \item \textit{(Conditional Ignorability)} The instrument is conditionally independent of the potential outcomes for treatment and outcome given covariates, i.e.\ $V \independent D(v), Y(v, D(v)) \mid X$.
    \item \textit{(Strong Positivity)} Let $p^\ast(v \mid x) := \P(V = v \mid X = x)$. We assume $\eta \leq p^\ast(1\mid X) \leq 1 - \eta$ almost surely for some $\eta \in (0, 1/2]$.
\end{enumerate}

\end{assumption}

Because of the second assumption, it makes sense to define the natural potential outcomes $Y(d) := Y(0, d) = Y(1, d)$. The goal is to provide a lower bound on the average treatment effect $\E[Y(1) - Y(0)]$, which in general is not point identified. We use a the construction of \citet{levis2023covariate} to fit this example into the framework of general irregular functionals. 

For any $d, y \in \{0, 1\}$, define the nuisance function $q_{yd}^\ast(x, v) := \P(Y = y, D = d \mid X = x, V = v)$. Next, define the scores $\psi_1(x; q), \dots, \psi_8(x; q)$ as follows:
\begin{align*}
\psi_1(x; q) &= q_{11}(x, 1) + q_{00}(x, 0) - 1 & \psi_2(x; q) &= q_{11}(x, 0) + q_{00}(x, 1) - 1 \\
\psi_3(x; q) &= -q_{01}(x, 1) - q_{10}(x, 1) & \psi_4(x; q) &= -q_{01}(x, 0) - q_{10}(x, 0) \\
\psi_5(x; q) &= q_{11}(x, 0) - q_{11}(x, 1) - q_{10}(x, 1) - q_{01}(x, 0) - q_{10}(x, 0) \\
\psi_6(x; q) &= q_{11}(x, 1) - q_{11}(x, 0) - q_{10}(x, 0) - q_{01}(x, 1) - q_{10}(x, 1)\\
\psi_7(x; q) &= q_{00}(x, 1) - q_{01}(x, 1) - q_{10}(x, 1) - q_{01}(x, 0) - q_{00}(x, 0) \\
\psi_8(x; q) &= q_{00}(x, 0) - q_{01}(x, 0) - q_{10}(x, 0) - q_{01}(x, 1) - q_{00}(x, 1).
\end{align*}
In the above, we write $q := (q_{yd} : (y, d) \in \{0, 1\}^2)$ and let $q^\ast$ denote the full-vector of true, unknown regressions. While each score only depends on a subset of the regressions, we let the scores depend on the full vector for simplicity.
\citet{levis2023covariate} show that one has the following lower bound on the ATE:
\[
V^\ast := \E[\max_{\ell = 1}^8 \psi_\ell(X; q^\ast)] \leq \E[Y(1) - Y(0)].
\]
Our goal is to show that our smoothing approach can be used to perform inference on the lower Balke and Pearl bound $V^\ast$. In the following corollary, for each $\ell \in [8]$, we leverage the representer $\zeta^\ast_\ell(V,X) \in \R^4$ whose coordinates are defined as
\begin{equation}
\label{eq:balke-pearl-representer}
\left(\zeta_\ell^\ast(V, X)\right)_{y, d} := \sgn_\ell(y, d, 0) \frac{\1\{V = 0\}}{p^\ast(0 \mid X)} + \sgn_\ell(y, d, 1)\frac{\1\{V = 1\}}{p^\ast(1 \mid X)},
\end{equation}
where $\sgn_\ell(y, d, v) = 1$ if $q_{yd}(x, v)$ appears in the definition of $\psi_\ell$ with positive sign, $\sgn_\ell(y, d, v) = -1$ if it appears with negative sign, and $\sgn_\ell(y, d, v) =0$ otherwise.
\begin{corollary}
\label{cor:balke-and-pearl}
Let $Z_1, \dots, Z_n$ be i.i.d.\ random variables satisfying Assumption~\ref{ass:balke}. Suppose the random variables $\Delta_k := \max_{\ell = 1}^8 \psi_\ell(X; q^\ast) - \psi_k(X; q^\ast)$ satisfy Condition 1 of Theorem~\ref{thm:normal_irregular} with constant $\delta > 0$ and that $\beta_n$ satisfies Condition 2 of the same theorem. For any $k \in [8]$, define the smoothed nuisance function $\alpha^\beta_k$ by
\[
\alpha^\beta_k(X, V) := \left(\partial_k \smax^\beta_\ell\psi_\ell(X; q^\ast)\right)  \zeta_k^\ast(V, X),
\]
where $\zeta_k^\ast$ is as outlined in Equation~\eqref{eq:balke-pearl-representer}. Further, for $\ell \in [8]$, let  $\wh{q}_{n}$ and $\wh{\alpha}_{n, \ell}$ be nuisance estimates of $q^\ast$ and $\alpha^{\beta_n}_\ell$ that are independent of $Z_1, \dots, Z_n$ and satisfy Conditions 4 and 5 of Theorem~\ref{thm:normal_irregular}. Then, letting $\Psi^\beta$ be as defined in Equation~\eqref{eq:smooth_score_irregular} and $\wh{V}_n := \P_n \Psi^{\beta_n}(Z; \wh{q}_n, \wh{\alpha}_n)$, we have
\[
\sqrt{n}(\wh{V}_n - V^\ast) \Rightarrow \calN(0, \Sigma^\ast),
\]
when $\Sigma^\ast := \Var[\rho_V(Z)] > 0$. Here, $\rho_V(Z) = \Psi^\ast(Z; q^\ast, \alpha^\ast) - V^\ast$ and the limiting score $\Psi^\ast$ is given by
\[
\Psi^\ast(Z; q^\ast, \alpha^\ast) = \max_\ell \psi_\ell(X; q^\ast) + \sum_{\ell = 1}^8 \fraks_\ell(\psi_k(X; q_k^\ast))\cdot\left(\sum_{(y, d, v)}\sgn_\ell(y, d, v)\frac{\mathbbm{1}\{V = v\}}{p^\ast(v \mid X)}\left\{U_{yd} - q_{yd}(X, V)\right\}\right).
\]
where $U_{yd} := \1\{D = d, Y = y\}$.
Further, per the results of \citet{levis2023covariate}, when $\arg\max_k \psi_k(X; q^\ast)$ is a singleton almost surely, $\rho_V(Z)$ is precisely the efficient influence function.
\end{corollary}

We prove Corollary~\ref{cor:balke-and-pearl} by first showing that, under Assumption~\ref{ass:balke}, the constituent scores satisfy Assumption~\ref{ass:constituent_scores}. We then check the relevant conditions of Theorem~\ref{thm:normal_irregular}. 
We provide a proof in Appendix~\ref{app:irregular:app}.

\subsection{An Application to Estimating Calibration Error}

We now consider an application of Theorem~\ref{thm:normal_irregular} related to the calibration literature. Calibration is a notion of consistency for an estimator $\theta$ that establishes the trustworthiness of its predictions~\citep{kumar2019verified, gupta2022post}. Formally, perfect calibration is defined as follows. 
\begin{definition}
\label{def:calibration}
An estimator $\theta : \calO \rightarrow \R$ of a regression $\theta_0(o) := \E[Y \mid O = o]$ is said to be perfectly calibrated if 
\[
\E[Y \mid \theta(O)] = \theta(O) \text{ almost surely}.
\]
\end{definition}
We use $O$ in this section to represent observed covariates to prevent notational collision with $X$ as used in Theorem~\ref{thm:normal_irregular} and Assumption~\ref{ass:constituent_scores}. 
In words, $\theta$ is perfectly calibrated if, on average, whenever it makes a prediction $\theta(O) = x$, the true outcome is also $x$. Given most estimators will not be perfectly calibrated in practice, various notions of approximate calibration have been developed in the calibration literature. Perhaps the most commonly used approximation is the $L^1$ calibration error, which aims to capture the ``distance'' from perfect calibration by measuring the absolute deviations from perfect calibration on each level set of $\theta$~\citep{naeini2015obtaining, gupta2022post}. 
Other analogous notions, such as $L^2$ calibration error \citep{kumar2019verified, whitehouse2024orthogonal, lee2023t, van2023causal}, have been proposed in the literature as well, but we do not define those in this paper.
The $L^1$ calibration error is defined as follows.

\begin{definition}
\label{def:cal_error}
The $L^1$ calibration error of an estimator $\theta : \calO \rightarrow \R$ of some regression $\theta_0(o) := \E[Y \mid O = o]$ is defined as 
\[
\Cal_1(\theta) := \E\left|\theta(O) - \chi^\ast(\theta(O))\right|,
\]
where $\chi^\ast(x) := \E[Y \mid \theta(O) = x]$ for any $x \in \range(\theta)$.
\end{definition}

The goal of this subsection is to use Theorem~\ref{thm:normal_irregular} to construct asymptotically-valid confidence intervals for the $L^1$ calibration error of a fixed (i.e.\ non-random) model $\theta : \calO \rightarrow \R$. Estimating the $L^1$ calibration error is understood to be a hard problem in the calibration literature, and \citet{gupta2022post} note that it is impossible to construct unbiased estimates in general. It is also non-trivial to construct asymptotically-normal estimate for $\Cal_1(\theta)$. First, 
the calibration function $\chi^\ast$, which is the projection of $Y$ onto the space of $\theta(O)$-measurable functions, is an unknown nuisance component and thus must be estimated from data. Second, the $L^1$ calibration error is defined as the expected absolute deviation between $\theta$ and $\chi^\ast$, and is thus not differentiable when both are equal to zero. This could happen when the estimator is perfectly calibrated, or when it is calibrated on certain subsegments of the population. In sum, performing inference on $\Cal_1(\theta)$ is a semi-parametric problem that falls into the scope of Theorem~\ref{thm:normal_irregular}. To prove our result, we just make the following  assumptions.

\begin{assumption}
\label{ass:calibration}
We assume the learner observes i.i.d.\ samples of the form $Z = (O, Y)$ as outlined above and that there is an absolute constant $G > 0$ such that $|Y| \leq G$ and $|\theta(O)| \leq G$ almost surely. 
\end{assumption}

With the above assumption, we can now present the main result of this subsection. For this example, we take $X = W := \theta(O)$, and take the outcome to be simply $Y$. The constituent scores are $\psi_1(x; \chi), \psi_2(x; \chi)$, defined respectively as
\[
\psi_1(x; \chi) := \chi(x) - x \quad \text{and} \quad \psi_2(x; \chi) := x - \chi(x).
\]
One can check the identity
\[
\Cal_1(\theta) := \E|\theta(O) - \chi^\ast(\theta(O))| = \E\big[\max\{\psi_1(\theta(O); \chi^\ast), \psi_2(\theta(O); \chi^\ast)\}\big].
\]
The following corollary follows in a straightforward manner from Theorem~\ref{thm:normal_irregular}. We prove the following result in Appendix~\ref{app:irregular:app}.
\begin{corollary}
\label{cor:l1-calibration}
Let $Z_1, \dots, Z_n$ be i.i.d.\ random variables satisfying Assumption~\ref{ass:calibration}. Suppose $\Delta := |\theta(O) - \chi^\ast(\theta(O))|$ satisfies the first condition of Theorem~\ref{thm:normal_irregular} with constant $\delta > 0$ and that $\beta_n$  satisfies the second condition of the same theorem. For any $\beta > 0$, we define the representers $\alpha^\beta_1$ and $\alpha^\beta_2$ respectively as 
\[
\alpha^\beta_1(X) := \partial_1\smax^\beta\{\psi_1(X; \chi^\ast), \psi_2(X; \chi^\ast)\} \quad \text{and} \quad \alpha^\beta_2(X) := - \partial_2 \smax^\beta\{\psi_1(X; \chi^\ast), \psi_2(X; \chi^\ast)\},
\]
where again $X := \theta(O)$. Let  $\wh{\chi}_n$ and $\wh{\alpha}_{n}$ be estimates of $\chi^\ast$ and $\alpha^{\beta_n}$ that are independent of $Z_1, \dots, Z_n$ and satisfy Assumptions 4 and 5 of Theorem~\ref{thm:normal_irregular}. Then, letting $\Psi^\beta$ be as defined in Equation~\eqref{eq:smooth_score_irregular} and $\wh{\Cal}_{1, n}(\theta) := \P_n \Psi^{\beta_n}(Z; \wh{\chi}_n, \wh{\alpha}_n)$, we have
\[
\sqrt{n}(\wh{\Cal}_{1, n}(\theta) - \Cal_1(\theta)) \Rightarrow \calN(0, \Sigma^\ast),
\]
when $\Sigma^\ast := \Var[\rho_{\Cal}(Z)] > 0$. Here $\rho_{\Cal}(Z) = \Psi^\ast(Z; \chi^\ast, \alpha^\ast) - \Cal_1(\theta)$, $\Psi^\ast$ is given by
\begin{align*}
\Psi^\ast(Z; \chi^\ast, \alpha^\ast) &= \left|\chi^\ast(\theta(O)) - \theta(O)\right| + \sgn(\chi^\ast(\theta(O)) - \theta(O))\left\{Y - \chi^\ast(\theta(O))\right\} \\
&= \sgn(\chi^\ast(\theta(O)) - \theta(O))\left\{Y - \theta(O)\right\},
\end{align*}
and $\sgn(x) := \1\{x > 0\} - \1\{x < 0\}$. 

\end{corollary}


\section{Validation of Coverage}
\label{sec:experiments}
We now evaluate the coverage of our softmax smoothing methods from Section~\ref{sec:static} in both non-parametric and semi-parametric settings. We also also measure the sensitivity of coverage to the choice of smoothing parameter $\beta$. Analyzing this sensitivity is important for practical applications of our estimator, as the exponent $\delta$ that governs the density of the sub-optimality gaps may not be known to the user.

Throughout this section, we study a setting with a binary treatment for simplicity. To align with the majority of the causal inference literature, we let our action set in this case be denoted by $\{0, 1\}$ and let $\tau^\ast(x) := \E[Y(1) - Y(0) \mid X = x] = Q^\ast(1, x) - Q^\ast(0, x)$ denote the CATE. We start by describing data generating processes in both the non-parametric and semi-parametric settings which ensure Assumption~\ref{ass:static} holds.

\subsection{Data Generating Processes}

\paragraph{Non-Parametric Data Generating Process}

We describe how a single sample $Z := (X, A, Y)$ is generated in the non-parametric setting. We first define the abstract data-generating process, and then specify the parameter values used in our experiments. Let
\begin{align*}
&C_1, C_2 \sim \Unif([0,1]), \\
&W = (W_1,\dots,W_d) \sim \calN(0,I_d), \\
&X := (C_1, C_2, W),
\end{align*}
where all random variables are mutually independent. Fix parameters $\pi_0 \in (0,1)$, $\rho \in (0,1)$, $\delta > 0$, and $t_0 > 0$. In our applications, we take $d = 3, \pi_0 = 0.2, \rho = 0.3$, and $t_0 = 0.3$. Further, define
\[
U := \frac{C_1 - \pi_0}{1-\pi_0}.
\]
We let the baseline regression under control be given by
\[
Q^\ast(0,X) := 1.0 + 0.5W_1,
\]
and define the CATE $\tau^\ast(X)$ by
\begin{align*}
\tau^\ast(X) := \begin{cases}
0 & \text{if } C_1 \leq \pi_0,\\[4pt]
\sgn(C_2 - 0.5)\, t_0\left(\frac{U}{\rho}\right)^{1/\delta} & \text{if } C_1 > \pi_0 \text{ and } U \leq \rho,\\[8pt]
\sgn(C_2 - 0.5)\left(t_0 + h(W_1)\right) & \text{if } C_1 > \pi_0 \text{ and } U > \rho,
\end{cases}
\end{align*}
where $h:\R \to \R_{\geq 0}$ is a non-negative function, which we take to be $h(W_1) := 0.5|W_1|$.  We then define
\[
Q^\ast(1,X) := Q^\ast(0,X) + \tau^\ast(X).
\]
Treatments are sampled in the following manner:
\[
A \mid X \sim \Bern(p^\ast(X)),\;\; \text{where } p^\ast(X) = \frac{1}{1 + \exp\{-0.75(-0.2 + W_1)\}},
\]
and outcomes are generated as
\[
Y := Q^\ast(0,X) + A\cdot \tau^\ast(X) + \sigma\epsilon,
\qquad
\epsilon \sim \calN(0,1),
\]
where $\epsilon$ is independent of $(X,A)$ and we  take $\sigma = 0.3$ in our settings.

First, because $\tau^\ast(X) = 0$ whenever $C_1 \leq \pi_0$, we have $\P(\tau^\ast(X)=0)=\pi_0>0$, and thus there is a positive probability of non-response to treatment. Second, one can check that the random variable $|\tau^\ast(X)|$ admits a Lebesgue density $f_\tau(u)$ on the interval $(0, t_0)$ that satisfies $f_\tau(u) \lesssim u^{\delta - 1}$. That is, the data-generating process satisfies the density condition from Assumption~\ref{ass:margin} with parameter $\delta$.

\paragraph{Semi-Parametric Setting}
For the semi-parametric setting, a single sample $Z := (X, A, Y)$ is generated as follows. 
\begin{align*}
&X = (X_1, X_2, \dots, X_5) \sim \calN(0, I_5)\\
&\epsilon \sim \calN(0, 1)\\
&A \sim \Bern\left(\frac{1.0}{1.0 + \exp\{X_1\}}\right) \\
&Y = X_1 + \theta_0 A (X_2 + 1.0) + \epsilon
\end{align*}
where we set the structural parameter as $\theta_0 = 1.0$. In this setting, there is zero-probability of non-response, and because the blip effect using $A = 0$ as a reference is precisely the CATE, we see that $\tau^\ast(X) = \theta_0(X_2 + 1.0)$ and thus Assumption~\ref{ass:margin} is satisfied with $\delta = 1.0$ (corresponding to a bounded density).

\subsection{Method of Evaluation}

\begin{table}[t]
  \centering
  \begin{tabular}{c|c|c|c|c}
    \toprule
    $\delta^{\ass}$
      & {$\delta^{\true} = 0.75$}
      & {$\delta^{\true} = 1.0$}
      & {$\delta^{\true} = 1.25$}
      & {$\delta^{\true} = 1.0$} \\
    & $n = 5,000$ (NP)
    & $n = 5,000$ (NP)
    & $n = 5,000$ (NP)
    & $n = 5,000$ (SP)\\
    \midrule
    $\delta^{\true} - 0.1$ 
      & $0.932\in(0.894, 0.957)$ 
      & $0.932\in(0.894, 0.957)$ 
      & $0.928\in(0.890, 0.954)$
      & $0.944\in(0.908, 0.966)$ \\
    $\delta^{\true}$ 
      & $0.932\in(0.894, 0.957)$ 
      & $0.932\in(0.894, 0.957)$ 
      & $0.920\in(0.880, 0.948)$
      & $0.944\in(0.908, 0.966)$\\
    $\delta^{\true} + 0.1$ 
    & $0.932\in(0.894, 0.957)$ 
      & $0.928\in(0.890, 0.954)$
      & $0.916\in(0.875, 0.944)$
      & $0.944\in(0.908, 0.966)$ \\
    \bottomrule
  \end{tabular}
  \caption{Empirical coverage (mean and Wilson 95\% CI) for the softmax smoothing estimator with different settings of $\delta^{\true}$ with $n = 5,000$ samples computed over 250 runs. The first three columns (labeled NP) correspond to the non-parametric DGP and estimator. The last column (labeled SP) corresponds to the semi-parametric setting. Each row corresponds to a different setting of $\delta^{\ass}$ relative to $\delta^{\true}$. The temperature $\beta$ is chosen according to Equation~\eqref{eq:beta_trans}.}
  \label{tbl:coverage}
\end{table}

\paragraph{Evaluation of Non-Parametric Estimator}
We start by describing how we measure coverage in a single run of our method. We then discuss how we aggregate our experimental results over independent runs to obtain estimates of coverage with corresponding confidence intervals. 
We generate $n = 5,000$ independent samples from the data generating process, where we let $\delta^{\true} \equiv \delta \in \{0.75, 1.0, 1.25\}$. We then compute the estimator outlined in Theorem~\ref{thm:normal_static} as follows. We use $K$-fold cross-fitting (see Remark~\ref{rmk:cross-fit}) with $K = 5$ and stratified splitting based on the treatment $A$. Using the data outside of each fold $k \in \{1, 2, \dots, 5\}$, we separately train nuisance models $\wh{Q}^{(-k)}(0, X), \wh{Q}^{(-k)}(1, X)$, and $\wh{p}^{(-k)}(X)$ predicting $Q^\ast(0, X), Q^\ast(1, X),$ and $p^\ast(X)$ respectively. In each run, we estimate the Q-functions $Q^\ast(0, X)$ and $Q^\ast(1, X)$ using the LGBMRegressor class from the module LightGBM~\citep{ke2017lightgbm} and we estimate the propensity $p^\ast(X)$ using the LogisticRegression class from the Scikit-learn module~\citep{pedregosa2011scikit} with $\ell_2$ regularization.\footnote{We use the following parameter settings for LGBMRegressor \texttt{
    \{`metric' : `rmse',
    `learning\_rate': 0.05,
    `num\_leaves' : 2**3,
    `n\_estimators' : 250,
    `max\_depth' : 3,
    `verbose' : -1,
    `random\_state' : 123\}
}.  For LogisticRegression, we use the parameter settings \texttt{`random\_state' : 123} and otherwise use default settings.} For each $\delta^{\true}$, we consider a grid of smoothing parameters to measure sensitivity of the coverage of the smoothing estimator. For each assumed parameter 
\[
\delta^{\ass} \in \{\delta^{\true} - 0.5, \delta^{\true} - 0.25, \delta^{\true} - 0.1, \delta^{\true}, \delta^{\true} + 0.1, \delta^{\true} + 0.25, \delta^{\true} + 0.5\}
\] we compute the softmax-smoothing estimator with
\begin{equation}
\label{eq:beta_trans}
\beta_n := 1.5 \cdot n^{\frac{1}{2(1 + \delta^{\ass})}}\log\log(n).
\end{equation}
We produce
corresponding confidence intervals using Corollary~\ref{cor:normal_static}. We multiply by $\log\log(n)$ to ensure $\beta_n = \omega\left(n^{\frac{1}{2(1 + \delta^{\ass})}}\right)$, which is needed when $\delta^{\ass} = \delta^{\true}$ per Theorem~\ref{thm:normal_static} to maintain coverage. For each setting of parameters $(n, \delta^{\true}, \delta^{\ass})$ (which implicitly defines $\beta_n$) outlined above, we repeat the measurement process $M = 250$ times and report 95\% Wilson Binomial confidence intervals (CIs). 

\paragraph{Evaluation of Semi-Parametric Estimator}
In the semi-parametric setting, we estimate the structural parameter $\theta_0 = 1.0$ using the moment condition estimator outlined in Theorem~\ref{thm:normal_static_param}, again using $K$-fold cross-fitting with stratified splitting and $K = 5$. All nuisances are estimated using the LGBMRegressor class with the same hyperparameter settings as above. We then use the estimate of the structural parameter to compute an estimate of the optimal policy value using the estimator in Corollary~\ref{cor:value_param_static}. In this setting, we consider the grid $\delta^{\ass} \in \{0.5, 0.75, 0.9, 1.0, 1.1, 1.25, 1.5\}$ and compute an estimator for each corresponding value of $\beta_n$ as noted in Equation~\eqref{eq:beta_trans}. Again, for each setting of $\delta^{\ass}$, we repeat the measurement process $M = 250$ times and produce Wilson CIs as above.

\subsection{Findings}
\paragraph{Non-Parametric Estimator}
We display a summary of the coverage of our estimator in the first three columns of Table~\ref{tbl:coverage} and more extensive measurements of sensitivity in Figure~\ref{fig:coverage}. Table~\ref{tbl:coverage}, which assume $\delta^{\ass} \in \{\delta^{\true} - 0.1, \delta^{\true}, \delta^{\true} + 0.1\}$, shows our estimator obtains close to 95\% coverage across settings of $\delta^{\true}$ even when $\delta^{\ass}$ is incorrect. Figure~\ref{fig:coverage}, which displays a wider range of $\delta^{\ass}$, seems to indicate that having $\delta^{\ass} \gg \delta^{\true}$ can lead to incorrect coverage. This makes sense in the context of our theoretical results, as when $\delta^{\ass} > \delta^{\true}$ (over-smoothing), the selected $\beta_n$ will satisfy $\beta_n = o\left(n^{\frac{1}{2(1 + \delta^{\true})}}\right)$, and thus our asymptotic normality results will fail to hold.
We note that in settings where $\delta^{\ass} \leq \delta^{\true}$ (under-smoothing) we generally observe valid coverage as well. Overall, our experimental results seem to indicate that it is advantageous to under-smooth (i.e.\ to select $\beta_n \gg n^{\frac{1}{2(1 + \delta)}}$) if there is uncertainty about the behavior of the density of the CATE near zero. This indicates that the bias from smoothing is likely the limiting factor in maintaining valid coverage. While under-smoothing requires the $L^2$ error of the regression estimate to decay as $o_\P(\beta_n^{-1/2}n^{-1/4})$, flexible learners may readily be able to obtain such rates for reasonable DGPs.

\begin{figure}
    \centering
    \begin{subfigure}{0.48\textwidth}
    \centering
    \includegraphics[width=\linewidth]{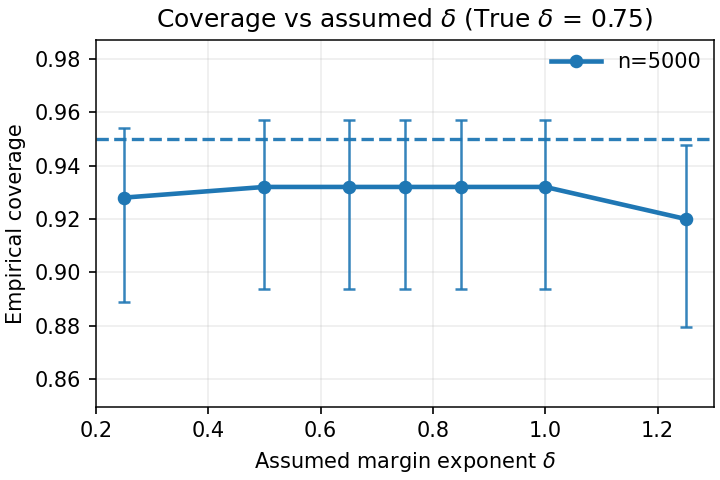}
    \caption{$\delta^{\true} = 0.75$, Non-Parametric}
    \end{subfigure}
    \begin{subfigure}{0.48\textwidth}
    \centering
    \includegraphics[width=\linewidth]{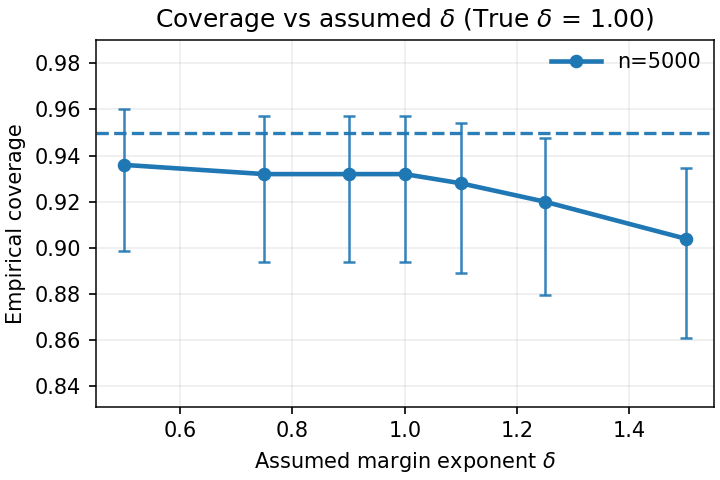}
    \caption{$\delta^{\true} = 1.0$, Non-Parametric}
    \end{subfigure}
    \begin{subfigure}{0.48\textwidth}
        \centering
    \includegraphics[width=\linewidth]{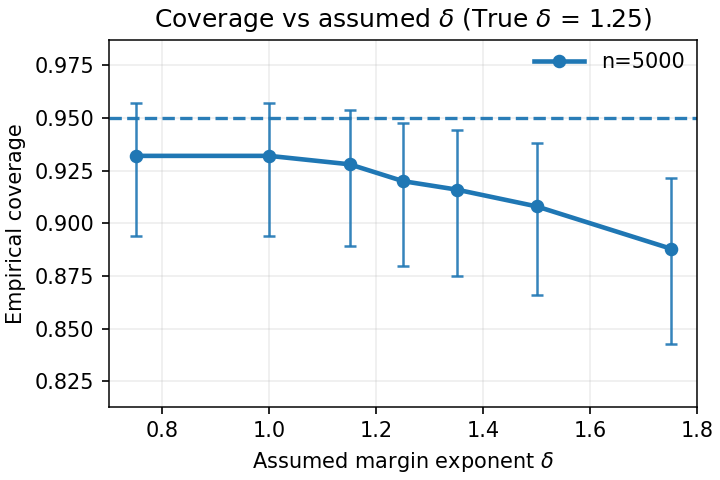}
    \caption{$\delta^{\true} = 1.25$, Non-Parametric}
    \end{subfigure}
    \begin{subfigure}{0.48\textwidth}
    \centering
    \includegraphics[width=\linewidth]{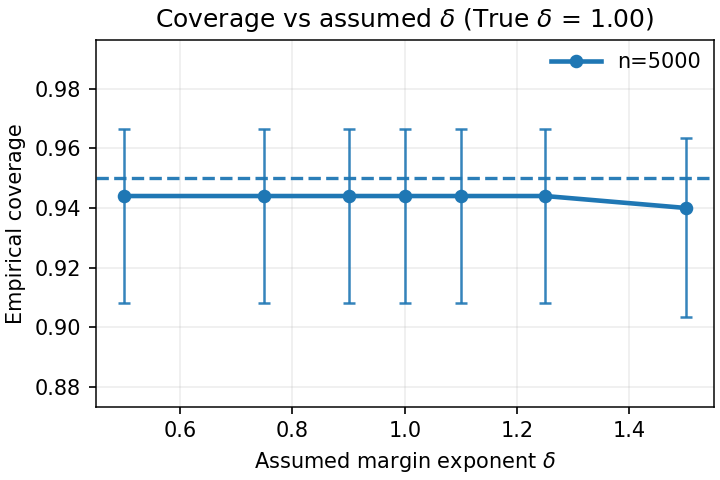}
    \caption{$\delta^{\true} = 1.0$, Semi-Parametric}
    \end{subfigure}
    \caption{We plot the empirical coverage of our softmax smoothing estimator under each of the data-generating processes (DGPs) described above. The first three figures correspond to a setting of $\delta^{\true}$ in the Non-Parametric DGP. The final figure corresponds to the semi-parametric DGP, in which $\delta^{\true} = 1.0$. The $x$-axis shows the various values of $\delta^{\ass}$ assumed by the estimator, which is plugged into Equation~\eqref{eq:beta_trans}. The  $y$-axis denotes coverage, and the dashed horizontal line corresponds to the target coverage of $95\%$. We plot the mean coverage with $n = 5,000$ samples, and include point-wise valid 95\% Wilson confidence intervals.}
    \label{fig:coverage}
\end{figure}

\paragraph{Semi-Parametric Estimator}
In the semi-parametric setting, due to the assumed linearity of the blip effects, there is no real trade-off between first-order biases and second order nuisance errors. That is, the learner just needs to select $\beta_n = o(n^{1/2})$ and $\beta_n = \omega\left(n^{\frac{1}{2(1 + \delta^{\true})}}\right)$ in order to obtain valid asymptotic coverage. Our coverage findings, illustrated in the final column of Table~\ref{tbl:coverage} and in the final image of Figure~\ref{fig:coverage}, indicate a general lack of sensitivity to the choice of the assumed density parameter $\delta^{\ass}$. Interestingly, even under-smoothing does not seem to result in invalid coverage.

\section{Conclusion}

In this paper, we considered the problem of constructing confidence intervals for the value of the optimal treatment policy. This is known to be a difficult problem in semi-parametric inference, as the non-Gateaux differentiability of the target functional prevents the application of standard, de-biased estimators~\citep{luedtke2016statistical, shi2020breaking, laber2014dynamic, chakraborty2013inference, semenova2023aggregated}. 
Our approach is intuitive and computationally lightweight, replacing the non-differentiable maximum function with a smoothed analogue defined in terms of the softmax function. Through a careful analysis of this smoother, we show one can construct estimators for the optimal policy value that allow for an arbitrarily large probability of non-response to a treatment. This estimator naturally generalizes to other important examples of irregular parameters, such as conditional Balke and Pearl bounds on the average treatment effect~\citep{levis2023covariate} and the $L^1$ calibration error of an ML model~\citep{gupta2022post}. 
In sum, our work shows that the careful choice and analysis of a smoothing function can turn irregular problems into ones that can again be handled by first-order de-biasing methods.

While the results presented in this paper are quite general, many important open questions remain. First, throughout our work we assume that the sequence  $(\beta_n)_{n \geq 1}$ is fixed and chosen according to the behavior of the sub-optimality gaps in a small neighborhood of zero. Although our experiments demonstrated some lack of sensitivity to the choice of smoothing parameter, it would be of great theoretical and practical interest to develop data-adaptive methods for selecting $\beta_n$, if this is possible. Second, while we established efficiency in the setting where the optimal treatment is almost surely unique, it would be interesting to reason about efficiency when non-response occurs with positive probability. In this setting, \citet{hirano2012impossibility} show that regular asymptotically linear estimators do not generally exist for such target parameters. Since semi-parametric efficiency is most often described in terms of such estimators, one would need to develop an analogue of efficiency theory for this irregular setting. We view this as an important direction as it would facilitate the direct comparison of estimators in irregular regimes.

\bibliographystyle{plainnat}
\bibliography{bib}

\appendix

\section{Properties of the Softmax Operator}
\label{app:softmax}
We now prove a variety of properties for the softmax operator $\smax^\beta$, which we recall, for a vector $u = (u_1, \dots, u_N) \in \R^N$, is defined as
\[
\smax^\beta_\ell(u) := \frac{\sum_{i = 1}^N u_i \exp\{\beta u_i\}}{\sum_{j = 1}^N \exp\{\beta u_j\}}.
\]

We start by proving a number of basic analytical properties of the softmax function that will be useful throughout our work. In the sequel, for simplicity, we define the $i$th softmax weight $w^\beta_i(u)$ to be
\[
w^\beta_i(u) := \frac{e^{\beta u_i}}{\sum_j e^{\beta u_j}}
\]

\begin{lemma}
\label{lem:softmax}
For any $\beta > 0$, consider the softmax function $\smax^\beta$ as defined above. The following hold.
\begin{enumerate}
    \item The $i$th partial derivative of $\smax^\beta(u)$, which is given by
    \[
    \partial_{u_i}\smax^\beta(u) \equiv (\nabla_u \smax^\beta(u))_i = w^\beta_i(u)\left[1 + \beta(u_i - \smax^\beta_\ell(u))\right],
    \]
    satisfies $\sup_u\|\nabla_u \smax^\beta(u)\|_\infty \leq 1 + \log(N)$.\label{pt:softmax_grad}
    \item $\smax^\beta(u)$ is $L$-Lipschitz with $L := \sqrt{N}(1  + \log(N))$.\label{pt:softmax_lip}
    \item The Hessian of $\smax^\beta(u)$
    satisfies $\|\nabla_u^2 \smax^\beta(u)\|_{op} \leq 6\beta  + 4\beta\log(N)$.\label{pt:softmax_hess}
    \item The softmax function obeys the following limiting properties:
    \[
    \smax^\beta_\ell(u_\ell) \xrightarrow[\beta \rightarrow \infty]{} \max_\ell\{u_\ell\}.
    \]\label{pt:softmax_lim}
    \item The $k$th partial derivative the softmax function obeys the following limiting properties:
    \[
    \partial_{u_k}\smax^\beta_\ell(u_\ell) \xrightarrow[\beta \rightarrow \infty]{} \fraks_k(u),
    \]
    where $\fraks_k(u) := \frac{\mathbbm{1}\{k \in \arg\max_\ell(u_1, \dots, u_N)\}}{|\arg\max_\ell(u_1, \dots, u_N)|}.$\label{pt:softmax_grad_lim}
\end{enumerate}
\end{lemma}

\begin{proof}
We prove each of the above points in succession.
\paragraph{Proof of Point~\ref{pt:softmax_grad}}
We start by proving the first point. It suffices to compute the $k$th partial derivative. Let $f(u) := \sum_i u_i \exp\{\beta u_i\}$ and $g(u) := \sum_j \exp\{\beta u_j\}$. Computing the $k$th partial for each yields:
\begin{align*}
\partial_k f(u) &= e^{\beta u_k} + \beta u_k e^{\beta u_k} \\
\partial_k g(u) &= \beta e^{\beta u_k}.
\end{align*}
Now, we can compute the overall partial derivative via the quotient rule:
\begin{align*}
\partial_k \frac{\sum_i u_ie^{\beta u_i}}{\sum_j e^{\beta u_j}} &= \frac{g(u)\partial_k f(u) - f(u)\partial_k g(u)}{g(u)^2} \\
&= \frac{e^{\beta u_k}(1 + \beta u_k)\sum_j e^{\beta u_j}  - \beta e^{\beta u_k}\sum_i u_i e^{\beta u_i}}{\left(\sum_j e^{\beta u_j}\right)^2} \\
&= \frac{e^{\beta u_k}}{\sum_j e^{\beta u_j}}\left[(1 + \beta u_k) - \beta \smax^\beta(u)\right] \\
&= w^\beta_k(u)\left[1 + \beta[u_k - \smax^\beta(u)]\right],
\end{align*}
which proves the desired formula for the gradient of the softmax. Now, we can bound the magnitude of any individual coordinate of the gradient vector. We consider two cases. First, suppose $u_k \leq \smax^\beta(u)$. Then, we can bound $|w_k(u)\beta(u_k - \smax^\beta(u))|$ as follows. Let $\wt{\Delta}_i := \smax^\beta(u) - u_i$ for $i \in [N]$. We have 
\begin{align*}
|w^\beta_k(u)\beta(u_k - \smax^\beta(u))| &= \beta \wt{\Delta}_k\frac{e^{\beta u_k}}{\sum_j e^{\beta u_j}} \\
&= \beta \wt{\Delta}_k \frac{e^{-\beta \smax^\beta(u)}}{e^{-\beta \smax^\beta(u)}}\frac{e^{\beta u_k}}{\sum_j e^{\beta u_j}} \\
&= \beta \wt{\Delta}_k \frac{e^{-\beta \wt{\Delta}_k}}{\sum_j e^{-\beta \wt{\Delta}_j}} \\
&\leq \sup_{u \geq 0} u e^{-u} = \frac{1}{e}.
\end{align*}
In the above, the final inequality follows from the fact that $\sum_j e^{-\beta \wt{\Delta}_j} \geq  1$, since there is always some index $i^\ast \in [N]$ such that $\wt{\Delta}_{i^\ast} =\smax^\beta_\ell(u) - u_{i^\ast} \leq 0$ and hence $e^{-\beta \wt{\Delta}_{i^\ast}} \geq 1$. The final equality follows from the fact that the function $u \mapsto ue^{-u}$ attains its maximum of $\frac{1}{e}$ when $u = 1$.
Thus, we can bound the $k$th partial in this setting by
\[
|\partial_k \smax^\beta(u)| \leq w^\beta_k(u) + w^\beta_k(u)|\beta\wt{\Delta}_k| \leq 1 + \frac{1}{e}.
\]

Next, we consider the case where $u_k \geq \smax^\beta(u)$. Here, we have the bound
\[
|w_k^\beta(u)[1 + \beta(u_k - \smax^\beta(u))]| \leq 1 + \beta(M - \smax^\beta(u)),
\]
where $M := \max_i u_i$ for convenience. Using a standard softmax trick, we have the bound:
\begin{align*}
M - \smax^\beta(u) &= \sum_{j}w^{\beta}_j(u)(M - u_j) \\
&\leq \frac{1}{\beta}\log\left(\sum_j w^{\beta}_j(u) e^{\beta (M - u_j)}\right) \\
&= \frac{1}{\beta}\log\left(\frac{e^{\beta M}}{\sum_j e^{\beta u_j}}\sum_j 1\right) \\
&\leq \frac{1}{\beta}\log(N).
\end{align*}
The first inequality above follows from Jensen's inequality, as for any random variables $X$ and parameter $\lambda > 0$,
\[
\E[X] \leq \frac{1}{\lambda}\log \E \exp\{\lambda X\}.
\]
Plugging this into the bound obtained in the previous display yields:
\[
w_k^\beta(u)|1 + \beta(u_k - \smax^\beta(u))| \leq 1 + \log(N).
\]
Since we have shown the component-wise upper bound in all cases and since $e^{-1} < \log(N)$ when $N \geq 2$, we are done.

\paragraph{Proof of Point~\ref{pt:softmax_lip}}

The proof of Lipschitz continuity is straightforward. In particular, let $u, u' \in \R^N$. We have via the mean value theorem that for some $\wb{u} \in [u, u']$ that
\begin{align*}
|\smax^\beta_\ell(u) - \smax^\beta_\ell(u')| &= |\langle \nabla_u \smax^\beta_\ell(\wb{u}), u - u'\rangle| \\
&\leq \sup_{\wb{u}} \|\nabla_u \smax^\beta_\ell(\wb{u})\|_2 \|u - u'\|_2 \\
&\leq \sqrt{N} \sup_{\wb{u}}\|\nabla_u \smax^\beta_\ell(u)\|_\infty \|u - u'\|_2 \\
&\leq \sqrt{N}(1 +\log(N)) \|u - u'\|_2,
\end{align*}
which proves the desired result.

\paragraph{Proof of Point~\ref{pt:softmax_hess}}
We now go about bounding the Hessian. First, we compute the $(k, \ell)$th coordinate of $\nabla_u^2 \smax^\beta(u)$. Computation similar to that performed in the previous paragraphs yields
\[
\partial_{k, \ell}\smax^\beta(u) = \beta w^\beta_k(u)\left[(\delta_{k, \ell} - w^\beta_\ell(u))(1 + \beta(u_k - \smax^\beta(u))) + (\delta_{k, \ell} - w^\beta_\ell(u)) - \beta w_\ell^\beta(u)(u_\ell - \smax^\beta(u))\right].
\]
We now bound the $\ell_1$-norm of any row (equivalently, any column) of the Hessian. By applying the triangle inequality, we have for any $k$,
\begin{align*}
\|(\nabla^2_u \smax^\beta(u))_k\|_{1} &\leq \underbrace{\beta w_k^\beta(u)|1 + \beta(u_k - \smax^\beta(u))|\sum_{\ell}|\delta_{\ell, k} - w^\beta_\ell(u)|}_{T_1} \\
&\qquad + \underbrace{\beta w_k^\beta(u)\sum_\ell w_\ell^\beta(u)\beta|u_\ell - \smax^\beta(u)|}_{T_2} \\
&\qquad + \underbrace{\beta w_k^\beta(u) \sum_{\ell}|\delta_{k, \ell} - w^\beta_\ell(u)|}_{T_3}.
\end{align*}

Our argument used in the proof of the previous part shows that $w_k^\beta(u)|1 + \beta(u_k - \smax^\beta(u))| \leq 1 + \log(N)$, and so we have 
\begin{align*}
T_1 \leq \beta(1 + \log(N))\sum_{\ell}|\delta_{\ell,k} - w_\ell^\beta(u)| \leq 2\beta(1 + \log(N)),
\end{align*}
where the final inequality follows from the fact
\[
\sum_\ell|\delta_{\ell, k} - w_\ell^\beta(u)| \leq \sum_\ell \delta_{\ell, k} + \sum_\ell w_\ell^\beta(u) = 2.
\]
Next we bound $T_2$. First,  note that we have $\beta u_\ell = \log w^\beta_\ell(u) + \log\left(\sum_j e^{\beta u_j}\right)$ and $\beta \smax^\beta_\ell(u) = \sum_j w^\beta_j(u)\log w^\beta_j (u) + \log\left(\sum_j e^{\beta u_j}\right)$. Thus we have the bound
\[
\beta|u_\ell - \smax^\beta_\ell(u)| = \left|\log w_\ell^\beta(u) - \sum_j w_j^\beta(u) \log w_j^\beta(u)\right| = \left|\log w_\ell^\beta(u)  + H(u; \beta)\right|
\]
where we have defined $H(u; \beta) := - \sum_j w_j^\beta(u)\log w_j^\beta(u)$ as the entropy of the softmax weights. Consequently, we have
\begin{align*}
T_2 &= \beta w_k^\beta(u) \sum_\ell w_\ell^\beta(u)\beta|u_\ell - \smax^\beta(u)| \\
&= \beta w_k^\beta(u) \sum_\ell w_\ell^\beta(u)  \left|\log w_\ell^\beta(u) + H(u; \beta)\right| \\
&\leq \beta w_k^\beta(u) \sum_\ell w_\ell^\beta(u)\left|\log w_\ell^\beta(u)\right| + \beta w_k^\beta(u) \sum_\ell w_\ell^\beta(u) H(u; \beta)\\
&= 2 \beta w_k^\beta(u) H(u; \beta) \leq 2 \beta \log(N)
\end{align*}
where the final equality follows from $|w_k^\beta(u)| \leq 1$ and $0 \leq H(u; \beta) \leq \log(N)$ since the weights $w_1^\beta(u),\dots,w_N^\beta(u)$ describe a distribution on $N$ elements.

Now, we lastly bound $T_3$. In particular we have
\begin{align*}
T_3 = \beta w_k^\beta(u)\sum_\ell |\delta_{k, \ell} - w_\ell^\beta(u)| \leq 2\beta .
\end{align*}
Putting these together yields a bound of $\|(\nabla_u^2 \smax^\beta(u))_k\|_1 \leq 6\beta + 4\beta\log(N)$. 
We conclude by noting that
\[
\|\nabla_u^2 \smax^\beta(u)\|_{op} \leq \max_k \|(\nabla_u^2 \smax^\beta(u))_k\|_1 \leq 6\beta  + 4\beta\log(N).
\]
In the above, the first inequality follows from Corollary 6.1.5 of \citet{horn2012matrix}, which states that for any complex-valued square matrix $A \in \C^{n \times n}$, one has 
\[
\|A\|_{op} \leq \min\left\{\max_i \sum_{j = 1}^n |a_{i, j}|, \max_j \sum_{i = 1}^n |a_{i, j}|\right\},
\]
where $|z| = \sqrt{z\wb{z}}$ for any $z \in \C$. This proves the desired result.

\paragraph{Proof of Point~\ref{pt:softmax_lim}}

The proof here is standard, but we include it for completeness. Letting again $M := \max_i u_i$, $\Delta_i := M - u_i$, and now defining $A := |\arg\max\{u_1, \dots, u_N\}|$, we have 
\begin{align*}
\lim_{\beta \rightarrow \infty} \smax^\beta_\ell(u) &= \lim_{\beta \rightarrow \infty}\frac{\sum_i u_i\exp\{\beta u_i\}}{\sum_j \exp\{\beta u_j\}} \\
&= \lim_{\beta \rightarrow \infty}\frac{\sum_i u_i\exp\{-\beta \Delta_i\}}{\sum_j \exp\{-\beta \Delta_j\}} \\
&= \lim_{\beta \rightarrow \infty} \frac{\sum_{i : \Delta_i = 0} u_i\exp\{-\beta \Delta_i\}}{\sum_j \exp\{-\beta \Delta_j\}} + \lim_{\beta \rightarrow \infty}\frac{\sum_{i : \Delta_i > 0}u_i \exp\{-\beta \Delta_i\}}{\sum_j \exp\{-\beta \Delta_j\}} \\
& = M\lim_{\beta \rightarrow \infty}\frac{A}{A + \sum_{j : \Delta_j > 0}\exp\{-\beta \Delta_j\}} + \lim_{\beta \rightarrow \infty}\frac{\sum_{i: \Delta_i > 0}u_i\exp\{-\beta \Delta_i\}}{A + \sum_{j : \Delta_j >0}\exp\{-\beta \Delta_j\}} \\
&= M.
\end{align*}

\paragraph{Proof of Point~\ref{pt:softmax_grad_lim}}
Recall that we have the identity
\[
(\nabla_u \smax^\beta(u))_i = \partial_i \smax^\beta(u) = \frac{e^{\beta u_i}}{\sum_j e^{\beta u_j}}\left(1 + \beta(u_i - \smax^\beta(u))\right).
\]
We consider two cases. First, assume that $i \in \arg\max\{u_1, \dots, u_N\}$. We start by showing that
\[
\beta (u_i - \smax^\beta(u)) \xrightarrow[\beta \rightarrow \infty]{} 0.
\]
To see this, let $\calA := \arg\max\{u_1, \dots, u_N\}$ and $A = |\calA|$. Observe that we have 
\begin{align*}
\lim_{\beta \rightarrow \infty}\beta (u_i - \smax^\beta(u)) &= \frac{\sum_j \beta (u_i - u_j)e^{\beta u_j}}{\sum_j e^{\beta u_j}} \\
&= \lim_{\beta \rightarrow \infty}\frac{\sum_{j \notin \calA} \beta(u_i - u_j)e^{\beta u_j}}{\sum_{j \in \calA}e^{\beta u_j} + \sum_{j \notin \calA}e^{\beta u_j}} \\
&= \lim_{\beta \rightarrow \infty}\frac{\sum_{j \notin \calA}\beta(u_i - u_j)e^{\beta(u_i - u_j)}}{\sum_{j \in \calA}e^{\beta(u_i - u_j)} + \sum_{j \notin \calA}e^{\beta (u_i - u_j)}} \\
&= \lim_{\beta \rightarrow \infty}\frac{\sum_{j \notin \calA}\beta \Delta_j e^{-\beta \Delta_j}}{A + \sum_{j \notin \calA}e^{-\beta \Delta_j}} \\
&= \frac{\lim_{\beta \rightarrow \infty}\sum_{j \notin \calA} \beta \Delta_j e^{-\beta \Delta_j}}{A + \lim_{\beta \rightarrow \infty}\sum_{j \notin \calA}e^{-\beta \Delta_j}} \\
&=0 
\end{align*}
where the final equality follows from the fact that $\lim_{u \rightarrow \infty} ue^{-u} = 0$ and $\lim_{u \rightarrow \infty}e^{-u} = 0$. Further, using a similar argument, we see that we have
\begin{align*}
\frac{e^{\beta u_i}}{\sum_j e^{\beta u_j}} &= \frac{e^{-\beta \Delta_i}}{\sum_j e^{-\beta \Delta_j}} \\
&= \frac{1}{A + \sum_{j \notin \calA}e^{-\beta \Delta_j}} \\
&\xrightarrow[\beta \rightarrow \infty]{} \frac{1}{A}.
\end{align*}
Consequently, when $i \in \calA$, we have
\[
\lim_{\beta \rightarrow \infty} \frac{e^{\beta u_i}}{\sum_j e^{\beta u_j}}\left[1 + \beta(u - \smax^\beta(u))\right] = \frac{1}{A}.
\]
Now, we consider the case where $i \notin \calA$. We already know that $\frac{e^{\beta u_i}}{\sum_j e^{\beta u_j}} \rightarrow 0$ as $\beta \rightarrow \infty$, and further we see we have
\begin{align*}
\lim_{\beta \rightarrow \infty} \frac{\beta e^{\beta u_i}}{\sum_j e^{\beta u_j}} &= \lim_{\beta \rightarrow \infty} \frac{\beta e^{-\beta \Delta_i}}{\sum_j e^{-\beta \Delta_j}} \\
&= \frac{1}{\Delta_i}\lim_{\beta \rightarrow \infty} \frac{\beta \Delta_ie^{-\beta \Delta_i}}{\sum_{j \in \calA} e^{-\beta \Delta_i} + \sum_{j \notin \calA}e^{-\beta \Delta_j}} \\
&= \frac{1}{\Delta_i}\cdot\frac{\lim_{\beta \rightarrow \infty} \beta \Delta_i e^{-\beta \Delta_i}}{A + \lim_{\beta \rightarrow \infty} e^{-\beta \Delta_j}} \\
&= 0
\end{align*}
where the first equality comes from multiplying the top and bottom by $e^{-\beta M}$ where $M := \max\{u_1, \dots, u_N\}$, and the final equality again comes from the fact that $\lim_{u \rightarrow \infty}ue^{-u} = 0$. Putting the pieces together, we have that
\[
\lim_{\beta\rightarrow \infty}\frac{e^{\beta u_i}}{\sum_{j} e^{\beta u_j}}\left[1 + \beta(u_i - \smax^\beta(u))\right] = \left(\lim_{\beta \rightarrow \infty}\frac{\beta e^{\beta u_i}}{\sum_j e^{\beta u_j}}\right)\cdot \left(\lim_{\beta \rightarrow \infty}\{u_i - \smax^{\beta}(u)\}\right) = 0 \cdot \Delta_i = 0,
\]
thus completing the proof.

\end{proof}

We now prove Lemma~\ref{lem:margin}.


\begin{figure}[H]
    \centering
    \includegraphics[scale=.5]{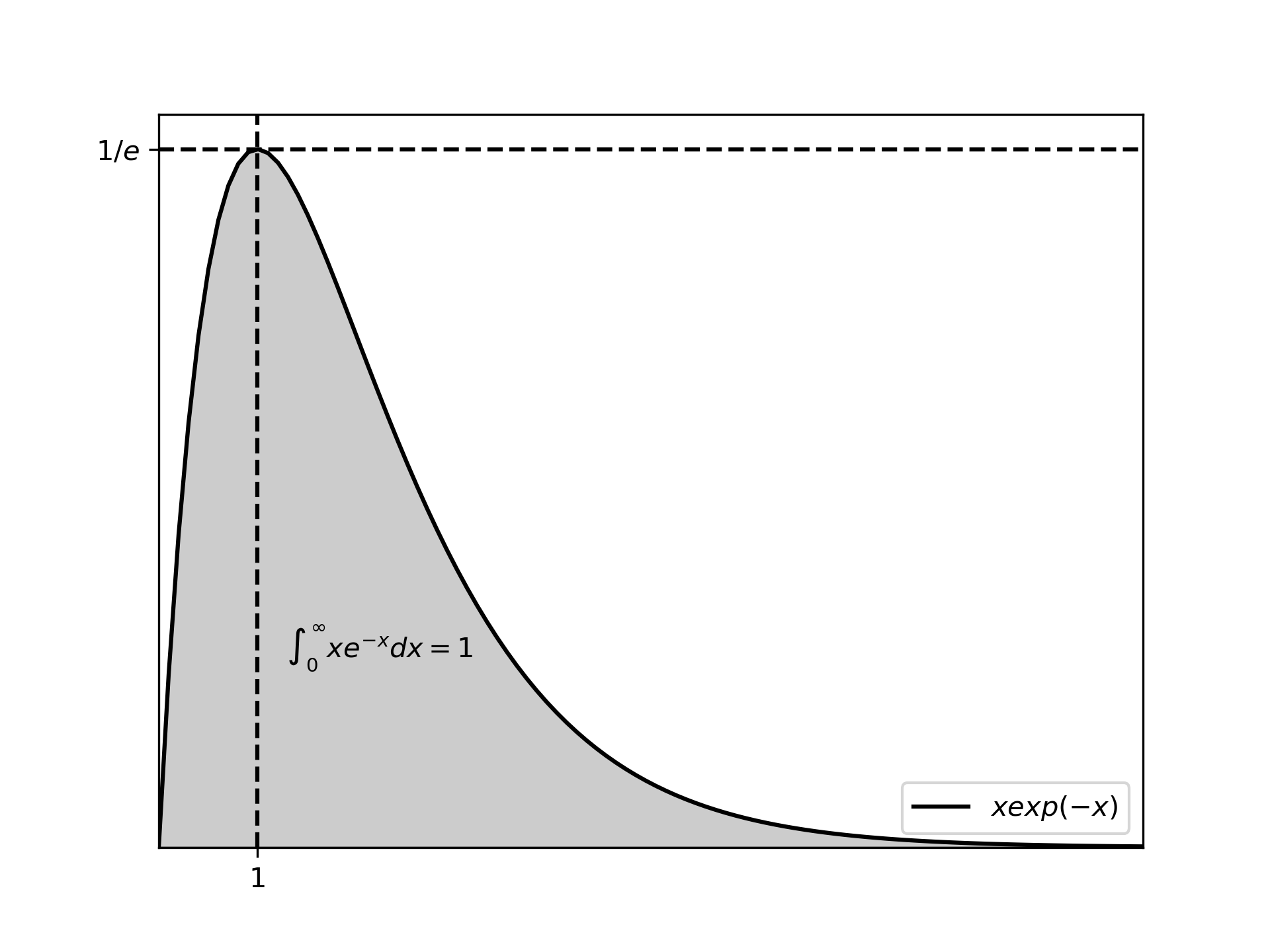}
    \caption{Behavior of function $x \exp\{-x\}$.}
    \label{fig:graph}
\end{figure}

\begin{proof}[Proof of Lemma~\ref{lem:margin}]
We start by proving the first claim. Let $M := \max_i U_i$ and recall  $w_i^\beta(u) := \frac{e^{\beta U_i}}{\sum_{j = 1}^N e^{\beta U_j}}$. First, observe that the weight $w_i^\beta(u)$ can be bounded as
\[
w_i^\beta(u) = \frac{e^{-\beta M}}{e^{-\beta M}}\frac{e^{\beta U_i}}{\sum_j e^{\beta U_j}} = \frac{e^{-\beta \Delta_i}}{\sum_j e^{-\beta \Delta_j}} \leq e^{-\beta \Delta_i},
\]
where the final inequality follows since the sum in the denominator contains $e^{-\beta \Delta_{i^\ast}} = 1$, where $i^\ast = \arg\max(u)$. From this, we can deduce the result via the following reasoning.
\begin{align*}
M - \smax_i^\beta U_i &= \sum_{i = 1}^N w_i^\beta(U)\Delta_i \leq \sum_{i = 1}^N \Delta_i e^{-\beta \Delta_i}.
\end{align*}
Taking expectations furnishes the desired result.


    

Now, we prove the second claim. Let $U$ be a random variable satisfying Assumption~\ref{ass:margin}. We first show that for any $\delta > 0$, we have
\begin{equation}
\label{eq:bias_bd_1}
\E[U\exp(-\beta U)] \leq H \left(\frac{1}{\beta}\right)^{1 + \delta}\Gamma(\delta + 1) + \frac{(1 + \epsilon)\log(\beta)}{\beta^{2 + \epsilon}}
\end{equation}
where $H >0$ and $\epsilon > 0$ are arbitrary and $\beta > 0$ is sufficiently large such that $ \alpha := (1 + \epsilon)\log(\beta)/\beta < c$ and $(1 + \epsilon)\log(\beta) > 1$. We can decompose the left-hand side of the above expression as
\begin{align*}
\E\left[U\exp\left(-\beta U\right)\right]
=~& \underbrace{\E\left[U\exp\left(-\beta U\right)\mathbbm{1}\left\{0<U\leq \alpha\right\}\right]}_{A_\beta} 
+ \underbrace{\E\left[U\exp\left(-\beta U\right)\mathbbm{1}\left\{U> \alpha\right\}\right]}_{B_\beta}
\end{align*}
We now bound each of the terms $A_{\beta}$ and $B_{\beta}$. First, since $\frac{(1+\epsilon)\log(\beta)}{\beta}\leq c$, we know that  the density of $U$ is upper bounded by $H/x^{1-\delta}$ for any $x\in (0, \alpha)$. Then the following holds
\begin{align*}
A_\beta :=~& \E\left[U\exp\left(-\beta U\right)\mathbbm{1}\left\{0 < U\leq \alpha\right\}\right]= \int
_{0}^{\alpha}u
\exp(-\beta u) f(u) du\\
\leq~& H\cdot\int
_{0}^{\alpha}u
\exp(-\beta u) \frac{1}{u^{1-\delta}} du\\
=~& H\beta^{-2}\int_{0}^{{(1+\epsilon)\log\beta} }v\exp(-v) \frac{\beta^{1-\delta}}{v^{1-\delta}} dv
\tag{by a change of variable, $\beta u=v$ and $dv = {\beta} du$}\\
\le~& H \left(\frac{1}{\beta}\right)^{1 + \delta}\int_{0}^{\infty}v^{\delta}\exp(-v) dv \leq H \left(\frac{1}{\beta}\right)^{1 + \delta}\Gamma(\delta + 1).
\end{align*}
We now bound $B_{\beta}$. Note that when $U > \alpha := \frac{(1+\epsilon)\log\beta}{\beta},$ we have
\begin{align*}
    U\exp\left(-\beta U\right)\overset{(a)}{ \le }\alpha\exp\left(-\beta\alpha\right)
    = \beta^{-2 - \epsilon}\cdot (1+\epsilon) \log\beta,
\end{align*}
where to obtain (a) we used the fact that the function $x\mapsto x\exp(-x)$ is decreasing for $x > 1$ (see Figure~\ref{fig:graph}) and we assumed that $\beta$ was large enough such that $(1+\epsilon) \log(\beta) > 1$. Combining these bounds yields the bound in Equation~\eqref{eq:bias_bd_1}.

We conclude by noting that if we simply select $\epsilon_\ast > \delta - 1$, there exists $\beta_\ast > 0$ such that the two stated conditions on $\beta$ hold and such that $\frac{(1 + \epsilon_\ast)\log(\beta)}{\beta^{2 + \epsilon_\ast}} \leq \left(\frac{1}{\beta}\right)^{1 + \delta}$. Thus, the claimed bound holds with $C = H\Gamma(1 + \delta) + 1$.
\end{proof}

\section{Proofs from Section~\ref{sec:static}}
\label{app:treatment}
The proofs from Section~\ref{sec:static} largely follow as corollaries of our more general results from Section~\ref{sec:irregular}, namely Proposition~\ref{prop:general_score} and Theorem~\ref{thm:normal_irregular}. Nonetheless, we provide short proofs for each result. We start with arguing the orthogonality of the smoothed score for the value of the optimal treatment policy.

\begin{proof}[Proof of Proposition~\ref{prop:ortho_static}]
To prove this result, we just need to check that the conditions of Assumption~\ref{ass:constituent_scores} from Section~\ref{sec:irregular} follow from Assumption~\ref{ass:static}. Then, the orthogonality of the described score will follow from Proposition~\ref{prop:general_score}. First, note that in this setting, the $k$th nuisance function $g_k^\ast(a, x)$ is just $g_k^\ast(a, x) := Q^\ast(a, x) := \E[Y \mid X = x, A =a]$, so the first assumption is satisfied with $U_k = Y$. Next, observe that in this setting we have
$\psi_k(X; Q) := Q(k, X)$, and so the second condition is trivially satisfied as well. Finally, note that because of  positivity, conditional ignorability, and Consistency/SUTVA (Assumption~\ref{ass:static}), we have
\[
D_Q\psi_k(X; Q)(\omega) = \omega(k, X) = \E\left[\frac{\1\{A = k\}}{p^\ast(k \mid X)}\omega(k, X) \mid X\right] = \E\left[\zeta_k^\ast(A, X)\omega(A, X) \mid X\right]
\]
where $\zeta_k(A, X) := \frac{\1\{A = k\}}{p^\ast(k \mid X)}$ is bounded (hence square-integrable) by strong positivity. Thus, Proposition~\ref{prop:general_score} implies the score
\[
\Psi^\beta(Z; Q, \alpha) :=\smax^\beta_\ell Q(\ell, X)  + \left(\sum_{\ell = 1}^N \alpha_\ell(A, X)\right)\left\{Y - Q(A, X)\right\}
\]
is Neyman orthogonal with respect to $Q$ and $\alpha = (\alpha_1, \dots, \alpha_N)$, where the $\ell$th true represent is given by $\alpha^\beta_\ell(A, X) = \partial_\ell\smax^\beta_k\{Q^\ast(k, X)\} \cdot \frac{\1\{A = \ell\}}{p^\ast(\ell \mid X)}$. 
\end{proof}

Next, we need to check that the assumptions of Theorem~\ref{thm:normal_static} and Assumption~\ref{ass:static} imply those of the more general result presented in Theorem~\ref{thm:normal_irregular}. 

\begin{proof}[Proof of Theorem~\ref{thm:normal_static}]
We iterate through each of the five assumptions of Theorem~\ref{thm:normal_irregular}. The first and second assumptions of Theorem~\ref{thm:normal_irregular} are the same as those of Theorem~\ref{thm:normal_static}. Next, we verify that strong positivity implies mean-squared continuity. In particular, letting $\eta$ be such that $\eta \leq p^\ast(\ell | X)$ for all $\ell \in [N]$, one checks that we have
\begin{align*}
\|\psi_\ell(X; g_\ell) - \psi_\ell(X; g_\ell')\|_{L^2(P_X)} &= \|Q(\ell, X) - Q'(\ell, X)\|_{L^2(P_X)} \\
&= \E\left[(Q(\ell, X) - Q'(\ell, X))^2\right]^{1/2} \\
&= \E\left[\frac{\mathbbm{1}\{A = \ell\}}{p^\ast(\ell \mid X)}\left(Q(\ell, X) - Q'(\ell, X)\right)^2\right]^{1/2} \\
&\leq \frac{1}{\sqrt{\eta}}\E[(Q(A, X) - Q'(A, X))^2]^{1/2} \\
&\lesssim \|Q - Q'\|_{L^2(P_{W})}.
\end{align*}

We note that the assumptions on the nuisance convergence rates and nuisance consistency are identical (the fourth condition in Theorem~\ref{thm:normal_irregular} and the third condition in Theorem~\ref{thm:normal_static}) and likewise noting that $U_k = Y$, the assumptions on boundedness (the fifth condition) are identical. Thus, the assumptions of Theorem~\ref{thm:normal_irregular} are satisfied, implying the estimate $\wh{V}$ is asymptotically linear with the noted influence function. This completes the proof.


\end{proof}


\subsection{Proofs from Section~\ref{sec:static:param}}
We now prove convergence results under semi-parametric restrictions on the blip functions $\gamma^\ast(a, x) := Q^\ast(a, x) - Q^\ast(a^\ast, x)$. We start with the proof of orthogonality for the score targeting the structural parameter.

\begin{proof}[Proof of Proposition~\ref{prop:ortho_static_param}]
First, we check orthogonality with respect to $\mu$. For any fixed $\Delta \in L^2(P_X)$, we have 
\begin{align*}
D_\mu\E\left[\Psi(Z; \theta_0, \mu^\ast, r^\ast)\right](\Delta) &= -\E\left[\Delta(X)(\phi(A, X) - r^\ast(X))\right] \\
&= - \E[\Delta(X)\E\left(\phi(A, X) - r^\ast(X) \mid X\right)] \\
&= 0.
\end{align*}
Next, we check orthogonality with respect to $r$. We have that 
\begin{align*}
D_r \E\left[\Psi(Z; \theta_0, \mu^\ast, r^\ast)\right](\Delta) &= - \E\left[\Delta(X)(Y - \mu^\ast(X))\right] + \theta_0^\top\E\left[\Delta(X)(\phi(A, X) - r^\ast(X))\right] \\
&\qquad + \theta_0^\top\E\left[(\phi(A, X) - r^\ast(X)) \Delta(X)\right] \\
&= -\E\left[\Delta(X)\E\left(Y - \mu^\ast(X) \mid X\right)\right]  + \theta_0^\top\E\left[\Delta(X)\E(\phi(A, X) - r^\ast(X) \mid X)\right]\\
&\qquad + \theta_0^\top\E\left[\E(\phi(A, X) - r^\ast(X) \mid X) \Delta(X))\right]\\
&= 0,
\end{align*}
where the third equality follows from the tower rule and the definitions of $\mu^\ast(X)$ and $r^\ast(X)$. Now, we show that $\theta_0$ is a solution to the population estimating equation. In particular, recalling $Q^\ast(a, x) := \E[Y \mid A = a, X= x]$, we have
\begin{align*}
&\E\left[\Psi(Z; \theta_0, \mu^\ast, r^\ast)\right] = \E\left[\left\{(Y - \mu^\ast(X)) - \theta_0^\top (\phi(A, X) - r^\ast(X))\right\}(\phi(A, X) - r^\ast(X))\right] \\
&\qquad = \E\left[\left\{(Q^\ast(A, X) - \E(Q^\ast(A, X) \mid X)) - (\gamma^\ast(A, X) - \E\left(\gamma^\ast(A, X) \mid X\right))\right\}(\phi(A, X) - r^\ast(X))\right] \\
&\qquad = \E\left[\left\{(Q^\ast(A, X) - \E(Q^\ast(A, X) \mid X)) - (Q^\ast(A, X) - \E\left(Q^\ast(A, X) \mid X\right))\right\}(\phi(A, X) - r^\ast(X))\right] \\
&\qquad =0,
\end{align*}
where the second equality follows from the tower rule after taking conditional expectations given $A$ and $X$ and the third equality follows because $\gamma^\ast(A, X) - \E(\gamma^\ast(A, X) \mid X) = Q^\ast(A, X) - \E(Q^\ast(A, X) \mid X)$. To see uniqueness, note that for any $\theta$, we have (letting $x^{\otimes 2} = x x^\top$ for any vector $x$)
\begin{align*}
\E\left[\Psi(Z; \theta, \mu^\ast, r^\ast)\right] &= (\theta_0 - \theta)^\top \E\left[(\phi(A, X) - r^\ast(X))^{\otimes 2}\right] \neq 0.
\end{align*}
unless $\E\left[(\phi(A, X) - r^\ast(X))^{\otimes 2}\right] = \E\left[\Cov(\phi(A, X) \mid X)\right] = 0$ (which we have assumed is not the case) or $\theta = \theta_0$, thus completing the proof.

\end{proof}

Next, we show that the de-biased estimate $\wh{\theta}_n$ for $\theta_0$ is asymptotically linear with the specified influence function.

\begin{proof}[Proof of Theorem~\ref{thm:normal_static_param}]
To show asymptotic normality, we check the appropriate conditions of Theorem~\ref{thm:smooth_clt_generic} for the score
\[
\Psi(Z; \theta, \mu, r) := \underbrace{-(\phi(A, X) - r(X))(\phi(A, X) - r(X))^\top}_{=:b(z; r)}\theta + \underbrace{(Y - \mu(X))(\phi(A, X) - r(X))}_{=:\nu(Z; \mu, r)}.
\]
We let $B(r) := \E[b(Z; r)]$ and $\calV(\mu, r) := \E[\nu(Z; \mu, r)]$ for notational ease. Since $\Psi$ depends on neither a finite-dimensional nuisance component nor the smoothing parameter $\beta > 0$, we do not need to check Conditions~\ref{cond:score} and \ref{cond:h_conditions} of Theorem~\ref{thm:smooth_clt_generic}.

\paragraph{Checking Theorem~\ref{thm:smooth_clt_generic}, Condition~\ref{cond:neyman}}
Neyman orthogonality was shown in Proposition~\ref{prop:ortho_static_param}.

\paragraph{Checking Theorem~\ref{thm:smooth_clt_generic}, Condition~\ref{cond:hessian}}
We let $g^\ast = (\mu^\ast, r^\ast)$, $\wh{g} = (\wh{\mu}, \wh{r})$, and $\wb{g} \in [g^\ast, \wh{g}]$ for notational convenience.  We have
\begin{align*}
D_g^2 \E_Z\left[\Psi(Z; \theta, \wb{g})\right](\wh{g} - g^\ast) &= 2D_{\mu, r}\E_Z\left[\wb{\mu}(X)\wb{r}(X)\right](\wh{\mu} - \mu^\ast, \wh{r} - r^\ast) \\
&\qquad - D_{r}^2\E_Z\left[\wb{r}(X)\wb{r}(X)^\top\theta_0\right](\wh{r} - r^\ast).
\end{align*}

Letting $r_k$ denote the $k$th coordinate of a vector-valued function $r : \calX \rightarrow \R^d$, we can bound the cross Gateaux derivative as
\begin{align*}
\left\|D_{\mu, r}\E_Z[\wb{\mu}(X)\wb{r}(X)](\wh{\mu} - \mu^\ast, \wh{r} - r^\ast)\right\|_2 &= \left\|\E_Z\left[(\wh{\mu} - \mu^\ast)(X)(\wh{r} - r^\ast)(X)\right]\right\|_2 \\
&= \sqrt{\sum_{k = 1}^d \E_Z\left[(\wh{\mu} - \mu^\ast)(X)(\wh{r}_k - r_{k}^\ast)(X)\right]^2}\\
&\leq \sqrt{\sum_{k = 1}^d \|\wh{\mu} - \mu^\ast\|_{L^2(P_{X})}^2 \|\wh{r}_{k} - r^\ast_k \|_{L^2(P_{X})}^2} \\
&= \|\wh{\mu} - \mu^\ast\|_{L^2(P_{X})}\cdot\|\wh{r} - r^\ast\|_{L^2(P_{X})} \\
&= o_\P(n^{-1/2}),
\end{align*}
where the  first inequality follows from apply the Cauchy-Schwarz inequality to the expectation inside the sum, the second to last equality follows since $\|r\|_{L^2(P_X)} := \left(\int_\calX \|r(x)\|_2^2 P_X(dx)\right)^{1/2}$ for a vector-valued function $r$, and the final equality follows from the assumption on nuisance estimation rates.

The second Gateaux derivative with respect to $r$ can be computed similarly. Letting $\wh{\Delta} := \wh{r} - r^\ast$, we have
\begin{align*}
D_{r}^2\E_Z\left[\wb{r}(X)\wb{r}(X)^\top\right](\wh{\Delta}) &= \frac{\partial^2}{\partial t^2}\E_Z\left[(\wb{r} + t \wh{\Delta})(X)(\wb{r} + t\wh{\Delta})(X)^\top\right]\Big|_{t = 0} \\
&= \E_Z\left[\frac{\partial^2}{\partial t^2}(\wb{r} + t \wh{\Delta})(X)(\wb{r} + t \wh{\Delta})(X)^\top\big|_{t = 0}\right] \\
&= 2\E_Z\left[\wh{\Delta}(X)\wh{\Delta}(X)^\top\right],
\end{align*}
Where the exchanging of the second partial derivative and expectation is justified by the fact that $\wh{\Delta}$ and $\wb{r}$ are both bounded. Thus, we have 
\begin{align*}
\left\|D_{r}^2\E_Z\left[(\theta_0)^\top \wb{r}(X)\wb{r}(X)\right](\wh{\Delta})\right\|_2 &= 2\left\|\E_Z\left[\wh{\Delta}(X)\wh{\Delta}(X)^\top\right]\theta_0\right\|_2 \\
&\leq 2 \left\|\E\left[\wh{\Delta}(X)\wh{\Delta}(X)^\top\right]\right\|_{op}\left\|\theta_0\right\|_2 \\
&= 2\sup_{\|u\| = 1}\E_Z\left[u^\top \wh{\Delta}(X)\wh{\Delta}(X)^\top u\right] \|\theta_0\|_2 \\
&\lesssim \E_Z\|\wh{\Delta}(X)\|_2^2 \\
&= \|\wh{r} - r^\ast\|_{L^2(P_{X})}^2 \\
&= o_\P(n^{-1/2}),
\end{align*}
where the first inequality follows from the fact that $\|A x\|_2 \leq \|A\|_{op}\|x\|_2$, the second equality follows from the definition of the operator norm, the second inequality follows from pushing the supremum inside the expectation and noting that $\|\theta_0\|_2$ is bounded by a constant by assumption, and the final equality follows from the assumption of nuisance estimation rates. This completes checking Condition~\ref{cond:hessian}.

\paragraph{Checking Theorem~\ref{thm:smooth_clt_generic}, Condition~\ref{cond:equi}}
 Using an argument exactly analogous to the one used to show stochastic equicontinuity in the proof of Theorem~\ref{thm:normal_irregular} (see the display containing Tag~\eqref{eq:equi-bd}), we have
\begin{align*}
\E\left\|\G_n b(Z; r) - \G_n b(Z; r')\right\|_F^2 \leq \E\|b(Z; r) - b(Z; r')\|_F^2.
\end{align*}
We have, for arbitrary $r, r' : \calX \rightarrow \R^d$ that are bounded in $L^\infty$ norm
\begin{align*}
&\E\left\|b(Z; r) - b(Z; r')\right\|_F^2  = \E\left\|(\phi(A, X) - r(X))^{\otimes 2} -(\phi(A, X) - r'(X))^{\otimes 2}\right\|_F^2 &(\text{Definition of } b)\\
&\qquad = \E\big\|(\phi(A, X) - r(X))^{\otimes 2} \pm (\phi(A, X) - r(X))(\phi(A, X) \\
&\qquad \qquad - r'(X))^\top - (\phi(A, X) - r'(X))^{\otimes 2}\big\|_F^2 \\
&\qquad \lesssim \E\left\|(\phi(A, X) - r(X))(r(X) - r'(X))^\top\right\|_F^2 \\
&\qquad \qquad + \E\left\|(\phi(A, X) - r'(X))(r(X) - r'(X))^\top\right\|_F^2 &(\text{Parallelogram inequality}) \\
&\qquad =\E\Big[\big\{\|\phi(A, X) - r(X)\|_2^2 + \|\phi(A, X) - r'(X)\|_2^2\big\}\|r(X) - r'(X)\|_2^2 \Big] &(\|ab^\top\|_F^2 = \|a\|_2^2\|b\|_2^2)\\
&\qquad \lesssim \E\left\|r(X) - r'(X)\right\|_2^2. 
\end{align*}
The second equality above follows from adding and subtracting $(\phi(A, X) - r(X))(\phi(A, X) - r'(X))^\top$ and the final inequality follows from the fact that $r, r',$ and $\phi$ have bounded $L^\infty$ norms (and hence their point-wise $\ell_2$ norms are almost surely bounded by an absolute constant as well).
The above reasoning allows us to conclude
\[
\E_{Z_1, \dots, Z_n}\|\G_n b(Z; \wh{r}) - \G_n b(Z; r^\ast)\|_F^2 \lesssim \|\wh{r} - r^\ast\|_{L^2(P_{X})}^2 = o_\P(1),
\]
where the last equality follows by nuisance consistency. As in the proof of Theorem~\ref{thm:normal_irregular}, conditionally applying Chebyshev's inequality (as in Tag~\eqref{eq:cond-chebyshev}) and applying Proposition~\ref{prop:vitali} (Vitali's Theorem) yields $\|\G_n b(Z; \wh{r}) - \G_n b(Z; r^\ast)\|_F = o_\P(1)$. An exactly analogous argument (instead adding and subtracting $(y - \mu(x))(\phi(a, x) - r'(x))$ ) yields that
\begin{align*}
\E\left\|\G_n \nu(Z; \mu, r) - \G_n \nu(Z; \mu', r')\right\|_2^2 \leq \E\left\|\nu(Z; \mu, r) - \nu(Z; \mu', r')\right\|_2^2 \lesssim \E\left|(\mu - \mu')(X)\right|^2 + \E\left\|(r - r')(X)\right\|_2^2.
\end{align*}
This in particular implies that
\begin{align*}
\E_{Z_1, \dots, Z_n}\left\|\G_n \nu(Z; \wh{\mu}, \wh{r}) - \G_n \nu(Z; \mu^\ast, r^\ast)\right\|_2^2 \lesssim \left\|\wh{\mu} - \mu^\ast\right\|_{L^2(P_X)}^2 + \left\|\wh{r} - r^\ast\right\|_{L^2(P_X)}^2 = o_\P(1),
\end{align*}
again by nuisance consistency. Another conditional application of Chebyshev's inequality thus yields 
\[
\|\G_n \nu(Z; \wh{\mu}, \wh{r}) - \G_n \nu(Z; \mu^\ast, r^\ast)\|_2 = o_\P(1),
\]
thus completing the proof of stochastic equicontinuity.

\paragraph{Checking Theorem~\ref{thm:smooth_clt_generic}, Condition~\ref{cond:regularity}}
First, we note that 
\[
-B(r^\ast) := \E\left[(\phi(A, X) - r^\ast(X))^{\otimes 2}\right]=\E\left[\Cov(\phi(A, X) \mid X)\right] \succ 0
\]
by assumption, and hence $B(r^\ast)^{-1}$ exists. Next, by again adding and subtracting $(\phi(A, X) - r(X))(\phi(A, X) - r'(X))$, we have 
\begin{align*}
&\left\|B(\wh{r}) - B(r^\ast)\right\|_{op} \leq \left\|B(\wh{r}) - B(r^\ast)\right\|_F \\
&\qquad \leq \left\|\E_Z\left[(\phi(A, X) - \wh{r}(X))(\wh{r}(X) - r^\ast(X))^\top\right]\right\|_F + \left\|\E_Z\left[(\phi(A, X) - r^\ast(X))(\wh{r}(X) - r^\ast(X))^\top\right]\right\|_F \\
&\qquad \leq \E_Z\left\|(\phi(A, X) - \wh{r}(X))(\wh{r}(X) - r^\ast(X))^\top\right\|_F + \E_Z\left\|(\phi(A, X) - r^\ast(X))(\wh{r}(X) - r^\ast(X))^\top\right\|_F \\
&\qquad = \E_Z\left\|\phi(A, X) - \wh{r}(X)\right\|_2\left\|(\wh{r} - r^\ast)(X)\right\|_2 + \E_Z\left\|\phi(A, X) - r^\ast(X)\right\|_2\left\|(\wh{r} - r^\ast)(X)\right\|_2 \\
&\qquad \lesssim \E_Z\left\|(\wh{r} - r^\ast)(X)\right\|_2 \\
&\qquad \leq \left\|\wh{r} - r^\ast\right\|_{L^2(P_X)},
\end{align*}

where the first inequality follows as $\|A\|_{op} \leq \|A\|_F$, the second inequality follows from the triangle inequality, the third follows from Jensen's inequality, the first equality follows from the equality $\|a b^\top\|_F = \|a\|_2 \|b\|_2$, and the final line follows from an application of Jensen's inequality again. This shows mean-squared continuity, and thus completes the proof.
\end{proof}

Lastly, we conclude this appendix by proving the asymptotic normality of the estimate for the value of the optimal policy value.

\begin{proof}[Proof of Corollary~\ref{cor:value_param_static}]
We supress dependence on the sample size $n$ in the sequel to streamline notation. First, observe that we again have
\[
\sqrt{n}(\wh{V}_n - V^\ast) = \sqrt{n}(\wh{V}_n - V^{\beta}) + \sqrt{n}(V^{\beta} - V^\ast) = \sqrt{n}(\wh{V} - V^{\beta}) + o_\P(1).
\]
The second inequality follows from applying Lemma~\ref{lem:margin} as in the Proof of Theorem~\ref{thm:normal_irregular} (which leverages assumptions 1 and 2 of Corollary~\ref{cor:value_param_static}). Asymptotic linearity will now follow if we check the relevant conditions of Theorem~\ref{thm:smooth_clt_generic}. Since there is no infinite-dimensional nuisance component in this setting, we do not need to check Conditions~\ref{cond:neyman}, \ref{cond:hessian}, \ref{cond:equi}, the first part of Condition~\ref{cond:regularity}, and Condition~\ref{cond:h_cont}.

First, note that we can rewrite the score as 
\[
m^\beta(z; V, \theta) = V - \Phi^\beta(z; \theta) \quad \text{and} \quad m^\ast(z; V, \theta) = V - \Phi^\ast(z; \theta),
\]
which instantiates Theorem~\ref{thm:smooth_clt_generic} with $b^\beta \equiv 1$, $h = \theta$, and $\nu^\beta(Z; h) = -\Phi^\beta(Z; \theta)$. Note we use $b^\beta$ here instead of $a^\beta$ to avoid collision with the notation $a \in [N]$, which generally represents treatment. 

\paragraph{Checking Theorem~\ref{thm:smooth_clt_generic}, Condition~\ref{cond:score}} Since we know that
$|V^\beta - V^\ast| = o(1)$, it follows that
\begin{align*}
\lim_{n \rightarrow \infty}\left|m^\beta(Z; V^\beta, \theta_0) - m^\ast(Z; V^\ast, \theta_0)\right| &\leq \lim_{n \rightarrow \infty}\left|V^\beta - V^\ast\right| + \lim_{n \rightarrow \infty}\left|\smax^\beta_{\ell \in [N]} \theta_0^\top \phi(\ell, X) - \max_{\ell \in [N]} \theta_0^\top \phi(\ell, X)\right| \\
&= \lim_{n \rightarrow \infty}\left|\smax^\beta_{\ell \in [N]} \theta_0^\top \phi(\ell, X) - \max_{\ell \in [N]} \theta_0^\top \phi(\ell, X)\right| \\
&= 0
\end{align*}
where the final equality follows from the point-wise convergence of the softmax function to the maximum function (Lemma~\ref{lem:softmax} Point~\ref{pt:softmax_lim}). This yields almost sure convergence, which is stronger than the required convergence in probability. Further, since $B^\beta = B^\ast = 1$, there is nothing to check for convergence here.
\paragraph{Checking Theorem~\ref{thm:smooth_clt_generic}, Condition~\ref{cond:regularity}}
Since we have $B^\ast = B^\beta = 1$, invertibility is immediate. Further, since there is no infinite-dimensional nuisance component, there is nothing to check for mean-squared continuity. Lastly,  since $V^\ast$ and $\Phi^\beta(Z; \theta_0)$ are both bounded by assumption (in $\ell^2$ and $L^\infty(P_Z)$ norm, respectively), all necessary boundedness conditions hold.

\paragraph{Checking Theorem~\ref{thm:smooth_clt_generic}, Condition~\ref{cond:h_conditions}}

We conclude by checking the relevant conditions outlined in Condition~\ref{cond:h_conditions}, which pertain to the continuity/regularity of $m^\beta(Z; V, \theta)$ with respect to the finite-dimensional nuisance $\theta$. 
\begin{enumerate}
    \item \textbf{(Condition~\ref{cond:h_diff})} Note that for any fixed $z \in \calZ$, we have
    \[
    \nu^\beta(z; \theta) =  -\Phi^\beta(z; \theta) = -\smax^\beta_{\ell \in [N]} \theta^\top \phi(\ell, x) + \theta^\top \phi(a, x) - y.
    \]
    This implies $\nu^\beta$ is infinitely-differentiable in $\theta$.
    \item \textbf{(Condition~\ref{cond:h_jacobian})} We claim that this condition holds with $J^\ast$ given by
    \[
    J^\ast = -\E\left[\sum_{k = 1}^N \fraks_k(\theta_0^\top \phi(\ell, X))\phi(k, X) - \phi(A, X)\right] = -\E\left[\phi^\infty(X) - \phi(A, X)\right],
    \]
    where $\phi^\infty(X)$ is as defined in the statement of Corollary~\ref{cor:value_param_static}.
    To see this, note that we can use the boundedness of $\phi(\ell, X)$ and $\theta_0$ to interchange the derivative with expectation, which allows us to write the Jacobian (i.e.\ gradient since $\Phi$ is real-valued) $\nabla_\theta \Phi(Z; \theta_0)$ of $\Phi(Z; \theta_0)$ as 
    \begin{align*}
    \nabla_\theta \Phi(Z;\theta_0) &= \E\left[\nabla_\theta \smax^\beta_\ell \left\{\theta_0^\top \phi(\ell, X)\right\} - \partial_\theta \theta_0^\top \phi(A, X)\right] \\
    &= \E\left[\sum_{k = 1}^N \left(\nabla_u \smax^\beta_\ell\left\{ \theta_0^\top \phi(\ell, X)\right\}\right)_k \phi(k, X)  - \phi(A, X)\right],
    \end{align*}
    where the final line follows from the chain rule and we recall we have let $\nabla_u \smax^\beta_\ell\left\{ \theta_0^\top \phi(\ell, X)\right\} = \nabla_u \smax^\beta_\ell\{u_\ell\}\vert_{u = \theta_0^\top \phi(\ell, X)}$ for convenience. We now show that
    \[
    \E\left[\left\{\nabla_u \smax^\beta_\ell\left\{ \theta_0^\top \phi(\ell, X)\right\}_k - \fraks_k(\theta_0^\top \phi(\ell, X))\right\}\phi(k, X)\right] = o(1)
    \]
    for each $k \in [N]$,
    which shows the desired convergence result. Observe that since we have assumed boundedness of $\phi(k, X)$, we have 
    \begin{align*}
    &\lim_{\beta \rightarrow \infty}\left\|\E\left[\left\{\nabla_u \smax^\beta_\ell\left\{ \theta_0^\top \phi(\ell, X)\right\}_k - \fraks_k(\theta_0^\top \phi(\ell, X))\right\}\phi(k, X)\right]\right\|_2 \\
    &\qquad \leq \lim_{\beta \rightarrow \infty}\E\left\|\left\{\nabla_u \smax^\beta_\ell\left\{ \theta_0^\top \phi(\ell, X)\right\}_k - \fraks_k(\theta_0^\top \phi(\ell, X))\right\}\phi(k,X) \right\|_2\\
    &\qquad \lesssim \lim_{\beta \rightarrow \infty}\E\left|\nabla_u \smax^\beta_\ell\left\{ \theta_0^\top \phi(\ell, X)\right\}_k - \fraks_k(\theta_0^\top \phi(\ell, X))\right| \\
    &\qquad = \E\left[\lim_{\beta \rightarrow \infty} \left|\nabla_u \smax^\beta_\ell\left\{ \theta_0^\top \phi(\ell, X)\right\}_k - \fraks_k(\theta_0^\top \phi(\ell, X))\right|\right] \\
    &\qquad = 0,
    \end{align*}
    where the first inequality follows from Jensen's inequality, the second equality follows from the bounded convergence theorem (since we know both $\fraks_k$ and $\nabla_u\smax^\beta_\ell\{u\}$ are uniformly bounded, by Lemma~\ref{lem:softmax}) and the final line follows from the point-wise convergence of $(\nabla_u\smax^\beta_\ell\{u\})_k$ to $\fraks_k(u)$ (also from Lemma~\ref{lem:softmax}).
    \item \textbf{(Condition~\ref{cond:h_hessian})} First, we check the almost sure boundedness of each component of the Jacobian/gradient of $\Phi^\beta$. In particular, for any $j \in [d]$, we have 
    \begin{align*}
    \left|\left(\partial_\theta \Phi^\beta(Z; \theta_0)\right)_j\right| &= \left|\sum_{k = 1}^N\phi_j(k, X)\left(\nabla_u\smax^\beta_\ell\left\{\theta_0^\top \phi(\ell, X)\right\}\right)_k - \phi_j(A, X)\right| \\
    &\leq N\max_k|\phi_j(k, X)|\cdot\left|\left(\nabla_u\smax^\beta_\ell\left\{\theta_0^\top \phi(\ell, X)\right\}\right)_k\right| + \left|\phi_j(A , X)\right|   \\
    &= O(1)
    \end{align*}
    since we have assumed $\phi(k, X)$ is almost surely bounded for each $k$ and since the partial derivatives of the softmax function are bounded by an absolute constant that doesn't depend on $\beta$ (Lemma~\ref{lem:softmax}, Point~\ref{pt:softmax_grad}).
    We now check the rate of growth of the second derivative/Hessian with respect to $\theta$. Again by the chain rule, we have for any $\theta \in [\wh{\theta}, \theta_0]$
    \begin{align*}
    \left\|\partial_\theta^2 \Phi^\beta(Z; \theta)\right\|_{op} &= \left\|\sum_{j, k = 1}^N\left(\nabla_u^2 \smax^\beta_\ell\left\{\theta^\top \phi(\ell, X)\right\}\right)_{k, j}\phi(k, X)\phi(j, X)^\top\right\|_{op} \\
    &\lesssim \max_k \|\phi(k, X)\|^2_2 \left\|\nabla_u^2 \smax^\beta_\ell\{\theta^\top \phi(\ell, X)\}\right\|_{op} \\
    &\lesssim \beta_n.
    \end{align*}
    where the final inequality follows from the boundedness of $\phi(k, X)$ in $L^\infty$ norm and the fact that the operator norm of the Hessian of the softmax function grows as $O(\beta)$ (Lemma~\ref{lem:softmax}, Point~\ref{pt:softmax_hess}). Since we have assumed $\beta_n = o(n^{1/2})$, this condition of the theorem holds.
    \item \textbf{(Condition~\ref{cond:h_linearity})} The asymptotic linearity of $\wh{\theta}_n$ around $\theta_0$ with mean zero influence function $\rho_\theta(z)$ is an assumption of Corollary~\ref{cor:value_param_static}, so there is nothing to prove.
    
\end{enumerate}
\end{proof}
\section{Proofs from Section~\ref{sec:irregular}}
\label{app:irregular}

We now prove our results related to inference on point-wise maximum of general smooth scores. We start by proving Proposition~\ref{prop:general_score}, which provides a general de-biased smoothed score for the parameters specified as an appropriate maximum of scores. 
\begin{proof}[Proof of Proposition~\ref{prop:general_score}]


First, we note that the fact that $V^\beta = \E\left[\Psi^\beta(Z; g^\ast, \alpha^\beta)\right]$ is immediate, as for any $k \in [N]$ we have via the tower rule for conditional expectations that
\[
\E\left[\alpha_k^\beta(W)^\top (U_k - g_k^\ast(W))\right] = \E\left[\alpha_k^\beta(W)^\top(\E[U_k \mid W] - g_k^\ast(W))\right] = 0.
\]
Next, we check Neyman orthogonality. Before proceeding, we argue that we can exchange Gateaux derivatives with expectation. Let $\omega = (\omega_1, \dots, \omega_N) \in L^2(P_W; \R^d)$ be arbitrary, where again $d = d_1 + \cdots + d_N$.
We have, for any $k \in [N]$, letting $g^\ast + t\omega_k := (g^\ast_1, \dots, g^\ast_k + t \omega_k, \dots, g^\ast_N)$ for notational ease,
\begin{align*}
D_{g_k}\E\left[\smax^\beta_\ell \psi(X; g_\ell^\ast)\right](\omega_k) &= \partial_t \E\left[\smax^\beta_\ell \psi(X; g^\ast + t\omega_k)\right]\Big|_{t = 0} \\
&= \lim_{\substack{t \rightarrow 0 \\ |t| \leq 1}}\E\left[\frac{\smax^\beta \psi(X; g^\ast + t \omega_k) - \smax^\beta \psi(X; g^\ast)}{t}\right] &(\text{Definition of partial derivative})\\
&= \lim_{\substack{t \rightarrow 0 \\ |t| \leq 1}}\E\left[\frac{D_{g_k} \smax^{\beta}\psi(X; \wt{g}_t)(t\omega_k)}{t}\right] &(\text{Mean value theorem point-wise}) \\
&= \lim_{\substack{t \rightarrow 0 \\ |t| \leq 1}}\E\left[D_{g_k}\smax^\beta\psi(X; \wt{g}_t)(\omega_k)\right] &(\text{Linearity}) \\
&= \lim_{\substack{t \rightarrow 0 \\ |t| \leq 1}}\E\left[ \partial_k\smax^\beta \psi(X; \wt{g}_t)\cdot D_{g_k}\psi_k(X; \wt{g}_{k,t})(\omega_k)\right] &(\text{Chain rule})
\end{align*}
In the above, $\wt{g}_t = (g_1^\ast, \dots, \wt{g}_{k, t}, \dots, g^\ast_N)$, where there is some $\lambda_t \in [0, 1]$ such that $\wt{g}_{k, t}  = \lambda_t g_k^\ast + (1 - \lambda_t)t\omega_k \in [g_k^\ast, g_k^\ast + t \omega_k]$. 

Observe that, by Lemma~\ref{lem:softmax} Part~\ref{pt:softmax_grad}, $\left|\partial_k \smax^\beta\{u\}\right| \leq C$ for some absolute constant $C > 0$ that does not depend on $\beta$. Further, we have $|D_{g_k}\psi_k\big(X; \wt{g}_{k, t}\big)(\omega_k)| \leq F_{\omega_k}^k(X)$ for some $F_{\omega_k}^k(X) \in L^1(P_X)$ by the second part of Assumption~\ref{ass:constituent_scores}. 
Noting that $\wt{g}_t \xrightarrow[t \rightarrow 0]{a.s.} g^\ast$, we see that $\partial_k\smax^\beta\psi(X; \wt{g}_t)\cdot D_{g_k}\psi_k(X; \wt{g}_{k, t})(\omega_k) \xrightarrow[t \rightarrow 0]{a.s.} \partial_k\smax^\beta\psi(X; g^\ast) \cdot D_{g_k}\psi_k(X; g_k^\ast)(\omega_k)$ as well. Hence, by the dominated convergence theorem (see Chapter 1 of \citet{durrett2019probability}), we can exchange limits with expectations. In particular, we have
\begin{align*}
D_{g_k}\E\left[\smax^\beta \psi(X; g^\ast)\right](\omega_k) &= \lim_{\substack{t \rightarrow 0 \\ |t| \leq 1}}\E\left[ \partial_k\smax^\beta \psi(X; \wt{g}_t)\cdot D_{g_k}\psi_k(X; \wt{g}_{k,t})(\omega_k)\right]\\
&= \E\left[\lim_{\substack{t \rightarrow 0 \\ |t| \leq 1}}\partial_k^\beta \smax^\beta\psi(X; \wt{g}_t) \cdot D_{g_k}\psi_k(X; \wt{g}_{k, t})(\omega_k)\right] \\
&= \E\left[\partial_k\smax^\beta\psi(X; g^\ast) \cdot D_{g_k}\psi_k(X; g_k^\ast)(\omega_k)\right] \\
&= \E\left[D_{g_k}\smax^\beta \psi(X; g^\ast)(\omega_k)\right]
\end{align*}

Now, we can compute the Gateaux derivative of $\E\left[\smax^\beta \psi(X; g^\ast)\right]$ with respect to $g_k$ as follows:
\begin{align*}
D_{g_k}\E\left[\smax_\ell^\beta \psi(X; g_\ell^\ast)\right](\omega_k) &= \E\left[D_{g_k}\smax^\beta_\ell \psi_\ell(X; g_\ell^\ast) (\omega_k)\right] &(\text{Interchange of limits})\\
&= \E\left[\partial_k \smax^\beta_\ell \psi_\ell(X; g_\ell^\ast) D_{g_k}\psi_k(X; g_k^\ast)(\omega_k)\right] &(\text{Chain rule})\\
&= \E\left[\partial_k \smax^\beta_\ell \psi_\ell(X; g_\ell^\ast)\E\left(\zeta_k^\ast(W)^\top\omega_k(W)\mid X\right)\right] &(\text{Assumption~\ref{ass:constituent_scores} Part 3})\\
&= \E\left[\partial_k \smax^\beta_\ell \psi_\ell(X; g_\ell^\ast)\zeta_k^\ast(W)^\top\omega_k(W)\right] &(\text{Tower rule})\\
&= \E[\alpha_k^\beta(W)^\top \omega_k(W)] &(\text{Definition of $\alpha_k^\beta$}).
\end{align*}

Now, we see that
\begin{align*}
D_{g_k}\E\left[\alpha_k^\beta(W)^\top(U_k - g_k^\ast(W))\right](\omega_k) &= -D_{g_k}\E\left[\alpha_k^\beta(W)^\top g_k^\ast(W)\right](\omega_k) \\
&= -\partial_t \E\left[\alpha_k^\beta(W)^\top(g_k^\ast(W) + t\omega_k(W))\right]\Big|_{t = 0} \\
&= -\partial_t t\E\left[\alpha^\beta_k(W)^\top \omega_k(W)\right]\Big|_{t = 0}\\
&=  -\E[\alpha_k^\beta(W)^\top \omega_k(W)]
\end{align*}
and so the Gateaux derivative of $\E[\Psi^\beta(Z; g^\ast, \alpha^\beta)]$ with respect to each $g_k$ vanishes. Checking the Gateaux derivative with respect to $\alpha_k$ is similar, as we have
\begin{align*}
D_{\alpha_k}\E\left[\Psi^\beta(Z; g^\ast, \alpha^\beta)\right](\omega_k) &= D_{\alpha_k}\E\left[\alpha_k^\beta(W)^\top(U_k - g_k^\ast(W))\right](\omega_k) \\
&= \E\left[\omega_k(W)^\top(U_k - g_k^\ast(W))\right] \\
&= \E\left[\omega_k(W)^\top(\E(U_k \mid W) - g_k^\ast(W))\right] \\
&=0.
\end{align*}
Thus, we have shown that the score is Neyman orthogonal.

\end{proof}

Next, we check our main asymptotic linearity/normality result for smoothed estimates of irregular fucntionals. This boils down to checking the conditions of the smoothed central limit theorem presented in Theorem~\ref{thm:smooth_clt} of Appendix~\ref{app:clt}.

\begin{proof}[Proof of Theorem~\ref{thm:normal_irregular}]
Throughout this proof, we again suppress dependence on the sample size $n$ to make notation less cluttered. First, we show we have
\[
\sqrt{n}(\wh{V} - V^\ast) = \sqrt{n}(\wh{V} - V^{\beta}) + \sqrt{n}\underbrace{(V^{\beta} - V^\ast)}_{\Bias(\beta)} = \sqrt{n}(\wh{V} - V^{\beta}) + o_\P(1).
\]
To show this, it is sufficient to show that $\Bias(\beta) = o(n^{-1/2})$. Under the choice of sequence of $\beta_n$, applying Lemma~\ref{lem:margin} yields, for $\beta$ sufficiently large,
\begin{align*}
\Bias(\beta) &= \left|\E\left[\smax^\beta_\ell(\psi_\ell(X; g_\ell^\ast)) - \max_\ell\{\psi_\ell(X; g_\ell^\ast)\}\right]\right| \\
&\leq \sum_{k = 1}^N\E\left[\Delta_k \exp\left\{-\beta\Delta_k\right\}\right] &(\text{Lemma~\ref{lem:margin}, First Part}) \\
&\lesssim \left(\frac{1}{\beta}\right)^{1 + \delta}  &(\text{Assumption 1 + Lemma~\ref{lem:margin}, Second Part})\\
&= o(n^{-1/2}) &\left(\text{Since }\beta = \omega\left(n^{\frac{1}{2(1 + \delta)}}\right)\right).
\end{align*}
Thus, we just need to show asymptotic linearity off $\sqrt{n}(\wh{V} - V^{\beta})$ to prove the result. It suffices to check the conditions of Theorem~\ref{thm:smooth_clt}.

\paragraph{Checking Theorem~\ref{thm:smooth_clt}, Condition~\ref{cond:simp_neyman}}
The first condition is satisfied as $\Psi^\beta(Z; g, \alpha)$ is Neyman orthogonal by construction (see Proposition~\ref{prop:general_score}).

\paragraph{Checking Theorem~\ref{thm:smooth_clt}, Condition~\ref{cond:simp_hessian}}

Now, let $h := (g, \alpha)$, $\wh{h} := (\wh{g}, \wh{\alpha})$, $h^\beta := (g^\ast, \alpha^\beta)$, and $\wb{h} := (\wb{g}, \wb{\alpha})$ where $\wb{g} \in [\wh{g}, g^\ast]$ and $\wb{\alpha} \in [\wh{\alpha}, \alpha^{\beta}]$. One can compute that
\begin{align*}
D_h^2\E_Z\Psi^{\beta}(Z; \wb{g}, \wb{\alpha})(\wh{h} - h^{\beta}) &= D_{g}^2 \E_X\left[\smax^{\beta}\psi(X; \wb{g}))\right](\wh{g} - g^\ast) \\
&\qquad + 2D_{g, \alpha}\E_Z\left[\sum_{k = 1}^N \wb{\alpha}_k(W)^\top(U_k - \wb{g}_k(W))\right](\wh{h} - h^{\beta}) \\
&=\E_X\left[D_g^2 \smax^\beta \psi(X; \wb{g})(\wh{g} - g^\ast)\right]  \\
&\qquad + 2D_{g, \alpha}\E_Z\left[\sum_{k = 1}^N \wb{\alpha}_k(W)^\top(U_k - \wb{g}_k(W))\right](\wh{h} - h^{\beta}),
\end{align*}
where the interchange of second derivative with expectation above can be justified by an analogous argument to the one used in the proof of Proposition~\ref{prop:general_score}. Thus, it suffices to show that each of the terms on the right hand side are of order $o_\P(n^{-1/2})$. In the sequel, we leverage the following chain rule for second Gateaux derivatives\footnote{Here and throughout, by $\nabla_u^2 F(H(g))$ and $\nabla_u F(G(g))$ we actually mean $\nabla_u^2 F(u) \vert_{u = H(g)}$ and $\nabla_u F(u) \vert_{u = H(g)}$ respectively.}: given a twice differentiable function $F : \R^N \rightarrow \R$ and a twice Gateaux differentiable function $H : L^2(P_W) \rightarrow \R^N$, one has,
\begin{equation}
\label{eq:hessian}
D^2_g (F \circ H)(g)(\omega) = D_g H(g)(\omega)^\top \nabla^2_u F(H(g)) D_g H(g)(\omega) + \nabla_u F(H(g))^\top D_g^2 H(g) (\omega),
\end{equation}
which follows from the same argument used to prove the multivariate chain rule (see \citet{skorski2019chain}). Analyzing the first term, we see we have
\begin{align*}
&D_{g}^2 \E_X\left[\smax^{\beta}_\ell(\psi_\ell(X; \wb{g}_\ell))\right](\wh{g} - g^\ast) = \E_X\left[D_g^2 \smax^{\beta}_\ell \psi_\ell(X; \wb{g}_\ell)(\wh{g} - g^\ast)\right] &(\text{Interchangeability of limits})\\
&\qquad =\underbrace{\E_X\left[D_g \psi(X; \wb{g})(\wh{g} - g^\ast)^\top \nabla^2\smax^{\beta}_\ell \psi_\ell(X; \wb{g}_\ell) D_g \psi(X; \wb{g})(\wh{g} - g^\ast)\right]}_{T_1}  \\
&\qquad\qquad  + \underbrace{ \E_X\left[\nabla_u\smax^{\beta} \psi_\ell(X; \wb{g}_\ell)^\top D_g^2 \psi(X; \wb{g})(\wh{g} - g^\ast)\right]}_{T_2} &(\text{Using Equation~\eqref{eq:hessian}}) \\
&= \underbrace{\E_X\left[D_g \psi(X; \wb{g})(\wh{g} - g^\ast)^\top \nabla^2\smax^{\beta}_\ell \psi_\ell(X; \wb{g}_\ell) D_g \psi(X; \wb{g})(\wh{g} - g^\ast)\right]}_{T_1},
\end{align*}
where the fact that $T_2 = 0$ follows from the assumption that each $\psi_k(X; g_k)$ is affine in $g_k$, which implies $D_{g_k}^2\psi_k(X; \wb{g}_k)(\wh{g}_k - g^\ast) = 0$ for each $k \in [N]$. We now argue that  $T_1 = o_\P(n^{-1/2})$ under the nuisance estimation rates assumed in the theorem. We have:
\begin{align*} 
T_1 &\leq \sup_{u \in \R^N} \|\nabla^2_u \smax^{\beta}_\ell(u)\|_{op}\cdot \E_X\|D_g\psi(X; \wb{g})(\wh{g} - g^\ast)\|_2^2 &(\text{Since } |x^\top A x| \leq \|A\|_{op}\|x\|_2^2)\\
&= \sup_{u \in \R^N}\|\nabla^2_u \smax^\beta_\ell(u)\|_{op} \sum_{k = 1}^N \|D_{g_k}\psi_k(X; \wb{g}_k)(\wh{g}_k - g_k^\ast)\|_{L^2(P_X)}^2\\
&= \sup_{u \in \R^N}\|\nabla_u^2\smax^\beta_\ell(u)\|_{op}\sum_{k = 1}^N \left\|\psi_k(X; \wh{g}_k) - \psi_k(X; g_k^\ast)\right\|_{L^2(P_X)}^2  &(\text{Affinity of $\psi_k$}) \\
&\lesssim \sup_{u \in \R^N}\|\nabla_u^2 \smax^\beta_\ell(u)\|_{op} \sum_{k = 1}^N \|\wh{g}_{k} - g^\ast_k\|_{L^2(P_W)}^2 &(\text{Mean-squared continuity}) \\
&\lesssim \beta\sum_{k = 1}^N\|\wh{g}_{k} - g^\ast_k\|_{L^2(P_W)}^2  &(\text{By Lemma~\ref{lem:softmax} Point 3})\\
&= o_\P(n^{-1/2})  &(\|\wh{g}_{n, k} - g^\ast_k\|_{L^2(P_W)} = o_\P(\beta^{-1/2} n^{-1/4})).
\end{align*}

Likewise, we can check the cross-derivative term. Let $\Delta_{g, k} := \wh{g}_{k} - g_k^\ast$ and $\Delta_{\alpha, k} := \wh{\alpha}_{ k} - \alpha^{\beta}_k$. Analyzing each term in the finite sum yields:
\begin{align*}
&D_{g_k, \alpha_k}\E_Z\left[\wb{\alpha}_k(W)^\top (U_k - \wb{g}_k(W))\right](\Delta_{g, k}, \Delta_{\alpha, k}) \\
&\qquad = -\frac{\partial^2}{\partial s \partial t}\E_Z\left[(\wb{\alpha}_k + t \Delta_{\alpha, k})^\top(\wb{g}_k + s\Delta_{g, k})\right]\Big|_{s,t = 0} \\
&\qquad =-\lim_{s, t \rightarrow 0}\E_Z\left[\frac{(\wb{\alpha}_k + t\Delta_{\alpha, k})^\top(\wb{g}_k + s\Delta_{g, k}) - \wb{\alpha}_k^\top(\wb{g}_k + s\Delta_{g, k}) - (\wb{\alpha}_k + t\Delta_{\alpha, k})^\top \wb{g}_k + \wb{\alpha}_k^\top \wb{g}_k}{st}\right] \\
&\qquad = -\lim_{s, t \rightarrow 0}\E_Z\left[\frac{s t \Delta_{\alpha, k}^\top\Delta_{g, k}}{s t}\right] = -\E_Z[\Delta_{\alpha, k}^\top\Delta_{g, k}],
\end{align*}
where outside of the first line above we have suppressed dependence on $W$ for notational ease. Leveraging this identity, we have that:
\begin{align*}
\left|D_{g_k, \alpha_k}\E_Z\left[\wb{\alpha}_k(W)^\top (U_k - \wb{g}_k(W))\right](\Delta_{g, k}, \Delta_{\alpha, k})\right| &= \left|\langle \wh{\alpha}_{k} - \alpha^{\beta}_k, \wh{g}_{k} - g_k^\ast\rangle_{L^2(P_W)}\right| \\
&\leq \|\wh{\alpha}_{k} - \alpha_{k}^{\beta}\|_{L^2(P_W)}\cdot \|\wh{g}_{k} - g_k^\ast\|_{L^2(P_W)} \\
&= o_\P(n^{-1/2}),
\end{align*}
where the final line follows from the assumed nuisance estimation rates. Thus, we have show that, under the assumptions of Theorem~\ref{thm:normal_irregular}, second order contributions vanish at sufficiently fast rates.

\paragraph{Checking Theorem~\ref{thm:smooth_clt}, Condition~\ref{cond:simp_score}}

We show the stronger result that $\Psi^{\beta}(Z; g^\ast, \alpha^{\beta}) \xrightarrow[n \rightarrow \infty]{} \Psi^\ast(Z; g^\ast, \alpha^\ast)$ almost surely, which in turns proves convergence in probability. First, applying Point~\ref{pt:softmax_lim} of Lemma~\ref{lem:softmax}, we see that
\begin{align*}
\lim_{n \rightarrow \infty}\smax^{\beta}_\ell \psi_\ell(X; g_\ell^\ast) = \lim_{\beta \rightarrow \infty}\smax^\beta_\ell \psi_\ell(X; g_\ell^\ast) = \max_\ell(\psi_\ell(X; g_\ell^\ast)).
\end{align*}
Next, Point~\ref{pt:softmax_grad_lim} of the same lemma implies that, for any $k \in [N]$,
\begin{align*}
\lim_{n \rightarrow \infty}\alpha_k^{\beta}(W)^\top(U_k - g_k^\ast(W)) &= \left(\lim_{\beta \rightarrow \infty}\partial_k \smax^\beta_\ell\psi(X; g_\ell^\ast)\right)\zeta_k^\ast(W)^\top(U_k - g_k^\ast(W)) \\
&= \fraks_k(\psi(X; g^\ast))\zeta_k^\ast(W)^\top(U_k - g_k^\ast(W)).
\end{align*}
The result readily follows from this.

\paragraph{Checking Theorem~\ref{thm:smooth_clt},  Condition~\ref{cond:simp_equi}}

We now check stochastic equi-continuity. In particular, we aim to show:
\[
\left|\G_n \Psi^{\beta}(Z; \wh{g}, \wh{\alpha}) - \G_n \Psi^{\beta}(Z; g^\ast, \alpha^{\beta})\right| = o_\P(1).
\]
To do this, we again let $g, g', \alpha, \alpha'$ be functions that have their $L^\infty(P_W)$ norm bounded by the constant $G$ and we again set $h := (g, \alpha)$ and $h' := (g', \alpha')$. We first show
\[
\E\left[\left(\G_n \Psi^{\beta}(Z; h) - \G_n \Psi^{\beta}(Z; h')\right)^2\right]^{1/2} \lesssim \sum_\ell\|g_\ell - g'_\ell\|_{L^2(P_W)} + \sum_\ell\|\alpha_\ell - \alpha'_\ell\|_{L^2(P_W)}.
\]
For ease of notation, define $f(Z; h, h') := \Psi^{\beta}(Z; h) - \Psi^{\beta}(Z; h')$. First, observe that we have
\begin{align*}
\E\left[\left(\G_n \Psi^{\beta}(Z; h) - \G_n \Psi^{\beta}(Z; h')\right)^2\right] &= \frac{1}{n}\E\left[\left(\sum_{i = 1}^n\{f(Z_i; h, h') - \E f(Z; h, h')\}\right)^2\right] \tag{A}\label{eq:equi-bd}\\
&= \frac{1}{n}\Var\left[\sum_{i = 1}^n f(Z_i; h, h')\right]  \\
&= \Var[f(Z; h, h')] \qquad (\text{Independence of the } Z_i)\\
&\leq \E\left[f(Z; h, h')^2 \right] \qquad (\text{Since } \Var[X] \leq \E[X^2])\\
&= \|\Psi^{\beta}(Z; h) - \Psi^{\beta}(Z; h')\|_{L^2(P_Z)}^2. 
\end{align*}
Next, a repeated application of the parallelogram inequality yields that:
\begin{align*}
\|\Psi^{\beta}(Z; h') - \Psi^{\beta}(Z; h)\|_{L^2(P_Z)}^2 &\lesssim \|\smax^\beta_\ell\psi_\ell(X; g_\ell) - \smax^\beta_\ell \psi_\ell(X; g_\ell')\|_{L^2(P_X)}^2\\
&+\sum_\ell\|\alpha_\ell(W)^\top(U_\ell - g_\ell(W)) -\alpha_\ell'(W)^\top(U_\ell - g_{\ell}'(W)\|_{L^2(P_Z)}^2.  
\end{align*}

\noindent Recalling that we have defined $\psi(X; g) := (\psi_1(X; g_1), \dots, \psi_N(X; g_N))$, the first term can be bounded above by
\begin{align*}
&\|\smax^\beta_\ell \psi_\ell(X; g_\ell) - \smax^\beta_\ell \psi_\ell(X; g_\ell')\|_{L^2(P_X)}^2 \leq L^2\|\psi(X; g) - \psi(X; g')\|_{L^2(P_X)}^2 \\
&\qquad= L^2\sum_\ell \|\psi_\ell(X; g_\ell) - \psi_\ell(X; g_\ell')\|_{L^2(P_X)}^2\\
&\qquad \lesssim \sum_\ell \|g_\ell - g_\ell'\|_{L^2(P_W)} &(\text{Mean-squared continuity}).
\end{align*}
where $L$ is the $\beta$-independent Lipschitz constant for $\smax^\beta$ outlined in Lemma~\ref{lem:softmax}.
An arbitrary index in the second term can likewise be bounded as:
\begin{align*}
&\|\alpha_\ell(W)^\top(U_\ell - g_\ell(W)) - \alpha'_\ell(W)(U_\ell - g_\ell'(W))\|_{L^2(P_Z)}^2 \\
&\qquad = \|\alpha_\ell(W)^\top(U_\ell - g_\ell(W)) \pm \alpha_\ell(W)^\top(U_\ell - g_\ell'(W)) - \alpha'_\ell(W)^\top(U_\ell - g_\ell'(W))\|_{L^2(P_Z)}^2 \\
&\qquad \lesssim \|\alpha_\ell(W)^\top(g_\ell - g_\ell')(W)\|_{L^2(P_W)}^2 + \|(\alpha_\ell - \alpha_\ell')(W)^\top(U_\ell  - g_\ell'(W))\|_{L^2(P_Z)}^2 \\
&\qquad \lesssim \|\alpha_\ell - \alpha_\ell'\|_{L^2(P_W)}^2 + \|g_\ell - g_\ell'\|_{L^2(P_W)}^2,
\end{align*}
where the second to last inequality follows from parallelogram inequality and the final inequality follows from the boundedness assumption on $\alpha, \alpha', g, g',$ and $U_\ell$. This in total yields:
\begin{equation}
\label{eq:nuisance_gap}
\|\Psi^{\beta}(Z; h) - \Psi^{\beta}(Z; h')\|_{L^2(P_Z)}^2 \lesssim \underbrace{\sum_\ell \|g_\ell - g_\ell'\|_{L^2(P_W)}^2 + \sum_\ell \|\alpha_\ell - \alpha_\ell'\|_{L^2(P_W)}^2}_{\Delta(h, h')^2}.
\end{equation}

Now, note that, by the assumption of bounded true nuisances and nuisance estimates (Assumption 5 of the theorem) that $|\Delta(\wh{h}, h^{\beta})^2| \leq B'$ for any $n \geq 1$ where $B' > 0$ is some absolute constant. Consequently, the process $(\Delta(\wh{h}, h^{\beta})^2)_{n \geq 1}$ is uniformly integrable. Since, by the assumption of nuisance consistency, we have $\Delta(\wh{h}, h^{\beta})^2 = o_\P(1)$, we have that $\lim_{n \rightarrow \infty} \E\left[\Delta(\wh{h}, h^{\beta})^2\right] = 0$ by Vitali's Theorem (Proposition~\ref{prop:vitali}). Thus, applying Chebyshev's inequality conditionally on the independent nuisance estimates $(\wh{g}, \wh{\alpha})$ yields that, for any $\delta \in (0, 1)$:
\begin{align*}
\lim_{n 
\rightarrow \infty}\P\left(|\G_n \Psi^{\beta}(Z; \wh{h}) - \G_n \Psi^{\beta}(Z; h^{\beta})| \geq \delta \right) &= \lim_{n \rightarrow \infty}\E\left[\P_Z\left(|\G_n \Psi^{\beta}(Z; \wh{h}) - \G_n \Psi^{\beta}(Z; h^{\beta})| \geq \delta\right)\right] \tag{B}\label{eq:cond-chebyshev} \\
& \lesssim \lim_{n \rightarrow \infty}\frac{1}{\delta^2}\E\left[\sum_\ell \|\wh{g}_\ell - g_\ell^\ast\|_{L^2(P_W)}^2 + \sum_\ell \|\wh{\alpha}_\ell - \alpha_\ell^\beta\|_{L^2(P_W)}^2  \right] \\
& = \lim_{n \rightarrow \infty} \frac{1}{\delta^2}\E\left[\Delta(\wh{h}, h^{\beta})^2\right] = 0.
\end{align*}
This proves the desired stochastic equi-continuity result.

\paragraph{Checking Theorem~\ref{thm:smooth_clt_generic}, Condition~\ref{cond:simp_reg}}
Regularity conditions are straightforward to check. First, observe that we have
\begin{align*}
|V^\beta| \leq \E|\smax^{\beta}_\ell \psi_\ell(X; g_\ell^\ast)| \leq \E|\max_\ell \psi_\ell(X; g_\ell^\ast)| = O(1)
\end{align*}
by assumption. Next, similar reasoning yields that, for any $\epsilon > 0$,
\begin{align*}
\|\Psi^{\beta}(X; g^\ast, \alpha^{\beta})\|_{L^{2 + \epsilon}(P_Z)} &\leq \|\Psi^{\beta}(Z; g^\ast, \alpha^{\beta})\|_{L^\infty(P_Z)} \\
&= \left\|\smax^{\beta} \psi(X; g^\ast) + \sum_\ell \alpha_\ell^\beta(W)^\top(U_\ell - g_\ell^\ast(W))\right\|_{L^\infty(P_Z)} \\
&\leq \left\|\smax^\beta \psi(X; g^\ast)\right\|_{L^\infty(P_Z)} + \sum_\ell \left\|\alpha_\ell^\beta(U_k - g_k^\ast(W))\right\|_{L^\infty(P_Z)} \\
& = O(1) 
\end{align*}
This completes the proof.
\end{proof}

We now prove Corollary~\ref{cor:normal_irregular} by showing an appropriate level of consistency for the plug-in estimate $\wh{\Sigma}_n$ of the population variance $\Sigma^\ast$. In the following proof, we need the following weak law of large numbers for triangular arrays. The following appears as Theorem 2.2.6 in \citep{durrett2019probability}, wherein a proof can be found.

\begin{prop}
\label{prop:wlln}
For each $n \geq 1$, let $X_1^{n}, \dots, X_n^n$ be a sequence of random variables, $S_n := X_1^n + \cdots + X_n^n$, $\sigma_n^2 := \Var[S_n]$, $\mu_n := \E S_n$, and $b_n$ a sequence of constants. Then, if $\sigma^2_n/b_n^2 \rightarrow 0$ as $n \rightarrow \infty$, we have
\[
\frac{S_n - \mu_n}{b_n} \xrightarrow[]{\P} 0.
\]
\end{prop}

\begin{proof}[Proof of Corollary~\ref{cor:normal_irregular}]

Again, we omit dependence on sample size $n \geq 1$ when referring to smoothing parameters $\beta$ or nuisance estimates $\wh{g}$ and $\wh{\alpha}$ in order to streamline notation. Observe that, letting $h = (g, \alpha)$ as above, we can equivalently write
\[
\wh{\Sigma} := \frac{1}{n}\sum_{m}\left\{\Psi^{\beta}(Z; \wh{h}) - \wh{V}\right\}^2 = \frac{1}{n}\sum_m \Psi^{\beta}(Z; \wh{h})^2 - \wh{V}^2.
\]
The conclusion of Theorem~\ref{thm:normal_irregular} immediately implies that $\wh{V}^2 = (V^\ast)^2 + o_\P(1)$ by the continuous mapping theorem, so it suffices to check that
\[
\frac{1}{n}\sum_m \Psi^{\beta}(Z; \wh{h})^2 - \E\left[\Psi^\ast(Z; h^\ast)^2\right] = o_\P(1).
\]
We rewrite the left-hand side of the above expression as
\begin{align*}
\frac{1}{n}\sum_m \Psi^{\beta}(Z; \wh{h})^2 - \E\left[\Psi^\ast(Z; h^\ast)^2\right] &= \underbrace{\frac{1}{n}\sum_m\left\{\Psi^{\beta}(Z; \wh{h})^2 - \Psi^{\beta}(Z; h^{\beta})^2\right\}}_{T_1} + \underbrace{\frac{1}{n}\sum_m \left\{\Psi^{\beta}(Z; h^{\beta})^2 - \Psi^\ast(Z; h^\ast)^2\right\}}_{T_2} \\
&+ \underbrace{\frac{1}{n}\sum_m \left\{\Psi^\ast(Z; h^\ast)^2 - \E\left[\Psi^\ast(Z; h^\ast)^2\right]\right\}}_{T_3}.
\end{align*}
We start by analyzing $T_2$ and $T_3$, as they are both simpler than $T_1$. First, since $T_3$ is an average of i.i.d.\ bounded mean-zero random variables $\Psi^\ast(Z; h^\ast)^2 - \E\left[\Psi^\ast(Z; h^\ast)^2\right]$, the weak law of large numbers implies that $T_3 = o_\P(1)$. 

Next, we use a weak law of large numbers for triangular arrays to argue that $T_2 = o_\P(1)$. Observe that, by the assumption of boundedness, $\sup_{n \geq 1}\Var\left[\Psi^{\beta}(Z; h^{\beta})^2 - \Psi^\ast(Z; h^\ast)^2\right] \leq C$ for some absolute constant $C > 0$, and hence $\Var[T_2] \leq C/n$. Applying the weak law of large numbers for triangular arrays with $S_n = \sum_m \left\{\Psi^{\beta}(Z; h^{\beta})^2 - \Psi^\ast(Z; h^\ast)^2\right\}$ and $b_n = n$ (Proposition~\ref{prop:wlln}) yields that
\[
T_2 - \E\left[\Psi^{\beta}(Z; h^{\beta})^2 - \Psi^\ast(Z; h^\ast)^2\right] = o_\P(1),
\]
so it suffices to show that the above expectation goes towards zero as $n$ grows large. In particular, the bounded convergence theorem alongside the fact that $\Psi^{\beta}(Z; h^{\beta}) \xrightarrow[n \rightarrow \infty]{a.s.} \Psi^\ast(Z; h^\ast)$ yields that
\[
\lim_{n \rightarrow \infty}\E\left[\Psi^{\beta}(Z; h^{\beta})^2 - \Psi^\ast(Z; h^\ast)^2\right] = \E\left[\lim_{n \rightarrow \infty}\Psi^{\beta}(Z; h^{\beta})^2 - \Psi^\ast(Z; h^\ast)^2\right] = 0,
\]
which thus yields that $T_2 = o_\P(1)$.

Now, we argue that $T_1 = o_\P(1)$. Observe that we can rewrite $T_1$ as follows.
\begin{align*}
T_1 &= \frac{1}{n}\sum_m \left\{\Psi^{\beta}(Z; \wh{h}) \pm \Psi^{\beta}(Z; h^{\beta})\right\}^2 - \frac{1}{n}\sum_m \Psi^{\beta}(Z; h^{\beta})^2 \\
&= \frac{1}{n}\sum_m \left\{\Psi^{\beta}(Z; \wh{h}) - \Psi^{\beta}(Z; h^{\beta})\right\}^2 + \frac{2}{n}\sum_m \left\{\Psi^{\beta}(Z; \wh{h}) - \Psi^{\beta}(Z; h^{\beta})\right\}\cdot \Psi^{\beta}(Z; h^{\beta}).
\end{align*}
Thus, applying the Cauchy-Schwarz inequality, we see that we have 
\begin{align*}
|T_1| &\leq \frac{1}{n}\sum_m \left\{\Psi^{\beta}(Z; \wh{h}) - \Psi^{\beta}(Z; h^{\beta})\right\}^2 \\
&\qquad + 2 \left(\frac{1}{n}\sum_m \left\{\Psi^{\beta}(Z; \wh{h}) - \Psi^{\beta}(Z; h^{\beta})\right\}^2\right)^{1/2}\left(\frac{1}{n}\sum_m \Psi^{\beta}(Z; h^{\beta})^2\right)^{1/2}.
\end{align*}
Since $\frac{1}{n}\sum_m \Psi^{\beta}(Z; h^{\beta})^2 = O_\P(1)$ by our analysis of $T_2$, it suffices to show that $n^{-1}\sum_m \left\{\Psi^{\beta}(Z_m; \wh{h}) - \Psi^{\beta}(Z_m; h^{\beta})\right\}^2 = o_\P(1)$ to show $T_1 = o_\P(1)$. First, we note that, in the proof of Theorem~\ref{thm:normal_irregular} above, we showed that
\[
\E_Z\left[\left(\Psi^{\beta}(Z; \wh{h}) - \Psi^{\beta}(Z; h^{\beta})\right)^2 \right] \lesssim \Delta(\wh{h}, h^{\beta})^2 = o_\P(1),
\]
where $\Delta(h, h')$ is as defined in Equation~\eqref{eq:nuisance_gap}. We further argued using uniform integrability that $\lim_{n \rightarrow \infty}\E\left[\Delta(\wh{h}, h^{\beta})^2\right] = 0$. Consequently, applying Markov's inequality, we get, for any $\delta > 0$,
\begin{align*}
\lim_{n \rightarrow \infty}\P\left(\left|\frac{1}{n}\sum_m \left\{\Psi^{\beta}(Z; \wh{h}) - \Psi^{\beta}(Z; h^{\beta})\right\}^2\right| \geq \delta\right) &\leq \frac{1}{\delta}\lim_{n \rightarrow \infty}\E\left[\E_Z\left((\Psi^{\beta}(Z; \wh{h}) - \Psi^{\beta}(Z; h^{\beta}))^2 \right)\right] \\
&\lesssim \frac{1}{\delta} \lim_{n \rightarrow \infty}\E\left[\Delta(\wh{h}, h^{\beta})^2\right] = 0.
\end{align*}
This thus shows that $T_1 = o_\P(1)$, and hence $\wh{\Sigma}_n \xrightarrow[n \rightarrow \infty]{\P} \Sigma^\ast$.

We now show how this suffices to conclude the proof. Observe first that the asymptotic linearity result of Theorem~\ref{thm:normal_irregular} can equivalently be written as 
\[
\sqrt{n}(\Sigma^\ast)^{-1/2}\left(\wh{V} - V^\ast\right) = \frac{1}{\sqrt{n}}\sum_{m} (\Sigma^{\ast})^{-1/2}\left\{\Psi^\ast(Z_m; h^\ast) - V^\ast\right\}+ o_\P(1),
\]
and the normal sum on the right-hand side converges in distribution to a standard normal random variable. To complete the proof, it suffices to show that
\[
\sqrt{n}(\wh{\Sigma}^{-1/2} - (\Sigma^\ast)^{-1/2})(\wh{V} - V^\ast) = o_\P(1).
\]
Observe that $\sqrt{n}(\wh{V} - V^\ast) = O_\P(1)$ by the conclusion of Theorem~\ref{thm:normal_irregular}, and since $\Sigma^\ast > 0$, the continuous mapping theorem alongside the consistency of $\wh{\Sigma}_n$ imply $\wh{\Sigma}_n^{-1/2} - (\Sigma^\ast)^{-1/2} = o_\P(1)$, thus proving the desired result.

\end{proof}

We briefly state a cross-fitting variant of Theorem~\ref{thm:normal_irregular}, which we have alluded to in Remark~\ref{rmk:cross-fit}. The proof of the following corollary is straightforward, but we nonetheless prove it for completeness.

\begin{corollary}
\label{cor:cross_fit}
Suppose the assumptions of Theorem~\ref{thm:normal_irregular} hold alongside the assumptions in Remark~\ref{rmk:cross-fit}. In particular, suppose the nuisance convergence assumptions (Assumptions 4 and 5) hold for all nuisance estimates $\{\wh{g}^{(-k)}_n : k \in [K]\}$ and $\{\wh{\alpha}^{(-k)}_n : k \in [K]\}$. Define the cross-fit estimator $\wh{V}$ as
\[
    \wh{V} := \frac{1}{n}\sum_{k = 1}^K\sum_{i \in \calI_k}\Psi^{\beta}\left(Z_i; \wh{g}^{(-k)}, \wh{\alpha}^{(-k)}\right).
\]
We have the following asymptotic linearity result:
\[
\sqrt{n}(\wh{V} - V^\ast) = \frac{1}{\sqrt{n}}\sum_{m = 1}^n \left\{\Psi^\ast(Z_m; g^\ast, \alpha^\ast) - V^\ast\right\} + o_\P(1).
\]
Consequently, $\sqrt{n}(\wh{V} - V^\ast) \Rightarrow \calN(0, \Sigma^\ast)$ where $\Sigma^\ast = \Var[\Psi^\ast(Z; g^\ast, \alpha^\ast)]$ and $\sqrt{n}\wh{\Sigma}_n^{-1/2}(\wh{V} - V^\ast) \Rightarrow \calN(0, 1)$, where 
\[
\wh{\Sigma}_n := \frac{1}{n}\sum_{k = 1}^K \sum_{i \in \calI_k}\left\{\Psi^{\beta}(Z_i; \wh{g}^{(-k)}, \wh{\alpha}^{(-k)}) - \wh{V}\right\}^2
\].

\end{corollary}

\begin{proof}
By construction $\wh{g}^{(-k)}$ and $\wh{\alpha}^{(-k)}$ are independent of $(Z_i : i \in \calI_k)$. We thus have from the conclusion of Theorem~\ref{thm:normal_irregular} that
\begin{align*}
\frac{1}{\sqrt{n}}\sum_{i \in \calI_k}\Psi^{\beta}\left(Z_i; \wh{g}^{(-k)}, \wh{\alpha}^{(-k)}\right) 
&= \frac{1}{\sqrt{n}}\sum_{i \in \calI_k}  \Psi^{\ast}\left(Z_i; g^\ast, \alpha^\ast\right) + o_\P(1).
\end{align*}
From this it follows that we have 
\[
\sqrt{n}(\wh{V} - V^\ast) = \frac{1}{\sqrt{n}}\sum_{k = 1}^K \sum_{i \in \calI_k}\left\{\Psi^{\ast}\left(Z_i; g^\ast, \alpha^\ast\right) - V^\ast\right\} + o_\P(1),
\]
proving the desired result. Similarly, for showing consistency of the plug-in variance estimate, one has
\begin{align*}
\wh{\Sigma} - \Sigma^\ast = \frac{1}{n}\sum_{k = 1}^K \frac{n}{K}(\wh{\Sigma}^{(k)} - \Sigma^\ast) = \frac{1}{K}\sum_{k = 1}^K o_\P(1) = o_\P(1),
\end{align*}
where we have defined $\wh{\Sigma}^{(k)} := \frac{K}{n}\sum_{i \in \calI_k} \left\{\Psi^{\beta}(Z_i; \wh{g}^{(-k)}, \wh{\alpha}^{(-k)}) - \wh{V}\right\}^2$ and note that $\wh{\Sigma}^{(k)} - \Sigma^\ast = o_\P(1)$ follows from an analogous argument to the one used in the proof of Corollary~\ref{cor:normal_irregular}, since $\wh{V} \xrightarrow[n \rightarrow \infty]{\P} V^\ast$. It thus follows that
\[
\sqrt{n}\wh{\Sigma}^{-1/2}(\wh{V} - V^\ast) = \frac{1}{\sqrt{n}}\sum_{m = 1}^n (\Sigma^{\ast})^{-1/2}\left\{\Psi^\ast(Z_m; g^\ast, \alpha^\ast) - V^\ast\right\} + o_\P(1).
\]
Both asymptotic normality results now readily follow.
\end{proof}

\subsection{Proofs for Applications of Theorem~\ref{thm:normal_irregular}}
\label{app:irregular:app}

We now turn to the proofs relating to the applications of our main theorem on de-biased inference on the point-wise maximum of general affine scores. We start with the problem of constructing asymptotically valid confidence intervals for lower conditional Balke and Pearl bounds.

\begin{proof}[Proof of Corollary~\ref{cor:balke-and-pearl}]

First, we show that the constituent scores satisfy Assumption~\ref{ass:constituent_scores}. Observe that, for any pair $(y, d) \in \{0, 1\}^2$, we have the identity 
\[
q_{yd}^\ast(x, v) = \E[\mathbbm{1}\{Y = y, D = d\} \mid X = x, V = v].
\]
Thus, letting $U_{yd} := \mathbbm{1}\{Y = y, D = d\}$ and $U := (U_{yd} : (y, d) \in \{0, 1\}^2)$, it is clear that we have $q^\ast(x, v) = \E[U \mid X = x, V= v] \in \R^4$. Thus the first part of Assumption~\ref{ass:constituent_scores} is satisfied. The second part of Assumption~\ref{ass:constituent_scores}, which states that each $\psi_\ell$ is affine in its nuisance, here $q^\ast$, is also trivially satisfied. 

Finally, we check the existence of a conditional representer for each score. For any pair $(y, d) \in \{0, 1\}^2$, we have under Assumption~\ref{ass:balke} that
\[
D q_{yd}(X, v)(\omega) = \omega_{yd}(X, v) = \E\left[\frac{\1\{V = v\}}{p^\ast(v \mid X)} \omega_{yd}(X, V) \mid X\right].
\]
Consequently, for $\ell \in [8]$, letting $\zeta^\ast_\ell(X, V)$ be as in Equation~\eqref{eq:balke-pearl-representer}, we have
\begin{align*}
D\psi_\ell(X; q)(\omega) &= \sum_{y, d, v}\sgn_\ell(y, d, v)D q_{yd}(X, v)(\omega) = \sum_{y, d, v}\sgn_\ell(y, d, v)\E\left[\frac{\1\{V = v\}}{p^\ast(v \mid X)} \omega_{yd}(X, V)\mid X \right] \\
&= \sum_{y, d}\left\{\sgn_\ell(y, d, 0)\E\left[\frac{\1\{V = 0\}}{p^\ast(0 \mid X)} \omega_{yd}(X, V)\mid X\right] + \sgn_\ell(y, d, 1)\E\left[\frac{\1\{V = 1\}}{p^\ast(1 \mid X)}\omega_{yd}(X, V) \mid X\right]\right\} \\
&= \E\left[\sum_{y, d}\left\{\sgn_\ell(y, d, 0)\frac{\1\{V = 0\}}{p^\ast(0 \mid X)} \omega_{yd}(X, V) + \sgn_\ell(y, d, 1)\frac{\1\{V = 1\}}{p^\ast(1 \mid X)}\omega_{yd}(X, V) \right\}\mid X\right] \\
&= \E\left[\zeta_\ell^\ast(X, V)^\top \omega(X, V) \mid X\right].
\end{align*}
Thus, Assumption~\ref{ass:balke} is fully satisfied.
    
We now check that the assumptions presented in Assumption~\ref{ass:balke} and the conditions of Corollary~\ref{cor:balke-and-pearl} imply the conditions of Theorem~\ref{thm:normal_irregular}. Conditions 1, 2, and 4 of Theorem~\ref{thm:normal_irregular} are satisfied directly by assumption. Condition 3 (mean-squared continuity) for any $\ell \in [8]$ follows from Assumption~\ref{ass:balke}. In particular, letting $q, q'$ be square-integrable, we have
\begin{align*}
&\left\|\psi_\ell(X; q) - \psi_\ell(X; q')\right\|_{L^2(P_X)}^2 = \E\left[\left(\sum_{(y, d, v)}\sgn_\ell(y, d, v)\left\{q_{yd}(x, v) - q_{yd}'(x, v)\right\} \right)^2\right] \\
&\qquad\lesssim \sum_{\substack{(y, d, v) \\ |\sgn_\ell(y, d, v)| = 1}} \E\left[\left(q_{yd}(X, v) - q_{yd}'(X, v)\right)^2\right]  \qquad\qquad \qquad \qquad \quad(\text{Parallelogram Law})\\ 
&\qquad \lesssim \sum_{\substack{(y, d, v) \\ |\sgn_\ell(y, d, v)| = 1}}\E\left[\frac{\mathbbm{1}\{V = v\}}{p^\ast(v \mid X)}\left\{q_{yd}(X, V) - q_{yd}'(X, V)\right\}^2\right] \quad(\text{Positivity +  Ignorability of } V)\\
&\qquad\lesssim \sum_{\substack{(y, d, v) \\ |\sgn_\ell(y, d, v)| = 1}}\E\left[\left\{q_{yd}(X, V) - q_{yd}'(X, V)\right\}^2\right]  \quad(\text{Strong Positivity})\\
&\qquad\lesssim \|q - q'\|_{L^2(P_W)}^2.
\end{align*}
Analogously, Condition 5 holds as $U_k$ and $q^\ast$ are bounded by definition, positivity ensures $\zeta_\ell^\ast$ and thus $\alpha_\ell^{\beta}$ are bounded, and the boundedness of $\wh{q}$ and $\wh{\alpha}$ are assumed. This completes the proof.
\end{proof}

Next, we prove the asymptotic normality of the smoothed estimate of the $L^1$ calibration error of some estimator $\theta : \calO \rightarrow \R$.

\begin{proof}[Proof of Corollary~\ref{cor:l1-calibration}]

We start by checking that each of the conditions in Assumption~\ref{ass:constituent_scores}. While we note that the scores $\psi_1(X; \chi) := \chi(X) - X$ and $\psi_2(X; \chi) := X - \chi(X)$ don't strictly fall into the affine structure proposed in the assumption, defining $\eta^\ast(X) := \chi^\ast(X) - X$ and $\wh{\eta}(X) := \wh{\chi}(X) - X$, we can instead work with the scores $\psi'_1(X; \eta) := \eta(X)$ and $\psi_2'(X; \eta) := - \eta(X)$ and translate back in the end. First, clearly we have $\eta^\ast(X) = \E[U \mid X]$ where $U := Y - X = Y - \theta(O)$, so the first condition is satisfied. Next, clearly $\psi_1'$ and $\psi_2'$ are affine in the sense of the second condition. Finally, we have the trivial representation $D_\eta \psi_1'(X; \eta)(\omega) = \omega(X)$ and $D_\eta \psi_2'(X; \eta)(\omega) = -\omega(X)$. Thus, Assumption~\ref{ass:constituent_scores} is satisfied.

Now we check that the assumptions of Corollary~\ref{cor:l1-calibration} imply the assumptions of Theorem~\ref{thm:normal_irregular}. 
Again, Conditions 1, 2, and 4 are again implied by the assumptions of Corollary~\ref{cor:l1-calibration}, since any estimation rates on $\chi^\ast$ are equivalent to the same rates on $\eta^\ast$. In the remainder of this section, we just check conditions for the score $\psi_1'$ as the analysis for $\psi_2'$ follows by symmetry. Mean squared continuity (Condition 3) is trivially satisfied with equality. Boundedness (Condition 5) is implied by Assumption~\ref{ass:calibration}, the boundedness assumptions on the nuisances made in Corollary~\ref{cor:l1-calibration}, and the following inequality:
\[
|\chi^\ast(\theta(O))| = |\E[Y \mid \theta(O)]| \leq \E\left[|Y| \mid \theta(O)\right] \leq \|Y\|_{L^\infty(P_Z)}.
\]
Thus, the conditions of the theorem are satisfied, which yields asymptotic normality of $\sqrt{n}(\wh{\Cal}_1(\theta) - \Cal_1(\theta))$ with limiting influence function
\begin{align*}
\Psi^\ast(Z; \eta^\ast, \alpha^\ast) &= \max\{\eta^\ast(\theta(O)), -\eta^\ast(\theta(O))\} + \fraks_1\big(\eta^\ast(\theta(O)), -\eta^\ast(\theta(O))\big)\{Y - \theta(O)\}  \\
&\qquad - \fraks_2\big(\eta^\ast(\theta(O)), -\eta^\ast(\theta(O))\big)\{Y - \theta(O)\} \\
&= |\eta^\ast(\theta(O))| + \sgn\big(\eta^\ast(\theta(O))\big)\{Y - \theta(O)\} \\
&= |\chi^\ast(\theta(O)) - \theta(O)| + \sgn\big(\chi^\ast(\theta(O)) - \theta(O)\big)\{Y - \theta(O)\} \\
&= \sgn\big(\chi^\ast(\theta(O)) - \theta(O)\big)\{Y - \theta(O)\},
\end{align*}
thus proving the desired result.

\end{proof}
\section{Estimating the Optimal Value in Dynamic Treatment Regimes}
\label{app:dynamic}

In this appendix, we consider the more intricate problem of estimating the value of the optimal treatment policy in \textit{dynamic treatment regimes}, i.e.\ when units are sequentially administered treatments. We consider a two stage treatment regime in this section for simplicity, but our results naturally extend to $T$-stage regimes as well. As in Section~\ref{sec:static}, we first focus on a fully non-parametric setting where we perform identification in terms of first and second stage $Q$-functions. In particular, we allow the Q-functions to be arbitrary, bounded functions. We derive a novel Neyman orthogonal score by combining our softmax smoothing strategy with de-biasing results for nested regressions~\citep{chernozhukov2022nested}. Next, we focus on a semi-parametric setting where we assume that the blip effects are linear in a known feature embedding. We provide an identification of the optimal policy value in terms of these blip effects, which ultimately allows us to avoid de-biasing.




We now describe our assumptions on the data generating process. We assume access to $n$ i.i.d.\ observations of the form $Z = (X_1, A_1, X_2, A_2, Y)$ which obey the Markovian causal directed acyclic graph (DAG) outlined in Figure~\ref{fig:cg}. Here, $X_1$ and $X_2$ respectively represent first and second stage covariates, $A_1$ and $A_2$ represent first and second stage categorical treatments (taking  values in the discrete set $[N] = \{1, 2, \dots, N\}$), and $Y$ is an observed outcome. While our Markovian assumption on the data-generating process implies the distribution of $Y$ only depends on the second state and treatment $X_2$ and $A_2$, this is without loss of generality. Namely, we can assume $(X_1, A_1) \subset X_2$ if we desire dependence on the first stage treatment and covariates, i.e.\ that $X_2$ contains all information about the covariates and treatment in the first round. We codify our assumptions on the data in Assumption~\ref{ass:seq} below. As opposed to the static setting, we find it easier to express the assumptions in this section in terms of a structural equation model.

\begin{figure}[H]
\centering
\begin{subfigure}[t]{0.4\textwidth}
    \centering
    \begin{tikzpicture} \large
    \node (S) at (0,0) {$X_1$};
    \node (T1) at (1.5, 1.5) {$A_1$};
    \node (X) at (3,0) {$X_2$};
    \node (T2) at (4.5,1.5) {$A_2$};
    \node  (Y) at (6,0) {$Y$};
    \draw[thick, ->] (S) -- (T1);
    \draw[thick, ->] (S) -- (X);
    \draw[thick, ->] (T1) -- (X);
    \draw[thick, ->] (X) -- (Y);
    \draw[thick, ->] (X) -- (T2);
    \draw[thick, ->] (T2) -- (Y);
    \end{tikzpicture}
    \caption{Causal graph for observed data}
    \label{fig:cg:cg}
\end{subfigure}
\ \ \ \ \ \ 
\begin{subfigure}[t]{0.4\textwidth}
    \centering
    \begin{tikzpicture} \large
    \node (S) at (0,0) {$X_1$};
    \node (T1) at (1, 1.5) {$A_1$};
    \node (pi1) at (2.5, 1.5) {$\pi_1(X_1)$};
    \node (X) at (3.5,0) {$X_2^{(\pi)}$};
    \node (T2) at (4.5,1.5) {$A_2$};
    \node (pi2) at (6,1.5) {$\pi_2(X^{(\pi)}_2)$};
    \node  (Y) at (6.5,0) {$Y^{(\pi)}$};
    \draw[thick, ->] (S) -- (T1);
    \draw[thick, ->] (S) -- (X);
    \draw[thick, ->, dashed] (S) to[bend right] (pi1);
    \draw[thick, ->] (pi1) -- (X);
    \draw[thick, ->] (X) -- (Y);
    \draw[thick, ->] (X) -- (T2);
    \draw[thick, ->, dashed] (X) to[bend right] (pi2);
    \draw[thick, ->] (pi2) -- (Y);
    \end{tikzpicture}
    \caption{Intervention graph (SWIG) of counterfactual outcome under alternative adaptive policy $\pi$}
    \label{fig:cg:swig}
\end{subfigure}
\caption{Causal graphs describing observed data and counterfactual outcomes under alternative policy $\pi$}
\label{fig:cg}
\end{figure}

\begin{assumption}
\label{ass:seq}
We assume access to i.i.d.\ samples $Z \sim P_Z$ of the form $Z = (X_1, A_1, X_2, A_2, Y)$ that are generated according to the following structural equation model.
\begin{align*}
Y &= f_Y(A_2, X_2, \xi_Y) \\
A_2 &= f_{A_2}(X_2, \xi_{A_2}) \\
X_2 &= f_{X_2}(A_1, X_1, \xi_{X_2}) \\
A_1 &= f_{A_1}(X_1, \xi_{A_1}) \\
X_1 &= f_{X_1}(\xi_{X_1}),
\end{align*}
where $\xi_Y, \xi_{A_1}, \xi_{A_2}, \xi_{X_1}, \xi_{X_2}$ are exogenous random variables. Further, defining the propensities $p_1^\ast(a_1 \mid x_1) := \P(A_1 = a_1 \mid X_1 = x_1)$ and $p_2^\ast(a_2 \mid x_2) := \P(A_2 = a_2 \mid X_2 = x_2)$, we assume strong positivity, i.e.\ that there is some $\eta > 0$ such that for any $a_1, a_2 \in [N]$, $x_1 \in \calX_1$, and $x_2 \in \calX_2$, we have
\[
\eta \leq p_1^\ast(a_1 \mid x_1), p_2^\ast(a_2 \mid x_2)
\]
\end{assumption}



In this setting, a policy is a tuple $\pi = (\pi_1, \pi_2)$ where $\pi_1 : \calX_1 \rightarrow [N]$ and $\pi_2 : \calX_2 \rightarrow [N]$, where $N$ denotes the number of actions/treatments. The value of a policy $\pi$ is given by $V^\pi = \E[Y^{(\pi)}]$, which is illustrated in the single world intervention graph (SWIG) also present in Figure~\ref{fig:cg}. In terms of the above structural equation model, this corresponds to replacing $A_1$ in the equation for $X_2$ by $\pi_1(X_1)$ and $A_2$ in the equation for $Y$ by $\pi_2(X_2^{(\pi)})$. Our target of inference is again the value of the optimal dynamic treatment policy, i.e. the quantity
\[
V^\ast := \max_\pi \E[Y^{(\pi)}].
\]

\subsection{Inference in Non-parametric Settings}

We first consider the non-parametric setting, where the Q-functions (which describe the optimal expected future rewards associated with taking any given action in the current state)  are allowed to be arbitrary functions. In the two stage setting, the second and first stage Q-functions are respectively given by
\begin{align*}
Q_2^\ast(a_2, x_2) &:= \E[Y \mid X_2 = x_2, A_2 = a_2] \\
Q_1^\ast(a_1, x_1) &:= \E\left[\max_{\ell \in [N]}Q_2^\ast(\ell, X_2) \mid X_1 = x_1, A_1 = a_1\right].
\end{align*}
In words, the quantity $Q_2^\ast(a_2, x_2)$ gives the expected reward/outcome from playing action $a_2 \in [N]$ in state $x_2 \in \calX_2$, and $Q_1^\ast(a_1, x_1)$ gives the expected reward from playing action $a_1 \in [N]$ in state $x_1 \in \calX_1$ and then playing optimally to maximize the expected outcome after. Under Assumption~\ref{ass:seq}, the value of the optimal dynamic treatment policy is identified in terms of Q-functions as follows:
\begin{equation}
\label{eq:ident_q_func}
V^\ast := \E\left[\max_{\ell\in [N]}Q_1^\ast(\ell, X_1)\right].
\end{equation}
This identification appears to be nearly identical to the setting of a single treatment (Section~\ref{sec:static}), modulo the fact that the first round Q-function $Q_1^\ast$ appears instead of the outcome regression. Naturally, we should be inclined to pursue a similar softmax smoothing and debiasing approach. However, we quickly run into trouble as the first stage Q-function $Q_1^\ast$ is itself defined as a maximum over actions of $Q_2^\ast$, preventing direct debiasing. Thus, we actually need to apply smoothing in two places. First, we need to construct a smooth approximation for the first-stage Q-function in terms of softmax smoothing, i.e.\ we define
\begin{equation}
Q^\beta_1(a_1, x_1) := \E\left[\smax^\beta_{\ell}Q^\ast_2(\ell, X_2) \mid X_1 = x_1, A_1 = a_1\right]
\end{equation}
With the smoothed approximation $Q_1^\beta$ of $Q_1^\ast$, we can then define the smoothed approximation to the value of the optimal treatment policy:
\begin{equation}
V^\beta := \E\left[\smax^\beta_{\ell}Q_1^\beta(\ell, X_1)\right]
\end{equation}
The follow proposition gives a Neyman orthogonal score for the smoothed policy value $V^\beta$.

\begin{prop}
\label{prop:ortho_dynamic}
For any $\beta >0$, consider the score $\Psi^\beta$ defined as 
\begin{align*}
\Psi^\beta(Z; Q_1, Q_2, \alpha_1, \alpha_2) &:= \smax^\beta_k Q_1(k, X_1) + \alpha_1(A_1, X_1)\left(\smax^\beta_{k}Q_2(k, X_2) - Q_1(A_1, X_1)\right) \\
&\qquad + \alpha_2(A_2, X_2)(Y - Q_2(A_2, X_2)),
\end{align*}
where the true nuisances are $Q_2^\ast$ and $Q_1^\beta$ as defined above, and the Riesz representers are $\alpha^\beta_1, \alpha^\beta_2$ given by 
\begin{align*}
\alpha^\beta_1(a_1, x_1) &:= \sum_{\ell = 1}^N  \partial_\ell \smax^\beta_k Q^\beta_1(k, x_1)\cdot \frac{\mathbbm{1}\left\{a_1 = \ell\right\}}{p_1^\ast(\ell \mid x_1)} \\
\alpha^\beta_2(a_2, x_2) &:= \sum_{\ell = 1}^N\E\left[\alpha^\beta_1(A_1, X_1) \mid X_2 = x_2\right]\partial_\ell \smax^\beta_k Q_2^\ast(k, x_2) \cdot \frac{\mathbbm{1}\{a_2 = \ell\}}{p_2^\ast(\ell \mid x_2)}.
\end{align*}
Then, $\Psi^\beta$ is Neyman orthogonal and $V^\beta = \E\left[ \Psi^\beta(Z; Q_1^\beta, Q_2^\ast, \alpha_1^\beta, \alpha_2^\beta)\right]$.
\end{prop}

\begin{proof}
For simplicity, we let $g^\beta = (Q_1^\beta, Q_2^\ast, \alpha_1^\beta, \alpha_2^\beta)$. First, we check that $V^\beta = \E\left[\Psi^\beta(Z; g^\beta)\right]$. To show this, it suffices to argue that $\E\left[\alpha_1^\beta(A_1, X_1)\left(\smax^\beta_{a_2}Q_2^\ast(a_2, X_2) - Q_1^\beta(A_1, X_1)\right) \right] = 0$ and $\E\left[\alpha^\beta_2(A_2, X_2)(Y - Q_2^\ast(A_2, X_2))\right] = 0$. The first term vanishes from the definition of $Q_1^\beta$ after applying the tower rule for conditional expectation, as we have
\begin{align*}
&\E\left[\alpha_1^\beta(A_1, X_1)\left(\smax^\beta_{a_2}Q_2^\ast(a_2, X_2) - Q_1^\beta(A_1, X_1)\right) \right] \\
&\qquad = \E\left[\alpha_1^\beta(A_1, X_1)\left\{\E\left(\smax^\beta_{a_2} Q_2^\ast(a_2, X_2) \mid A_1, X_1\right) - Q_1^\beta(A_1, X_1)\right\}\right] \\
&\qquad = 0.
\end{align*}
Likewise, we have 
\begin{align*}
\E\left[\alpha_2^\beta(A_2, X_2)(Y - Q_2^\ast(A_2, X_2))\right] = \E\left[\alpha_2^\beta(A_2, X_2)(\E(Y \mid A_2, X_2) - Q_2^\ast(A_2, X_2))\right] = 0.
\end{align*}
Now we check the Neyman orthogonality of the score. We do this by taking the Gateaux derivatives of expected score with respect to each nuisance component. First, we evaluate the derivative with respect to the first Q-function $Q_1$. Letting $\Delta \in L^2(P_{W_1})$ be arbitrary, where $W_1 = (A_1, X_1)$, we have
\begin{align*}
D_{Q_1}\E\left[\alpha^\beta_1(A_1, X_1)Q_1^\beta(A_1, X_1)\right](\Delta) &= \frac{\partial}{\partial t} \E\left[\alpha^\beta_1(A_1, X_1)\left\{ Q_1^\beta(A_1, X_1) + t \Delta(A_1, X_1)\right\}\right]\Big|_{t = 0} \\
&= \E\left[\alpha^\beta_1(A_1, X_1)\Delta(A_1, X_1)\right] \\
&= \E\left[\sum_{\ell = 1}^N \partial_\ell \smax_k^\beta Q^\beta_1(k, X_1) \frac{\mathbbm{1}\left\{A_1 = \ell\right\}}{p^\ast_1(\ell \mid X_1)} \Delta(A_1, X_1)\right] \\
&=\E\left[\sum_{\ell = 1}^N \partial_\ell \smax_k^\beta Q^\beta_1(k, X_1) \frac{\mathbbm{1}\left\{A_1 = \ell\right\}}{p^\ast_1(\ell \mid X_1)} \Delta(\ell, X_1)\right] \\
&=\E\left[\sum_{\ell = 1}^N \partial_\ell \smax_k^\beta Q^\beta_1(k, X_1)  \Delta(\ell, X_1)\right].
\end{align*}

We also see that
\begin{align*}
D_{Q_1}\E\left[\smax^\beta_\ell Q^\beta_1(\ell, X_1)\right](\Delta) &= \frac{\partial}{\partial t}\E\left[\smax^\beta_\ell\left\{Q_1^\beta(\ell, X_1) + t \Delta(\ell, X_1)\right\}\right]\Big|_{t = 0} \\
&= \E\left[\frac{\partial}{\partial t} \smax^\beta_\ell \left\{Q_1^\beta(\ell, X_1) + t \Delta(\ell, X_1)\right\}\big|_{t = 0}\right] \\
&= \E\left[\sum_{\ell = 1}^N \partial_\ell \smax_k^\beta Q^\beta_1(k, X_1)  \Delta(\ell, X_1)\right],
\end{align*}
where we can exchange the derivative and expectation using the same argument used in the proof of Proposition~\ref{prop:general_score} in Appendix~\ref{app:irregular}. Thus the derivative with respect to $Q_1$ vanishes. 

We now check the Gateaux derivative with respect to $Q_2$. In particular, we have 
\begin{align*}
&D_{Q_2}\E\left[\alpha^\beta_1(A_1, X_1) \smax^\beta_\ell Q^\ast_2(\ell, X_2)\right](\Delta) = \frac{\partial}{\partial t}\E\left[\alpha^\beta_1(A_1, X_1)\smax^\beta_\ell\left\{Q_2^\ast(\ell, X_2) + t \Delta(\ell, X_2)\right\}\right]\Big|_{t = 0} \\
&\qquad = \E\left[\alpha^\beta_1(A_1, X_1)\frac{\partial}{\partial t} \smax^\beta_\ell\left\{Q_2^\ast(\ell, X_2) + t \Delta(\ell, X_2)\right\}\Big|_{t = 0}\right] \\
&\qquad = \E\left[\alpha^\beta_1(A_1, X_1)\sum_{\ell = 1}^N \partial_\ell \smax^\beta_k \{Q_2^\ast(k, X_2)\} \Delta(k, X_2)\right] \\
&\qquad = \E\left[\E\left(\alpha^\beta_1(A_1, X_1) \mid X_2\right)\sum_{\ell = 1}^N \partial_\ell \smax^\beta_k\{Q_2^\ast(k, X_2)\}\Delta(k, X_2)\right] \\
&\qquad = \E\left[\E\left(\alpha^\beta_1(A_1, X_1) \mid X_2\right) \sum_{\ell = 1}^N\frac{\mathbbm{1}\{A_2 = \ell\}}{p^\ast_2(\ell \mid X_2)}\partial_\ell\smax^\beta_k\{Q_2^\ast(k, X_2)\}\Delta(A_2, X_2) \right] \\
&\qquad = \E\left[\alpha^\beta_2(A_2, X_2)\Delta(A_2, X_2)\right]\\
&\qquad = D_{Q_2}\E\left[\alpha^\beta_2(A_2, X_2)Q_2^\ast(A_2, X_2)\right](\Delta).
\end{align*}
The fourth from last equality follows from applying the tower rule given $X_2$ and the second last equality follows from the definition of $\alpha^\beta_2(A_2, X_2)$. Thus, the Gateaux derivative with respect to $Q_2$ vanishes as well.

We conclude by ensuring the Gateaux derivatives with respect to the Riesz representers vanish. First, we have that
\begin{align*}
&D_{\alpha_1} \E\left[\alpha^\beta_1(A_1, X_1)\left\{\smax^\beta_k Q^\ast_2(k, X_2) - Q_1^\beta(A_1, X_1)\right\}\right](\Delta) \\
&\qquad = \E\left[\Delta(A_1, X_1) \left\{\smax^\beta_k Q^\ast_2(k, X_2) - Q_1^\beta(A_1, X_1)\right\}\right] \\
&\qquad = \E\left[\Delta(A_1, X_1)\left\{\E\left(\smax^\beta_kQ^\ast(k, X_2) \mid A_1, X_1\right) - Q^\beta_1(A_1, X_1)\right\}\right] \\
&\qquad = 0
\end{align*}
where the final equality follows from the definition of $Q^\beta_1$. A similar argument yields that
\begin{align*}
D_{\alpha_2}E\left[\alpha^\beta_2(A_2, X_2)(Y - Q_2^\ast(A_2, X_2)\right](\Delta) &= \E\left[\Delta(A_2, X_2)(Y - Q_2^\ast(A_2, X_2))\right] \\
&= \E\left[\Delta(A_2, X_2)(\E(Y \mid A_2, X_2) - Q_2^\ast(A_2, X_2))\right]\\
&= 0.
\end{align*}

Thus, since we have shown all Gateaux derivatives vanish, we have shown Neyman orthogonality of the score, as desired.
\end{proof}

\begin{theorem}
\label{thm:normal_dynamic} 
Let $\beta_n$ be a smoothing parameter. Let $Z_1, \dots, Z_n$ be i.i.d.\ random variables generated according to Assumption~\ref{ass:seq} (or equivalently, generated per the causal diagram in Figure~\ref{fig:cg}). Additionally, let $\wh{Q}_{1,n}, \wh{Q}_{2, n}, \wh{\alpha}_{1, n},$ and $\wh{\alpha}_{2, n}$ be nuisance estimates for $Q_1^{\beta_n}$, $Q_2^\ast$, $\alpha_1^{\beta_n}$, and $\alpha_2^{\beta_n}$ that are independent of $Z_1, \dots, Z_n$. Suppose the following conditions hold.
\begin{enumerate}
    \item \textit{(Controlled Density Near Zero)} We assume that, for any $k \in [N]$, first and second stage sub-optimality gaps $\Delta_{1, k} := \max_{\ell \in [N]}Q^\ast_1(\ell, X_1) - Q^\ast_1(k, X_1)$ and $\Delta_{2, k} := \max_{\ell \in [N]}Q^\ast_2(\ell, X_2) - Q_2^\ast(k, X_2)$ satisfy Assumption~\ref{ass:margin} with constants $c, H, \delta > 0$.
    \item \textit{(Growing Smoothing Parameters)} We assume the smoothing parameters $\beta_n$ satisfies $\beta_n = \omega\left(n^{\frac{1}{2(1 + \delta)}}\right)$.
    \item \textit{(Nuisance Convergence)} We assume the nuisance estimates converge at sufficiently fast rates, in the particular that
    \begin{itemize}
    \item $\|\wh{Q}_{1, n} - Q^{\beta_n}_1\|_{L^2(P_{W_1})},\;\; \|\wh{Q}_{2, n} - Q^\ast_2\|_{L^2(P_{W_2})} = o_\P\left(n^{-1/4}\beta_n^{-1/2}\right)$,
    \item $\|\wh{\alpha}_{1, n} - \alpha^{\beta_n}_1\|_{L^2(P_{W_1})}\cdot\|\wh{Q}_{n, 1} - Q^{\beta_n}_1\|_{L^2(P_{W_1})} = o_\P(n^{-1/2})$, and 
    \item $\|\wh{\alpha}_{2, n} - \alpha^{\beta_n}_2\|_{L^2(P_{W_2})}\cdot \|\wh{Q}_{2, n} - Q^\ast_2\|_{L^2(P_{W_2})} = o_\P(n^{-1/2})$.
    \end{itemize}
    Further, we assume all nuisance estimates are consistent in $L^2$.
    \item \textit{(Boundedness)} There is some universal constant $G > 0$ such that all nuisances, nuisance estimates, and the outcome $Y$ are bounded by $G$ in absolute value almost surely.
    \end{enumerate}
    Then, letting $\wh{V}_n := \P_n \Psi^\beta\big(Z; \wh{Q}_1, \wh{Q}_2, \wh{\alpha}_1, \wh{\alpha}_2\big)$ we have the following asymptotic linearity result:
    \[
    \sqrt{n}(\wh{V}_n - V^\ast) = \frac{1}{\sqrt{n}}\sum_{i = 1}^n \rho_V(Z_i) + o_\P(1),
    \]
    where $\rho_V(Z_i) := \Psi^\ast(Z_i; Q_1^\ast, Q_2^\ast, \alpha_1^\ast, \alpha_2^\ast) - V^\ast$, $\Psi^\ast$ is given by
    \begin{align*}
    \Psi^\ast(Z; Q_1, Q_2, \alpha_1, \alpha_2) &= \max_{\ell}Q_1(\ell, X_1)  + \alpha_1(A_1, X_1)(\max_\ell Q_2(\ell, X_2)  -  Q_1(A_1, X_1)) \\
    &\qquad + \alpha_2(A_2, X_2)(Y - Q_2(A_2, X_2)),
    \end{align*}
    and the limiting nuisances are $Q_1^\ast, Q_2^\ast$  and representers $\alpha_1^\ast, \alpha_2^\ast$ given respectively by 
    \begin{align*}
    \alpha_1^\ast(A_1, X_1) &:= \sum_{\ell = 1}^N \fraks_\ell(Q_1^\ast(\cdot, X_1))\frac{\mathbbm{1}\{A_k = \ell\}}{p^\ast_1(\ell \mid X_1)}\\
    \alpha_2^\ast(A_2, X_2) &:= \E[\alpha_1^\ast(A_1, X_1) \mid X_2]\sum_{\ell = 1}^N \fraks_\ell(Q_2^\ast(\cdot, X_2))\frac{\mathbbm{1}\{A_2 = \ell\}}{p^\ast_2(\ell \mid X_2)}.
    \end{align*}
    Consequently, if $\Var[\rho_V(Z)] > 0$, we have the following asymptotic normality result:
    \[
    \sqrt{n}(\wh{V}_n - V_\ast) \Rightarrow \calN(0, \Var[\rho_V(Z)]).
    \]
\end{theorem}

\begin{proof}
The proof of this result is largely the same as the proof of Theorem~\ref{thm:normal_irregular}. The main complicating factor is that we must take extra care to handle the bias introduced from two applications of softmax smoothing --- one from smoothing the first round Q-function and the other from smoothing the maximum used in defining the optimal policy value. To account for the additional bias introduced from double smoothing, we define the ``naively smoothed'' parameter
\[
V^\beta_\ast := \E\left[\smax^\beta_k Q^\ast_1(k, X_1)\right].
\]
Using the $V^\beta_\ast$, we can bound the bias above as
\[
\Bias(\beta) := |V^\beta - V^\ast| \leq |V^\beta - V^\beta_\ast| + |V^\beta_\ast - V^\ast|.
\]
First, we bound the difference $|V^\beta - V^\beta_\ast|$. In what follows, we let $Q_1^\beta(\cdot, X_1) := (Q_1^\beta(1, X_1), \dots, Q_1^\beta(N, X_1))$ be the vector of partial evaluations of $Q_1^\beta$ and define $Q_1^\ast(\cdot, X)$ similarly. Since the softmax operator is Lipschitz in a point-wise sense (Lemma~\ref{lem:softmax} Point~\ref{pt:softmax_lip}), we have
\begin{align*}
|V^\beta - V^\beta_\ast| &= \left|\E\left[\smax^\beta_k Q_1^\beta(k, X_1) - \smax^\beta_k Q_1^\ast(k, X_1)\right]\right| \\
&\leq \E\left|\smax^\beta_k Q_1^\beta(k, X_1) - \smax^\beta_k Q_1^\ast(k, X_1)\right| \\
&\lesssim \E\left\|Q_1^\beta(\cdot, X_1) - Q^\ast_1(\cdot, X_1)\right\|_2 &(\text{Lemma~\ref{lem:softmax} Point~\ref{pt:softmax_lip}})\\
&\leq \E\left\|Q_1^\beta(\cdot, X_1) - Q_1^\ast(\cdot, X_1)\right\|_1 &(\|x\|_2 \leq \|x\|_1) \\
&\leq \sum_{\ell = 1}^N \E\left|Q_1^\beta(\ell, X_1) - Q_1^\ast(\ell, X_1)\right| &(\text{Triangle Inequality})\\
&\lesssim \E\left|Q_1^\beta(A_1, X_1) - Q_1^\ast(A_1, X_1)\right| &(\text{Strong Positivity})\\
&= \E\left|\E\left[\smax^\beta_k Q_2^\ast(k, X_2) - \max_k Q_2^\ast(k, X_2) \mid A_1, X_1\right]\right| &(\text{Definitions of }Q_1^\ast, Q_1^\beta) \\
&\leq \E\left|\smax^\beta_k Q_2^\ast(k, X_2) - \max_k Q_2^\ast(k, X_2)\right| &(\text{Jensen's Inequality + Tower Rule})\\
&\leq \E\left[\sum_{k = 1}^N \Delta_{2, k} \exp\left\{-\beta \Delta_{2, k}\right\}\right] &(\text{Lemma~\ref{lem:margin}}) \\
&\lesssim \left(\frac{1}{\beta}\right)^{1 + \delta} &(\text{Also Lemma~\ref{lem:margin}}),
\end{align*}
where the final inequality holds for $n  \geq 1$ large enough such that $\beta_n \geq \beta^\ast$, where $\beta^\ast$ is some value that depends only on the parameters $c, H, \delta$.

Next, we can bound the other bias term, $|V^\beta_\ast - V^\ast|$. We have 
\begin{align*}
|V^\beta_\ast - V^\ast| &= \left|\E\left[\smax^\beta_k Q_1^\ast(k, X_1) - \max_k Q_1^\ast(k, X_1)\right]\right| \\
&\leq \E\left|\smax^\beta_k Q_1^\ast(k, X_1) - \max_k Q_1^\ast(k, X_1)\right| \\
&\leq \E\left[\sum_{k = 1}^N \Delta_{1, k}\exp\left\{-\beta\Delta_{1,k}\right\}\right] &(\text{Lemma~\ref{lem:margin}}) \\
&\lesssim \left(\frac{1}{\beta}\right)^{1 + \delta} &(\text{Also Lemma~\ref{lem:margin}}).
\end{align*}
Thus, since we have bounded both terms in the upper bound of the bias, we have 
\[
\Bias(\beta) \lesssim \left(\frac{1}{\beta}\right)^{1 + \delta}.
\]
Now, we have 
\begin{align*}
\sqrt{n}(\wh{V}_n - V^\ast) &= \sqrt{n}(\wh{V}_n - V^\beta) + \sqrt{n}(V^\beta - V^\ast) \\
&= \sqrt{n}(\wh{V}_n - V^\beta) + o_\P(1),
\end{align*}
where the final line follows because
\begin{align*}
\sqrt{n}|V^\ast - V^\beta| = \sqrt{n}\Bias(\beta) \lesssim \sqrt{n}\left(\frac{1}{\beta}\right)^{1 + \delta} = o_\P(1),
\end{align*}
which follows from our above bound on $\Bias(\beta)$ plus our assumption that $\beta \equiv \beta_n = \omega\left(n^{\frac{1}{2(1 + \delta)}}\right)$.

Thus, to conclude the proof, we just need to show that
\[
\sqrt{n}(\wh{V}_n - V^\beta) = \frac{1}{\sqrt{n}}\sum_{i = 1}^n \left\{\Psi^\ast(Z_i; Q_1^\ast, Q_2^\ast, \alpha_1^\ast, \alpha^\ast) - V^\ast\right\} + o_\P(1),
\]
which will follow from checking the relevant conditions of Theorem~\ref{thm:smooth_clt}.

\paragraph{Checking Theorem~\ref{thm:smooth_clt}, Condition~\ref{cond:simp_neyman}}

We have already showed this condition holds in Proposition~\ref{prop:ortho_dynamic}.

\paragraph{Checking Theorem~\ref{thm:smooth_clt}, Condition~\ref{cond:simp_hessian}}
Letting $\wh{g} := (\wh{Q}_1, \wh{Q}_2, \wh{\alpha}_1, \wh{\alpha}_2)$, $g^\beta := (Q_1^\beta, Q_2^\ast, \alpha^\beta_1, \alpha^\beta_2)$, and $\wb{g} \in [g^\beta, \wh{g}]$ be arbitrary, we see we have 
\begin{align*}
D_g^2\E_Z\left[\Psi^\ast(Z; \wb{g})\right](\wh{g} - g^\beta) &= \underbrace{D_{Q_1}^2\E_Z\left[\smax^{\beta}_k \wb{Q}_1(k, X_1)\right](\wh{Q}_1 - Q_1^\beta)}_{=: T_1} \\
&\qquad - 2 \underbrace{D_{Q_1, \alpha_1}\E_Z\left[\wb{\alpha}_1(A_1, X_1) \wb{Q}_1(A_1, X_1)\right](\wh{Q}_1 - Q_1^\beta, \wh{\alpha}_1 - \alpha^\beta_1)}_{=: T_2} \\
&\qquad + \underbrace{D_{Q_2}^2\E_Z\left[\wb{\alpha}_1(A_1, X_1)\smax^\beta_k \wb{Q}_2(k, X_2)\right](\wh{Q}_2 - Q_2^\ast)}_{=:T_3} \\
&\qquad - 2\underbrace{D_{Q_2, \alpha_2}\E_Z\left[\wb{\alpha}_2(A_2, X_2)\wb{Q}_2(A_2, X_2)\right](\wh{Q}_2 - Q_2^\ast, \wh{\alpha}_2 - \alpha^\beta_2)}_{=:T_4}\\
&\qquad + 2\underbrace{D_{Q_2, \alpha_1}\E_Z\left[\wb{\alpha}_1(A_1, X_1)\smax^\beta_k \wb{Q}_2(k, X_2)\right](\wh{Q}_2 - Q_2^\ast, \wh{\alpha}_1 - \alpha_1)}_{=: T_5}
\end{align*}
We now show that each of the terms in the right hand side of the above expression ($T_1, T_2, T_3, T_4$, and $T_5$) are $o_\P(n^{-1/2})$ under the assumptions of the theorem. First, observe that 
\begin{align*}
T_1 &= \frac{\partial^2}{\partial t^2}\E_Z\left[\smax^\beta_k \left\{\wb{Q}_1(k, X_1) + t(\wh{Q}_1 - Q_1^\beta)(k, X_1)) \right\}\right]\Big|_{t = 0} \\
&= \E_Z\left[\frac{\partial^2}{\partial t^2}\smax^\beta_k \left\{\wb{Q}_1(k, X_1) + t (\wh{Q}_1 - Q_1^\beta)(k, X_1))\right\}\Big|_{t = 0}\right] \\
&= \E_Z\left[(\wh{Q}_1(\cdot, X_1) - Q_1^\beta(\cdot, X_1))^\top \nabla_u^2 \smax_k^\beta\{u\}\Big|_{u = \wb{Q}_1(\cdot, X_1)}(\wh{Q}_1(\cdot, X_1) - Q_1^\beta(\cdot, X_1))\right] \\
&\leq \sup_u\|\nabla_u^2 \smax^\beta_k u\|_{op}\cdot\E_Z\left\|\wh{Q}_1(\cdot, X_1) - Q_1^\beta(\cdot, X_1)\right\|_2^2 \hspace{3.5cm}(|x^\top A x| \leq \|A\|_{op}\|x\|_2^2) \\
&\lesssim \beta \E_Z\left\|\wh{Q}_1(\cdot, X_1) - Q_1^\beta(\cdot, X_1)\right\|_2^2 \hspace{6cm}(\text{Lemma~\ref{lem:softmax} Point~\ref{pt:softmax_hess}}) \\
&= \beta \sum_{\ell = 1}^N\E_Z\left[(\wh{Q}_1(\ell, X_1) - Q_1^\beta(\ell, X_1))^2\right]\\
&\lesssim \beta\E_Z\left[(\wh{Q}_1(A_1, X_1) - Q_1^\beta(A_1, X_1))^2\right]\hspace{5cm}(\text{Strong Positivity}) \\
&=  \beta \|\wh{Q}_1 - Q_1^\beta\|_{L^2(P_{W_1})}^2 \\
&= o_\P(n^{-1/2}) \hspace{7cm} (\text{Since } \|\wh{Q}_1 - Q_1^\beta\|_{L^2(P_{W_1})} = o_\P(n^{-1/4}\beta_n^{-1/2}))
\end{align*}
In the above, the interchange of the second derivative and expectation is justified by the assumption that $\wh{Q}_1$ and $Q_1^\beta$ (and hence $\wb{Q}_1$) are bounded by some absolute constant alongside the fact that the first two derivatives of the softmax function are uniformly bounded (Lemma~\ref{lem:softmax} Points~\ref{pt:softmax_grad} and \ref{pt:softmax_hess}).

An analogous argument can be used to control the growth of $T_3$. Namely, we have
\begin{align*}
T_3 &= \frac{\partial^2}{\partial t^2}\E_Z\left[\wb{\alpha}_1(A_1, X_1)\smax^\beta_k\left\{ \wb{Q}_2(k, X_2) + t(\wh{Q}_2 - Q_2^\ast)(k, X_2)\right\}\right]\Big|_{t = 0} \\
&= \E_Z\left[\wb{\alpha}_1(A_1, X_1)(\wh{Q}_2(\cdot, X_2) - Q_2^\ast(\cdot, X_2))^\top \nabla_u^2 \smax_k^\beta\{u\}\Big|_{u = \wb{Q}_2(\cdot, X_2)}(\wh{Q}_2(\cdot, X_2) - Q_2^\ast(\cdot, X_2))\right]\\
&\leq \|\wb{\alpha}_1\|_{L^\infty(P_{W_1})}\E_Z\left[(\wh{Q}_2(\cdot, X_2) - Q_2^\ast(\cdot, X_2))^\top \nabla_u^2 \smax_k^\beta\{u\}\Big|_{u = \wb{Q}_2(\cdot, X_2)}(\wh{Q}_2(\cdot, X_2) - Q_2^\ast(\cdot, X_2))\right]\\
&\lesssim \beta\|\wb{\alpha}_1\|_{L^\infty(P_{W_1})} \E_Z\left\|\wh{Q}_2(\cdot, X_2) - Q_2^\ast(\cdot, X_2)\right\|_2^2 \\
&\lesssim \beta \E_Z\left[(\wh{Q}_2(A_2, X_2) - Q_2^\ast(A_2, X_2))^2\right] \\
&= o_\P(n^{-1/2}) \hspace{5cm} (\text{Since } \|\wh{Q}_2 - Q_2^\ast\|_{L^2(P_{W_2})} = o_\P(n^{-1/4}\beta^{-1/2})),
\end{align*}
where the first inequality follows from using the $L^1/L^\infty$ Holder's inequality, the second inequality follows from the uniform boundedness of the operator norm of the softmax operator (again, Lemma~\ref{lem:softmax} Point~\ref{pt:softmax_hess}), and the second to last line follows from strong positivity and the uniform boundedness of $Q_2^\ast$, which implies $\wb{\alpha}_1$ is bounded by an absolute constant independent of $\beta$.

We now check the cross-derivatives. We start by bounding the absolute value of $T_2$. We have:
\begin{align*}
|T_2| &= \left|D_{Q_1, \alpha_1}\E_Z\left[\wb{\alpha}_1(A_1, X_1)\wb{Q}_1(A_1, X_1)\right](\wh{Q}_1 - Q_1^\beta, \wh{\alpha}_1 - \alpha^\beta_1)\right| \\
&= \left|\frac{\partial^2}{\partial s \partial t}\E_Z\left[(\wb{\alpha}_1 + t(\wh{\alpha}_1 - \alpha_1^\beta))(\wb{Q}_1 + s(\wh{Q}_1 - Q^\beta_1)\right]\Big|_{s, t = 0}\right| \\
&= \left|\E_Z\left[(\wh{\alpha}_1 - \alpha_1^\beta)(A_1, X_1)(\wh{Q}_1 - Q_1^\beta)(A_1, X_1)\right]\right| \\
&= \left|\langle \wh{\alpha}_1 - \alpha_1^\beta, \wh{Q}_1 - Q_1^\beta \rangle_{L^2(P_{W_1})}\right| \\
&\leq \left\|\wh{\alpha}_1 - \alpha^\beta_1\right\|_{L^2(P_{W_1})}\cdot \left\|\wh{Q}_1 - Q_1^\beta\right\|_{L^2(P_{W_1})} \\
&= o_\P(n^{-1/2}),
\end{align*}
where the final line follows from the assumption made on nuisance convergence rates. An analogous argument allows us to show that 
\begin{align*}
|T_4| &= \left|\langle \wh{\alpha}_2 - \alpha_2^\beta, \wh{Q}_2 - Q_2^\ast\rangle_{L^2(P_{W_2})}\right| \leq \left\|\wh{\alpha} - \alpha^\beta_2\right\|_{L^2(P_{W_2})} \cdot \left\|\wh{Q}_2 - Q_2^\ast\right\|_{L^2(P_{W_2})} = o_\P(n^{-1/2}).
\end{align*}
Finally, we analyze $T_5$. We have 
\begin{align*}
|T_5| &= \left|D_{Q_2, \alpha_1}\E_Z\left[\wb{\alpha}_1(A_1, X_2)\smax^\beta_k \wb{Q}_2(k, X_2)\right](\wh{Q}_2 - Q_2^\ast, \wh{\alpha}_1 - \alpha_1^\beta)\right| \\
&= \left|\frac{\partial^2}{\partial s \partial t}\E_Z\left[(\wb{\alpha}_1(A_1, X_1) + s(\wh{\alpha}_1 - \alpha_1^\beta)(A_1, X_1))\cdot\smax^\beta_k\left\{ \wb{Q}_2(k, X_2) + t(\wh{Q}_2 - Q_2^\ast)(k, X_2)\right\}\right]\Big|_{s, t = 0}\right| \\
&= \left|\E_Z\left[\frac{\partial^2}{\partial s \partial t}(\wb{\alpha}_1(A_1, X_1) + s(\wh{\alpha}_1 - \alpha_1^\beta)(A_1, X_1))\cdot\smax^\beta_k\left\{ \wb{Q}_2(k, X_2) + t(\wh{Q}_2 - Q_2^\ast)(k, X_2)\right\}\Big|_{s, t = 0}\right]\right| \\
&= \left|\E_Z\left[(\wh{\alpha}_1 - \alpha_1^\beta)(A_1, X_1)\sum_{k = 1}^N(\wh{Q}_2 - Q_2^\ast)(k, X_2)\partial_k \smax^\beta_\ell \wb{Q}_2(\ell, X_2)\right]\right| \\
&\leq \E_Z\left|(\wh{\alpha}_1 - \alpha_1^\beta)(A_1, X_1)\sum_{k = 1}^N(\wh{Q}_2 - Q_2^\ast)(k, X_2)\partial_k \smax^\beta_\ell \wb{Q}_2(\ell, X_2)\right| \qquad \qquad \qquad \qquad (\text{Jensen's Inequality})\\
&\leq \sum_{k= 1}^N\sup_u |\partial_k \smax^\beta_\ell \{u\}| \cdot \E_Z\left|(\wh{\alpha}_1 - \alpha_1^\beta)(A_1, X_1) (\wh{Q}_2 - Q_2^\ast)(k, X_2)\right| \qquad \qquad \qquad \qquad (\text{Triangle Inequality}) \\
&\leq  \sum_{k= 1}^N\sup_u |\partial_k \smax^\beta_\ell \{u\}| \cdot \|\wh{\alpha}_1 - \alpha_1^\beta\|_{L^2(P_{W_1})}\cdot \|\wh{Q}_2(k, X_2) - Q_2^\ast(k, X_2)\|_{L^2(P_{X_2})} \hspace{1.8cm}
(\text{Cauchy-Schwarz}) \\
&\lesssim \sum_{k = 1}^N \sup_u |\partial_k \smax^\beta_\ell \{u\}| \cdot \|\wh{\alpha}_1 - \alpha_1^\beta\|_{L^2(P_{W_1})}\cdot \|\wh{Q}_2 - Q_2^\ast\|_{L^2(P_{W_2})} \hspace{3cm} (\text{Strong Positivity}) \\
&\leq N \sup_u \|\nabla_u \smax^\beta \{u\}\|_{\infty} \|\wh{\alpha}_1 - \alpha_1^\beta\|_{L^2(P_{W_1})} \cdot \|\wh{Q}_2 - Q_2^\ast\|{L^2(P_{W_2})} \\
&\lesssim \|\wh{\alpha}_1 - \alpha_1^\beta\|_{L^2(P_{W_1})} \cdot \|\wh{Q}_2 - Q^\ast\|_{L^2(P_{W_2})} \\
&= o_\P(n^{-1/2}).
\end{align*}
In the above, the second to last inequality follows since the maximum coordinate of the gradient of the softmax operator is bounded by a constant that does not depend on $\beta$ (Lemma~\ref{lem:softmax} Point~\ref{pt:softmax_grad}) and the final equality follows from our assumption on nuisance convergence rates.

\paragraph{Checking Theorem~\ref{thm:smooth_clt}, Condition~\ref{cond:simp_score}} We show the stronger result that $\Psi^\beta(Z; g^\beta) \rightarrow \Psi^\ast(Z; g^\ast)$ almost surely, where we have let $g^\ast := (Q_1^\ast, Q_2^\ast, \alpha_1^\ast, \alpha_2^\ast)$. First, we argue that all $\beta$-dependent nuisance components converge almost surely to their natural un-smoothed analogues. First, observe that
\begin{align*}
\lim_{n \rightarrow \infty}Q_1^\beta(A_1, X_1) &= \lim_{n \rightarrow \infty}\E\left[\smax^\beta_k Q_2^\ast(k, X_2)  \mid A_1, X_1\right] \\
&=\E\left[\lim_{n \rightarrow \infty}\smax^\beta_k Q_2^\ast(k, X_2) \mid A_1, X_1\right] &(\text{Bounded Convergence Theorem}) \\
&= \E\left[\max_k Q_2^\ast(k, X_2) \mid A_1, X_1\right]  \\
&= Q_1^\ast(A_1, X_1),
\end{align*}
where the second to last equality follows from Lemma~\ref{lem:softmax} Point~\ref{pt:softmax_lim}. Next, observe that
\begin{align*}
\lim_{n \rightarrow \infty}\alpha_1^\beta(A_1, X_1) &= \lim_{n \rightarrow \infty}\sum_{k = 1}^N \frac{\mathbbm{1}\{A_1 = k\}}{p^\ast(k \mid X_1)}\partial_k \smax^\beta_\ell Q_1^\beta(\ell, X_1) \\
&= \sum_{k = 1}^N \frac{\mathbbm{1}\{A_1 = k\}}{p^\ast(k \mid X_1)}\lim_{n \rightarrow \infty}\partial_k \smax^\beta_\ell Q_1^\beta(\ell, X_1) \\
&= \sum_{k = 1}^N \frac{\mathbbm{1}\{A_1 = k\}}{p^\ast(k \mid X_1)}\fraks_k( Q_1^\ast(\cdot, X_1)) \\
&= \alpha_1^\ast(A_1, X_1),
\end{align*}
where the second to last equality follows from the fact that $\lim_{\beta \rightarrow \infty}\partial_k \smax^\beta_\ell\{u\} = \fraks_k(u)$ (Lemma~\ref{lem:softmax} Point~\ref{pt:softmax_grad_lim}) and the fact that $Q_1^\beta(A_1, X_1) \xrightarrow[n \rightarrow \infty]{a.s.} Q_1^\ast(A_1, X_1)$. Lastly, observe that 
\begin{align*}
\lim_{n \rightarrow \infty}\alpha_2^\beta(A_2, X_2) &= \lim_{n \rightarrow \infty}\E\left[\alpha_1^\beta(A_1, X_1) \mid X_2\right]\sum_{k = 1}^N \frac{\mathbbm{1}\{A_2 = k\}}{p_2^\ast(k \mid X_2)}\partial_k\smax^\beta_\ell Q_2^\ast(\ell, X_2) \\
&= \left(\lim_{n \rightarrow \infty} \E\left[\alpha_1^\beta(A_1, X_1) \mid X_2\right]\right) \cdot \left(\sum_{k = 1}^N \frac{\mathbbm{1}\{A_k = k\}}{p^\ast_2(k \mid X_2)}\lim_{n \rightarrow \infty}\partial_k\smax^\beta_\ell Q_2^\ast(\ell, X_2)\right) \\
&= \E\left[\alpha_1^\ast(A_1, X_1) \mid X_2\right]\sum_{k = 1}^N\frac{\mathbbm{1}\{A_2 = k\}}{p_2^\ast(k \mid X_2)}\fraks_k(Q_2^\ast(\cdot, X_2)) \\
&= \alpha_2^\ast(A_2, X_2),
\end{align*}
where we can interchange the limit and expectation in the second equality by applying the bounded convergence theorem and we again invoke Lemma~\ref{lem:softmax} Point~\ref{pt:softmax_grad_lim} in computing the limit of the partial derivative of the softmax function. 

To show almost sure convergence of the scores, it now suffices to show that $\smax^\beta_k Q_1^\beta(k, X_1) \xrightarrow[n \rightarrow \infty]{a.s.} \max_k Q_1^\ast(k, X_1)$ and that $\smax^\beta_k Q_2^\ast(k, X_2) \xrightarrow[n \rightarrow \infty]{a.s.} \max_k Q_2^\ast(k, X_2)$. The latter directly follows from Lemma~\ref{lem:softmax} Point~\ref{pt:softmax_lim} and the former follows from applying the same point of the lemma alongside the earlier shown result that $Q_1^\beta(A_1, X_1) \xrightarrow[n \rightarrow \infty]{a.s.} Q_1^\ast(A_1, X_1)$. 

\paragraph{Checking Theorem~\ref{thm:smooth_clt}, Condition~\ref{cond:simp_equi}}
Again, we proceed as in the proof of Theorem~\ref{thm:normal_irregular}. Let $g = (Q_1, Q_2, \alpha_1, \alpha_2)$ and $g' = (Q_1', Q_2', \alpha_1', \alpha_2')$ be vectors of nuisances with $L^\infty$ norm bounded by the constant $G > 0$ mentioned in Condition 4 of Theorem~\ref{thm:normal_dynamic}. We have the inequality
\begin{align*}
\E\left|\G_n \Psi^\beta(Z; g) - \G_n \Psi^\beta(Z; g')\right|^2 &= \frac{1}{n}\sum_{i = 1}^n \Var[\Psi^\beta(Z; g) - \Psi^\beta(Z; g')] \\
&\leq \frac{1}{n}\sum_{i = 1}^n \E\left[\left(\Psi^\beta(Z_i; g) - \Psi^\beta(Z_i; g')\right)^2\right] \\
&= \E\left[\left(\Psi^\beta(Z; g) - \Psi^\beta(Z; g')\right)^2\right] \\
&= \left\|\Psi^\beta(Z; g) - \Psi^\beta(Z; g')\right\|_{L^2(P_Z)}^2
\end{align*}
We can now upper bound the final quantity in the previous display using the parallelogram law as
\begin{align*}
&\left\|\Psi^\beta(Z; g) - \Psi^\beta(Z; g')\right\|_{L^2(P_Z)}^2 \lesssim \left\|\smax^\beta_\ell Q_1(\ell, X_1) - \smax^\beta_\ell Q_1'(\ell, X_1)\right\|_{L^2(P_Z)}^2 \\
&\qquad + \left\|\alpha_1(A_1, X_1)\left(\smax^\beta_\ell Q_2(\ell, X_2) - Q_1(A_1, X_1)\right) - \alpha_1'(A_1, X_1)\left(\smax^\beta_\ell Q_2'(\ell, X_2) - Q_1'(A_1, X_1)\right)\right\|_{L^2(P_Z)}^2 \\
&\qquad + \left\|\alpha_2(A_2, X_2)(Y - Q_2(A_2, X_2)) - \alpha_2'(A_2, X_2)(Y - Q_2'(A_2, X_2))\right\|_{L^2(P_Z)}^2.
\end{align*}
Letting $L$ denote the Lipschitz constant of the softmax function (which we showed was Lipschitz in Lemma~\ref{lem:softmax} Point~\ref{pt:softmax_lip}), we can bound the first term above as
\begin{align*}
\left\|\smax^\beta_\ell Q_1(\ell, X_1) - \smax^\beta_\ell Q_1'(\ell, X_1)\right\|_{L^2(P_Z)}^2 &\leq L^2\left\|Q_1(\cdot, X_1) - Q_1'(\cdot, X_1)\right\|_{L^2(P_Z)}^2  \\
&= L^2 \sum_{\ell = 1}^N \|Q_1(\ell, X_1) - Q_1'(\ell, X_1)\|_{L^2(P_{X_1})}^2\\
&\lesssim \left\|Q_1(A_1, X_1) - Q_1'(A_1, X_1)\right\|_{L^2(P_Z)}^2
\end{align*}
where again $Q(\cdot, x_1) := (Q(1, x_1), \dots, Q(N, x_1))$ for any function $Q : [N] \times \calX \rightarrow \R$. The final inequality above follows from strong positivity. The second term can be similarly bounded. In particular, adding and subtracting $\alpha_1(A_1, X_1) (\smax^\beta_\ell Q_2'(\ell, X_2) - Q_1'(A_1, X_1))$ within the norm and applying the parallelogram inequality yields
\begin{align*}
&\left\|\alpha_1(A_1, X_1)\left(\smax^\beta_\ell Q_2(\ell, X_2) - Q_1(A_1, X_1)\right) - \alpha_1'(A_1, X_1)\left(\smax^\beta_\ell Q_2'(\ell, X_2) - Q_1'(A_1, X_1)\right)\right\|_{L^2(P_Z)}^2 \\
&\qquad \leq \left\|\alpha_1(A_1, X_1)(\smax^\beta_\ell Q_2(\ell, X_2) - \smax^\beta Q_2'(\ell, X_2))\right\|_{L^2(P_Z)}^2 + \left\|\alpha_1(A_1, X_2)(Q_1(A_1, X_1) - Q_1'(A_1, X_1))\right\|_{L^2(P_Z)}^2 \\
&\qquad \qquad + \left\|(\alpha_1(A_1, X_1) - \alpha_1'(A_1, X_1))(\smax^\beta Q_2'(\ell, X_2) - Q_1(A_1, X_1))\right\|_{L^2(P_Z)}^2 \\
&\qquad \lesssim \left\|\smax^\beta_\ell Q_2(\ell, X_2) - \smax^\beta_\ell Q_2'(\ell, X_2)\right\|_{L^2(P_Z)}^2 + \left\|Q_1(A_1, X_1) - Q_1'(A_1, X_1)\right\|_{L^2(P_Z)}^2\\
&\qquad \qquad + \left\|\alpha_1(A_1, X_1) - \alpha_1'(A_1, X_1)\right\|_{L^2(P_Z)}^2 \\
&\qquad \lesssim \left\|Q_2(A_2, X_2) - Q_2'(A_2, X_2)\right\|_{L^2(P_Z)}^2 + \left\|Q_1(A_1, X_1) - Q_1'(A_1, X_1)\right\|_{L^2(P_Z)}^2 \\
&\qquad \qquad + \left\|\alpha_1(A_1, X_1) - \alpha_1'(A_1, X_1)\right\|_{L^2(P_Z)}^2,
\end{align*}
where the second to last inequality follows because all nuisances and nuisance estimates are assumed to be bounded. The final inequality follows from the same argument used in the paragraphs above. Similarly, we have:
\begin{align*}
 &\left\|\alpha_2(A_2, X_2)(Y - Q_2(A_2, X_2)) - \alpha_2'(A_2, X_2)(Y - Q_2'(A_2, X_2))\right\|_{L^2(P_Z)}^2 \\
 &\qquad \lesssim \left\|\alpha_2(A_2, X_2) - \alpha_2'(A_2, X_2)\right\|_{L^2(P_Z)}^2 + \left\|Q_2(A_2, X-2) - Q_2'(A_2, X_2)\right\|_{L^2(P_Z)}^2.
\end{align*}
Thus, in aggregate, we have that 
\begin{align*}
\E\left|\G_n \Psi^\beta(Z; g) - \G_n \Psi^\beta(Z; g')\right|^2 &\lesssim \underbrace{\left\|Q_1 - Q_1'\right\|_{L^2(P_Z)}^2 + \left\|Q_2 - Q_2'\right\|_{L^2(P_Z)}^2 + \left\|\alpha_1 - \alpha_1'\right\|_{L^2(P_Z)}^2 + \left\|\alpha_2 - \alpha_2'\right\|_{L^2(P_Z)}^2}_{\Delta(g, g')}
\end{align*}
Now, we show the above empirical process converge in probability to zero. Given $\delta > 0$, we have 
\begin{align*}
\lim_{n \rightarrow \infty}\P\left(\left|\G_n \Psi^\beta(Z; \wh{g}) - \G_n \Psi^\beta(Z; g^\beta)\right| > \delta \right) &= \lim_{n \rightarrow \infty}\E\left[\P_Z\left(\left|\G_n \Psi^\beta(Z; \wh{g}) - \G_n \Psi^\beta(Z; g^\beta)\right| > \delta\right)\right] \\
&\leq \lim_{n \rightarrow \infty}\frac{1}{\delta^2}\E\left[\E_Z\left|\G_n \Psi^\beta(Z; \wh{g}) - \G_n \Psi^\beta(Z; g^\beta)\right|^2\right] \\
&\leq \lim_{n \rightarrow \infty}\frac{1}{\delta^2}\E\left[\E_Z\left| \Psi^\beta(Z; \wh{g}) - \Psi^\beta(Z; g^\beta)\right|^2\right] \\
&\lesssim \lim_{n \rightarrow \infty}\frac{1}{\delta^2}\E\left[\Delta(\wh{g}, g^\beta)^2\right] \\
&= 0.
\end{align*}

The final evaluation of the above limit follows because $\Delta(\wh{g}, g^\beta)$ is bounded by an absolute constant and hence uniformly integrable and because $\Delta(\wh{g}, g^\beta) = o_\P(1)$, which allows us to apply Proposition~\ref{prop:vitali}. This completes the proof of stochastic equicontinuity.

\paragraph{Checking Theorem~\ref{thm:smooth_clt}, Condition~\ref{cond:simp_reg}}
Finally, the boundedness of the score follows immediately from the boundedness of the outcome $Y$ and all nuisances and nuisance estimates. Thus, this implies a finite $(2 + \epsilon)$th moment for any $\epsilon > 0$. This completes the proof of Theorem~\ref{thm:normal_dynamic}.

\end{proof}

\subsection{Inference Under Semi-Parametric Restrictions}

As in Section~\ref{sec:static}, we now consider performing inference on the value of the optimal treatment policy when we make semi-parametric restrictions on the blip functions/blip effects.  In the dynamic setting with two rounds of treatment, the first and second round blip effects are given in terms of the Q-functions as 
\begin{align*}
\gamma_1^\ast(a_1, x_1) &:= Q_1^\ast(a_1, x_1) - Q_1^\ast(a^\ast_1, x_1) \\
\gamma^\ast_2(a_2, x_2) &:= Q_2^\ast(a_2, x_2) - Q_2^\ast(a^\ast_2, x_2),
\end{align*}
where $a^\ast_1, a^\ast_2\in [N]$ are some arbitrary, fixed reference actions (e.g.\ a control treatment). Under Assumption~\ref{ass:seq}, the value of the optimal dynamic treatment policy is identified as:
\begin{equation}
\label{eq:ident-blip-dynamic}
V^\ast := \E\left[\max_{\ell}\gamma_1^\ast(\ell, X_1) - \gamma_1^\ast(A_1, X_1)\right] + \E\left[\max_\ell\gamma_2^\ast(\ell, X_2) - \gamma_2^\ast(A_2, X_2)\right] + \E\left[Y\right]
\end{equation}
We now state our semi-parametric assumption of the blip effect $\gamma_1^\ast$ and $\gamma_2^\ast$.

\begin{assumption}
\label{ass:dynamic_param}
We assume that the blip effects $\gamma^\ast_1(a_1, x_1)$ and $\gamma^\ast_2(a_2, x_2)$ are linear in known feature embeddings of $(a_1, x_1)$ and $(a_2, x_2)$ respectively, i.e.\ that
\[
\gamma^\ast_1(a_1, x_1) = \theta_0^\top \phi_1(a_1, x_1) \text{ for some } \theta_0 \in \R^{d_1} \quad \text{and} \quad \gamma^\ast_2(a_2, x_2) = \psi_0^\top \phi_2(a_2, x_2) \text{ for some }\psi_0 \in \R^{d_2}.
\]

\end{assumption}

Under Assumption~\ref{ass:dynamic_param}, the identification established in Equation~\eqref{eq:ident-blip-dynamic} simplifies to the following:
\begin{equation}
\label{eq:ident-blip-param}
V^\ast = \E\left[\max_{\ell}\theta_0^\top \phi_1(\ell, X_1)  - \theta_0^\top \phi_1(A_1, X_1)  + \max_\ell\psi_0^\top \phi_2(\ell, X_2) - \psi_0^\top \phi_2(A_2, X_2) + Y\right].
\end{equation}
Given this identification, our strategy for performing inference on the value $V^\ast$ is as follows. First, we develop Neyman orthogonal scores for estimating the structural parameters $\theta_0$ and $\psi_0$ (Proposition~\ref{prop:ortho_dynamic_param}). Using these scores, we then prove an asymptotic linearity/normality result for natural estimates of these parameters (Theorem~\ref{thm:normal_dynamic_param}). We then show how these asymptotically linear estimates can be used to construct an asymptotically linear plug-in estimate for the value $V^\ast$. We start by constructing our Neyman orthogonal scores.

\begin{prop}
\label{prop:ortho_dynamic_param}
Consider the scores $\Psi_2$ and $\Psi_1$ defined respectively as 
\[
\Psi_2(Z; \mu_2, r_2, \psi) := \left\{\epsilon_2(Z; \psi) - \mu_2(X_2) + \psi^\top r_2(X_2)\right\}(\phi_2(A_2, X_2) - r_2(X_2))
\]
and 
\[
\Psi_1^\beta(Z; \mu_1, r_1, q_1, \psi, \theta) := \left\{\epsilon_1^\beta(Z; \psi, \theta) - \mu_1(X_1) - q_1(X_1) + \theta^\top r_1(X_1)\right\}(\phi_1(A_1, X_1) - r_1(X_1))
\]
where $\epsilon_1^\beta(Z; \psi, \theta) := Y + \smax^\beta_\ell \psi^\top \phi_2(\ell, X_2) - \psi^\top \phi_2(A_2, X_2) - \theta^\top \phi_1(A_1, X_1)$ and $\epsilon_2(Z;\psi) := Y - \psi^\top \phi_2(A_2, X_2)$, and the true nuisances are given by 
\begin{align*}
\mu_1^\ast(x_1) &:= \E[Y \mid X_1 = x_1] & \mu_2^\ast(x_2) &:= \E[Y \mid X_2 = x_2] \\
r_1^\ast(x_1) &:= \E[\phi_1(A_1, X_1) \mid X_1 = x_1] &  r_2^\ast(x_2) &:= \E[\phi_2(A_2, X_2) \mid X_2 = x_2] \\
q^\beta(x_1)&:= \E[\smax^\beta_\ell \psi_0^\top \phi_2(\ell, X_2) - \psi_0^\top \phi_2(A_2, X_2) \mid X_1 = x_1].
\end{align*}
Then, both $\Psi_1$ and $\Psi_2$ are Neyman orthogonal. Furthermore, under the assumption that $\E[\Cov(\phi_2(A_2, X_2) \mid X_2)] \succ 0$, $\psi_0$ is the unique solution to the estimating equation $0 = \E\left[\Psi_2(Z; \mu_2^\ast, r_2^\ast, \psi_0)\right]$. Likewise, assuming $\E[\Cov(\phi_1(A_1, X_1) \mid X_1)] \succ 0$, there is a unique solution $\theta^\beta$ (not necessarily equal to $\theta_0$) to the estimating equation $0 = \E\left[\Psi_1^\beta(Z; \mu_1^\ast, r_1^\ast, q^\beta, \psi_0, \theta^\beta)\right]$.
\end{prop}

\begin{proof}
Showing the Neyman orthogonality of $\Psi_2$ and that $\psi_0$ is the unique solution to the estimating equation $0 = \E\left[\Psi_2(Z; \mu_2^\ast, r_2^\ast, \psi)\right]$ is exactly analogous to the proof of Proposition~\ref{prop:ortho_static_param}, so we do not repeat the argument here. We now show orthogonality $\Psi^\beta_1$ for any $\beta > 0$. Letting $\Delta : \calX_1 \rightarrow \R$ be some square integrable function, direct computation yields that
\begin{align*}
D_{\mu_1}\E\left[\Psi_1^\beta(Z; \mu_1^\ast, r_1^\ast, q^\beta, \psi_0, \theta^\beta)\right](\Delta) &= -\frac{\partial}{\partial t} \E\left[(\mu_1 + t \Delta)(X_1) (\phi_1(A_1, X_1) - r_1^\ast(X_1))\right]\Big|_{t = 0} \\
&= -\E\left[\Delta(X_1)(\phi_1(A_1, X_1) - r_1^\ast(X_1))\right] \\
&= -\E\left[\Delta(X_1)\left(\E(\phi(A_1, X_1) \mid X_1) - r_1^\ast(X_1)\right)\right] \\
&= 0,
\end{align*}
where the final equality follows from the definition of $r_1^\ast$ and the second to last equality follows from an application of the tower rule for conditional expectation. An identical argument yields that 
\[
D_{q}\E\left[\Psi_1^\beta(Z; \mu_1^\ast, r_1^\ast, q^\beta, \psi_0, \theta^\beta)\right](\Delta) = 0.
\]
We conclude by checking the Gateaux derivative with respect to $r_1^\ast$. Let $\omega : \calX_1 \rightarrow \R^{d_1}$ be an arbitrary square-integrable function. Observe that we have 
\begin{align*}
&D_{r_1}\E\left[\Psi_1^\beta(Z; \mu_1^\ast, r_1^\ast, q^\beta, \psi_0, \theta^\beta)\right](\omega) \\
&\qquad = D_{r_1}\E\left[\left\{(Y - \mu_1^\ast(X_1)) + \smax^\beta_\ell\{ \psi_0^\top (\phi_2(\ell, X_2) - \phi_2(A_2, X_2))\} - q^\beta(X_1)\right\}(\phi_1(A_1, X_1) - r_1^\ast(X_1))\right](\omega) \\
&\qquad\qquad  + D_{r_1}\E\left[(\theta^\beta)^\top(\phi_1(A_1, X_1) - r_1^\ast(X_1))(\phi_1(A_1, X_1) - r_1^\ast(X_1))\right](\omega) \\
&\qquad = -\E\left[\omega(X_1)\left\{(Y - \mu_1^\ast(X_1)) + \smax^\beta_\ell \psi_0^\top (\phi_2(\ell, X_2) - \phi_2(A_2, X_2)) - q^\beta(X_1)\right\}\right] \\
&\qquad \qquad + (\theta^\beta)^\top\E\left[\omega(X_1)(\phi_1(A_1, X_1) - r_1^\ast(X_1))\right] + (\theta^\beta)^\top\E\left[(\phi_1(A_1, X_1) - r_1^\ast(X_1))\omega(X_1)\right]\\
&\qquad = 0, 
\end{align*}
where the final equality follows again from a an application of the tower rule for conditional expectation given $X_1$. Thus, we have shown Neyman orthogonality.

What remains now is for us to show that there is a unique solution $\theta^\beta$ to the estimating equation $0 = \E\left[\Psi_1^\beta(Z; \mu_1^\ast, r_1^\ast, q^\beta, \psi_0, \theta^\beta)\right]$. Rearranging the zero moment condition, we arrive at 
\begin{align*}
&\theta^\top\E\left[\Cov(\phi_1(A_1, X_1) \mid X_1)\right] \\
&\qquad = \E\left[\left\{(Y - \mu_1^\ast(X_1)) + \smax^\beta_\ell\{\psi_0^\top(\phi_2(\ell, X_2) - \phi_2(A_2, X_2))\} - q^\beta(X_1)\right\}(\phi_1(A_1, X_1) - r_1^\ast(X_1))\right],
\end{align*}
which always has a solution. The solution is unique precisely when $\E[\Cov(\phi_1(A_1, X_1) \mid X_1)] \succ 0$, which is the condition noted in the proposition. 
\end{proof}

Note that in establishing orthogonal estimating equations for the first stage structural parameter, we smooth the maximum of the second round blip effects in order to guarantee twice differentiability. As the smoothing parameter $\beta > 0$ grows towards infinity, the softmax function will converge to the regular maximum function (as argued in Lemma~\ref{lem:softmax}). Thus, we define ``un-smoothed'' analogues of the quantities $\Psi^\ast_1$, $\epsilon_1$, and $q^\ast$, since these will be useful in establishing influence functions for asymptotic linearity in the sequel. We define these parameters as follows.
\begin{align*}
\epsilon_1(Z; \psi, \theta) &:= Y + \max_\ell \psi^\top \phi_2(\ell, X_2) - \psi^\top \phi_2(A_2, X_2) - \theta^\top \phi_1(A_1, X_1) \\
q^\ast(x_1) &:= \E\left[\max_\ell \psi_0^\top \phi_2(\ell, X_2) - \psi_0^\top \phi_2(A_2, X_2) \mid X_1 = x_1\right] \\
\Psi^\ast_1(Z; \mu_1, r_1, q, \psi, \theta) &:= \left\{\epsilon_1(Z; \psi, \theta)  -\mu_1(X_1) - q(X_1) + \theta^\top r_1(X_1)\right\}(\phi_1(A_1, X_1) - r_1(X_1)).
\end{align*}
With these definitions, we can state our main theorem in the semi-parametric setting.

\begin{theorem}[Normality of Structural Parameters]
\label{thm:normal_dynamic_param}
Let $\beta_n$ be a smoothing parameter. Let $Z_1, \dots, Z_n$ be i.i.d.\ random variables satisfying Assumption~\ref{ass:seq}, and let $\wh{\mu}_{b,n}$, $\wh{r}_{b, n}$, and $\wh{q}_n$ be nuisance estimates for $\mu_b^\ast$, $r_b^\ast$, and $q^\beta$ (where $b \in \{1, 2\}$), which are assumed to be independent of the observations $Z_1, \dots, Z_n$. Suppose the following assumptions hold.
\begin{enumerate}
    \item \textit{(Controlled Density Near Zero)} We assume that, for any $k \in [N]$, that the second stage sub-optimality gaps $\Delta_{2, k} := \max_{\ell}\psi_0^\top \phi_2(\ell, X_2) - \psi_0^\top \phi_2(k, X_2)$ satisfy Assumption~\ref{ass:margin} with constants $c, H, \delta > 0$.
    \item \textit{(Smoothing Parameter Growth)} We assume the smoothing parameter $\beta_n$ satisfies $\beta_n = \omega\left(n^{\frac{1}{2(1 + \delta)}}\right)$ and $\beta_n = o(n^{1/2})$.
    \item \textit{(Nuisance Convergence)} We assume the nuisance estimates convergence at sufficiently fast rates, in the particular that
    \[
    \|\wh{\mu}_{b, n} - \mu_b^\ast\|_{L^2(P_{X_b})}, \|\wh{r}_{b, n} - r_b^\ast\|_{L^2(P_{X_b})}, \|\wh{q}_n - q^{\beta_n}\|_{L^2(P_{X_1})} = o_\P(n^{-1/4}),
    \]
    where $b \in \{1, 2\}$. Further, we assume all nuisance estimates are consistent in $L^2$.
    \item \textit{(Boundedness)} There is some universal constant $G > 0$ such that, for each $k \in [N]$ and $b \in \{1, 2\}$:
    \[
    \|\phi_b(k, \cdot)\|_{L^\infty(P_{Z})}, \|\mu_b^\ast\|_{L^\infty(P_Z)}, \|\wh{\mu}_{b, n}\|_{L^\infty(P_{Z})}, \|r_b^\ast\|_{L^\infty(P_Z)}, \|\wh{r}_{b, n}\|_{L^\infty(P_Z)}, \|Y\|_{L^\infty(P_Y)} \leq G.
    \]
    Additionally, we assume the structural parameters $\theta_0, \psi_0$ are bounded in $\ell_2$ norm by $G$, and that for any $\beta > 0$ 
    \[
\|\wh{q}_{n}\|_{L^\infty(P_Z)}, \|q^\beta\|_{L^\infty(P_Z)}, \|\theta^\beta\|_2 \leq G.
    \]
    
\end{enumerate}
Then, letting $\wh{\psi}_n$ and $\wh{\theta}_n$ be the solution to the empirical estimating equations 
\[
0 = \P_n \Psi_2(Z; \wh{\mu}_{2, n},\wh{r}_{2, n}, \wh{\psi}_n) \quad \text{and}
 \quad 0 = \P_n \Psi^\beta_1(Z; \wh{\mu}_{1, n}, \wh{r}_{1, n}, \wh{q}_n, \wh{\psi}_n, \wh{\theta}_n)
 \]
respectively, 
we have the following asymptotic linearity results:
\[
\sqrt{n}(\wh{\psi}_n  - \psi_0) = \frac{1}{\sqrt{n}}\sum_{i =1 }^n \rho_\psi(Z_i)+ o_\P(1) \quad \text{and} \quad \sqrt{n}(\wh{\theta}_n - \theta_0) = \frac{1}{\sqrt{n}}\sum_{i = 1}^n \rho_\theta(Z_i) + o_\P(1)
\]
where the influence functions $\rho_\psi$ and $\rho_\theta$ are respectively given by
\begin{align*}
\rho_\psi(Z) &:= \E\left[\Cov(\phi_2(A_2, X_2) \mid X_2)\right]^{-1}\Psi_2(Z; \mu_2^\ast, r_2^\ast, \psi_0) \\
\rho_\theta(Z) &:= \E\left[\Cov(\phi_1(A_1, X_1) \mid X_1)\right]^{-1}\left\{\Psi^\ast_1(Z; \mu^\ast_1, r^\ast_1, q^\ast, \psi_0, \theta_0) + J_\ast^\top \rho_\psi(Z)\right\}
\end{align*}
where $J_\ast := \E\left[(\phi_2^\infty(X_2) - \phi_2(A_2, X_2))(\phi_1(A_1, X_1) - r_1^\ast(X_1))\right]$,  $\phi_2^\infty(x_2) := \frac{1}{|\calM_2(x_2)|}\sum_{\ell \in \calM_2(x_2)} \phi(\ell, x_2)$, and $\calM_2(x_2) := \arg\max_{\ell \in [N]}\psi_0^\top \phi_2(\ell, x_2)$. Consequently, asymptotic normality also holds:
\[
\sqrt{n}(\wh{\psi}_n - \psi_0) \Rightarrow \calN\left(0, \Sigma_2\right) \quad \text{and} \quad \sqrt{n}(\wh{\theta}_n - \theta_0) \Rightarrow \calN\left(0, \Sigma_1\right)
\]
when $\Sigma_2 := \Cov(\rho_\psi(Z_i)) > 0$ and $\Sigma_1 := \Cov(\rho_\theta(Z_i)) > 0$.
\end{theorem}

\begin{proof}
Again, we only prove results pertaining to the first stage structural parameter $\theta_0$, as the proof of asymptotic linearity (and hence asymptotic normality) for $\psi_0$ is exactly analogous to the proof of Theorem~\ref{thm:normal_static_param}. Our recipe for proving the present result is largely the same as the one used in the rest of the paper. First, we show that under the assumptions of the theorem, the bias in the first-stage estimate decays at a rate of $o(n^{-1/2})$, and thus doesn't impact asymptotic linearity. Then, we check the conditions for Theorem~\ref{thm:smooth_clt_generic}, which yields asymptotic linearity. 

To apply the theorem, we decompose the score $\Psi(Z; g, \psi, \theta)$ as 
\[
\Psi(Z; g, \psi, \theta) := a(Z; r_1) \theta + \nu^\beta(Z; \mu_1, r_1, q, \psi),
\]
where, defining $\zeta_1^\beta(Z; \mu_1, q, \psi) := Y + \smax^\beta_\ell \psi^\top \phi_2(\ell, X_2) - \psi^\top \phi_2(A_2, X_2) - \mu_1(X_1) - q(X_1)$, we have
\begin{align*}
a(Z; r_1) &:= - (\phi_1(A_1, X_1) - r_1(X_1))^{\otimes 2} \\
\nu^\beta(Z; \mu_1, q, r_1, \psi) &:= \zeta_1^\beta(Z; \mu_1, q, \psi)(\phi_1(A_1, X_1) - r_1(X_1)).
\end{align*}
We also define $A(r_1) := \E_Z[a(Z; r_1)]$ and $\calV^\beta(\mu_1, r_1, q, \psi) := \E_Z[\nu^\beta(Z; \mu_1, r_1, q, \psi)]$ as the expected values of the constituent scores.

\paragraph{Bounding the Bias}

As in previous proofs, we write
\[
\sqrt{n}(\wh{\theta} - \theta_0) = \sqrt{n}(\wh{\theta} - \theta^\beta) + \sqrt{n}(\theta^\beta - \theta_0) 
\]
and show that $\sqrt{n}(\theta^\beta - \theta_0) = o(1)$.
Note that we can write $\theta_0$ and $\theta^\beta$ respectively as zeros of the non-orthogonal scores
\begin{align*}
\wt{\Psi}^\ast(Z; r_1, q, \psi, \theta) &= \epsilon_1(Z; \phi, \theta) (\phi_1(A_1, X_1) - r_1(X_1)) \\
\wt{\Psi}^\beta(Z; r_1, q, \psi, \theta) &= \epsilon_1^\beta(Z; \phi, \theta) (\phi_1(A_1, X_1) - r_1(X_1)).
\end{align*}
Rearranging the moment conditions, we see that we can write 
\begin{align*}
(\theta^\beta)^\top \Phi_1 =  \E\left[\left\{Y + \smax^\beta_\ell \psi_0^\top \phi_2(\ell, X_2) - \psi_0^\top \phi_2(A_2, X_2) \right\}(\phi_1(A_1, X_1) - r_1^\ast(X_1))\right] \\
\theta_0^\top \Phi_1 =  \E\left[\left\{Y + \max_\ell \psi_0^\top \phi_2(\ell, X_2) - \psi_0^\top \phi_2(A_2, X_2) \right\}(\phi_1(A_1, X_1) - r_1^\ast(X_1))\right]
\end{align*}
where we have let $\Phi_1 := \E\left[\Cov(\phi_1(A_1, X_1) \mid X_1)\right]$ for succinctness. Thus, we have 
\begin{align*}
\|\theta^\beta - \theta_0\|_2 &= \left\|\Phi_1^{-1}\E\left[\left\{\max_\ell \psi_0^\top \phi_2(\ell, X_2) - \smax^\beta_\ell \psi_0^\top \phi_2(\ell, X_2)\right\}(\phi_1(A_1, X_1) - r_1^\ast(X_1))\right]\right\|_2 \\
&\lesssim \left\|\E\left[\left\{\max_\ell \psi_0^\top \phi_2(\ell, X_2) - \smax^\beta_\ell \psi_0^\top \phi_2(\ell, X_2)\right\}(\phi_1(A_1, X_1) - r_1^\ast(X_1))\right]\right\|_2 \\
&\lesssim \E\left|\max_\ell \psi_0^\top \phi_2(\ell, X_2) - \smax^\beta_\ell \psi_0^\top \phi_2(\ell, X_2)\right|\left\|\phi_1(A_1, X_1) - r_1^\ast(X_1))\right\|_{L^\infty(P_Z)} \\
&\lesssim \E\left|\max_\ell \psi_0^\top \phi_2(\ell, X_2) - \smax^\beta_\ell \psi_0^\top \phi_2(\ell, X_2)\right| \\
&\leq \sum_{k = 1}^N \E\left[\Delta_{2, k}\exp\left\{-\beta \Delta_{2, k}\right\}\right] \\
&\lesssim \left(\frac{1}{\beta}\right)^{1 + \delta} \\
&= o(n^{-1/2}),
\end{align*}
where the first equality comes from subtracting the two equations in the previous display and right multiplying by $\Phi^{-1}$ (which exists by assumption). The first inequality follows since positive definiteness of $\Phi_1$ implies the maximum eigenvalue of $\Phi_1^{-1}$ is bounded above by a constant. The second inequality follows from an application of the $L^1/L^\infty$ Holder inequality in each coordinate of the expectation on the previous line. The third inequality follows from the assumption of boundedness of $\phi_1$ and $r_1^\ast$. The fourth inequality follows from Lemma~\ref{lem:margin}, and the fifth follows from Lemma~\ref{lem:margin} alongside the assumed density condition of the theorem. Lastly, the final equality follows since by assumption $\beta_n = \omega\left(n^{\frac{1}{{2(1 + \delta)}}}\right)$.

\paragraph{Checking Theorem~\ref{thm:smooth_clt_generic}, Condition~\ref{cond:neyman}} This was shown in Proposition~\ref{prop:ortho_dynamic_param}.

\paragraph{Checking Theorem~\ref{thm:smooth_clt_generic}, Condition~\ref{cond:hessian}}

We let $g^\beta := (\mu_1^\ast, r_1^\ast, q^\beta)$, $\wh{g} := (\wh{\mu}_1, \wh{r}_1, \wh{q})$, and $\wb{g} \in [g^\beta, \wh{g}]$ for notational convenience. We are interested in bounding
\begin{align*}
D_g^2\E\left[\Psi_1^\beta(Z; \wb{g}, \theta^\beta, \psi_0)\right](\wh{g} - g^\beta) &= 2D_{\mu_1, r_1}\E\left[\wb{\mu}_1(X_1)\wb{r}_1(X_1)\right](\wh{\mu}_1 - \mu_1^\ast, \wh{r}_1 - r_1^\ast) \\
&+ 2D_{q_1, r_1}\E\left[\wb{q}(X_1)\wb{r}_1(X_1)\right](\wh{q} - q^\beta, \wh{r}_1 - r_1^\ast) \\
&- D_{r_1}^2\E\left[(\theta^\beta)^\top \wb{r}_1(X_1)\wb{r}_1(X_1)\right](\wh{r}_1 - r_1^\ast).
\end{align*}
We can just compute the above cross-Gateaux derivatives and second Gateaux derivatives separately. The cross-derivatives are analogous to ones appearing earlier in this work, and so we have 
\begin{align*}
\big\|D_{\mu_1, r_1}\E[\wb{\mu}_1(X_1)\wb{r}_1(X_1)](\wh{\mu}_1 - \mu_1^\ast, \wh{r}_1 - r_1^\ast)\big\|_2 &= \big\|\E\left[(\wh{\mu}_1 - \mu_1^\ast)(X_1)(\wh{r}_1 - r_1^\ast)(X_1)\right]\big\|_2 \\
&= \sqrt{\sum_{k = 1}^d \E\left[(\wh{\mu}_1 - \mu_1^\ast)(X_1)(\wh{r}_{1, k} - r_{1, k}^\ast)(X_1)\right]^2}\\
&\leq \sqrt{\sum_{k = 1}^d \|\wh{\mu}_1 - \mu_1^\ast\|_{L^2(P_{X_1})}^2 \|\wh{r}_{1, k} - r_{1, k}^\ast\|_{L^2(P_{X_1})}^2} \\
&= \|\wh{\mu}_1 - \mu_1^\ast\|_{L^2(P_{X_1})}\|\wh{r}_1 - r_1^\ast\|_{L^2(P_{X_1})} \\
&= o_\P(n^{-1/2}),
\end{align*}
where the first inequality follows from applying Cauchy-Schwarz to each term within the sum, and the final line follows from the assumed nuisance estimation rates.
The exact same argument yields that
\[
\left\|D_{q, r_1}\E[\wb{q}(X_1)\wb{r}_1(X_1)](\wh{q} - q^\beta, \wh{r}_1 - r_1^\ast)\right\|_2 = o_\P(n^{-1/2}).
\]
We now compute the second Gateaux derivative with respect to $r_1$. Letting $\wh{\Delta} := \wh{r}_1 - r_1^\ast$, we have
\begin{align*}
D_{r_1}^2\E\left[r_1(X_1)r_1(X_1)^\top\right](\wh{\Delta}) &= \frac{\partial^2}{\partial t^2}\E\left[(\wb{r}_1 + t \wh{\Delta})(X_1)(\wb{r}_1 + t\wh{\Delta})(X_1)^\top\right]\big|_{t = 0} \\
&= \E\left[\frac{\partial^2}{\partial t^2}(\wb{r}_1 + t \wh{\Delta})(X_1)(\wb{r}_1 + t \wh{\Delta})(X_1)^\top\big|_{t = 0}\right] \\
&= 2\E\left[\wh{\Delta}(X_1)\wh{\Delta}(X_1)^\top\right].
\end{align*}
Where the exchanging of the second partial derivative and expectation is justified by the fact that $\wh{\Delta}$ and $\wb{r}_1$ are both bounded.Thus, we have 
\begin{align*}
\left\|D_{r_1}^2\E\left[(\theta^\beta)^\top \wb{r}_1(X_1)\wb{r}_1(X_1)\right](\wh{\Delta})\right\|_2 &= 2\left\|\E\left[\wh{\Delta}(X_1)\wh{\Delta}(X_1)^\top\right]\theta^\beta\right\|_2 \\
&\leq 2 \left\|\E\left[\wh{\Delta}(X_1)\wh{\Delta}(X_1)^\top\right]\right\|_{op}\left\|\theta^\beta\right\|_2 \\
&= 2\sup_{\|u\| = 1}\E\left[u^\top \wh{\Delta}(X_1)\wh{\Delta}(X_1)^\top u\right] \|\theta^\beta\|_2 \\
&\lesssim \E\|\wh{\Delta}(X_1)\|_2^2 \\
&= \|\wh{r}_1 - r_1^\ast\|_{L^2(P_{X_1})}^2 \\
&= o_\P(n^{-1/2}),
\end{align*}
where the first inequality follows from the fact that $\|A x\|_2 \leq \|A\|_{op}\|x\|_2$, the second equality follows from the definition of the operator norm, the second inequality follows from pushing the supremum inside the expectation and noting that $\|\theta^\beta\|_2$ is bounded by a constant by assumption, and the final equality follows from the assumption of nuisance estimation rates. This completes checking Condition~\ref{cond:hessian}.

\paragraph{Checking Theorem~\ref{thm:smooth_clt_generic}, Condition~\ref{cond:score}}

The only nuisance that depends on the smoothing parameter $\beta > 0$ is $q^\beta(x_1) : = \E\left[\smax^\beta_\ell \psi_0^\top\phi_2(\ell, X_2) - \psi_0^\top \phi_2(A_2, X_2) \mid X_1 = x_1\right]$. We show that this converges almost surely to $q^\ast(x_1) := \E\left[\max_\ell \psi_0^\top\phi_2(\ell, X_2) - \psi_0^\top \phi_2(A_2, X_2) \mid X_1 = x_1\right]$. Observe that the bounded convergence theorem for conditional expectations (which we can apply since we have assumed $\psi_0$ and $\phi_2(\ell, X_2)$ are bounded) yields that
\begin{align*}
\lim_{n \rightarrow \infty}\left|q^\beta(X_1) - q^\ast(X_1)\right| &= \lim_{n \rightarrow \infty}\left|\E\left[\smax^\beta_\ell \psi_0^\top\phi_2(\ell, X_2) - \max_\ell \psi_0^\top \phi_2(\ell, X_2) \mid X_1\right]\right| \\
&\leq \lim_{n \rightarrow \infty}\E\left[\left|\smax^\beta_\ell \psi_0^\top\phi_2(\ell, X_2) - \max_\ell \psi_0^\top \phi_2(\ell, X_2)\right| \mid X_1\right] \\
&= \E\left[\lim_{n \rightarrow \infty}\left|\smax^\beta_\ell \psi_0^\top\phi_2(\ell, X_2) - \max_\ell \psi_0^\top \phi_2(\ell, X_2)\right| \mid X_1\right] \\
&= 0,
\end{align*}
where the final equality follows from the point-wise convergence of the softmax function (Lemma~\ref{lem:softmax} Part~\ref{pt:softmax_lim}). The same part of the aforementioned lemma alongside the fact that $\theta^\beta \xrightarrow[n \rightarrow \infty]{} \theta_0$ (shown above) also implies that
\begin{align*}
\epsilon_1^\beta(Z; \psi_0, \theta^\beta) \xrightarrow[n \rightarrow \infty]{a.s.} \epsilon_1(Z; \psi_0, \theta_0).
\end{align*}
The above two observations imply that
\[
\Psi_1^\beta(Z; g^\beta, \theta^\beta, \psi_0) \xrightarrow[n \rightarrow \infty]{a.s.} \Psi_1^\ast(Z; g^\ast, \theta_0, \psi_0).
\]
Since $a^\beta(Z; r_1) \equiv a^\ast(Z; r_1) := -(\phi_1(A_1, X_1) - r_1(X_1))(\phi_1(A_1, X_1) - r_1(X_1))^\top$ does not depend on $\beta > 0$, there isn't anything to show regarding the convergence of $A^\beta$ to $A^\ast$ in probability.

\paragraph{Checking Theorem~\ref{thm:smooth_clt_generic}, Condition~\ref{cond:equi}}

We now check stochastic equicontinuity, first for $A(r_1) := \E[a(Z; r_1)]$, and next for $\calV^\beta(\mu_1, r_1, q, \psi) := \E\left[\nu^\beta(Z; \mu_1, r_1, q, \psi)\right]$, where we recall
\[
\nu^\beta(Z; \mu_1, r_1, q) := \left\{Y + \smax^\beta_\ell \psi^\top \phi_2(\ell, X_2) - \psi^\top \phi_2(A_2, X_2) - q(X_1)  - \mu_1(X_1)\right\}(\phi_1(A_1, X_1) - r_1(X_1))
\]

Using an argument exactly analogous to the one used to show stochastic equicontinuity in the proof of Theorems~\ref{thm:normal_irregular} and Theorem~\ref{thm:normal_static_param}, we have
\begin{align*}
\E\left\|\G_n a(Z; r_1) - \G_n a(Z; r_1')\right\|_F^2 \lesssim \E\|a(Z; r_1) - a(Z; r_1')\|_F^2.
\end{align*}
We have, for arbitrary $r_1, r_1'$ that are bounded in $L^\infty$ by the constant $G > 0$ noted in Assumption 4 of Theorem~\ref{thm:normal_dynamic_param}
\begin{align*}
&\E\left\|a(Z; r_1) - a(Z; r_1')\right\|_F^2  = \E\left\|(\phi_1(A_1, X_1) - r_1(X_1))^{\otimes 2} -(\phi_1(A_1, X_1) - r_1'(X_1))^{\otimes 2}\right\|_F^2\\
&\qquad = \E\Big\|(\phi_1(A_1, X_1) - r_1(X_1))^{\otimes 2} \pm (\phi_1(A_1, X_1) - r_1(X_1))(\phi_1(A_1, X_1) \\
&\qquad \qquad - r_1'(X_1))^\top - (\phi_1(A_1, X_1) - r_1'(X_1))^{\otimes 2}\Big\|_F^2 \\
&\qquad \lesssim \E\left\|(\phi_1(A_1, X_1) - r_1(X_1))(r_1(X_1) - r_1'(X_1))^\top\right\|_F^2 + \E\left\|(\phi_1(A_1, X_1) - r_1'(X_1))(r_1(X_1) - r_1'(X_1))^\top\right\|_F^2 \\
&\qquad =\E\Big[\big\{\|\phi_1(A_1, X_1) - r_1(X_1)\|_2^2 + \|\phi_1(A_1, X_1) - r_1'(X_1)\|_2^2\big\}\|r_1(X_1) - r_1'(X_1)\|_2^2 \Big] \\
&\qquad \lesssim \E\left\|r_1(X_1) - r_1'(X_1)\right\|_2^2. 
\end{align*}
Where the first equality follows from the definition of $a(Z; r_1)$, the second follows from adding and subtracting $(\phi_1(A_1, X_1) - r_1(X_1))(\phi_1(A_1, X_1) - r_1'(X_1))^\top$, the first inequality follows from the parallelogram inequality, the third equality follows from the fact that $\|ab^\top\|_F^2 = \|a\|_2^2\|b\|_2^2$, and the final inequality follows from the fact that $r_1, r_1'$ have bounded $L^\infty$ norm (and hence their point-wise $\ell_2$ norms are bounded by an absolute constant as well).
The above reasoning allows us to conclude
\[
\E_Z\|\G_n a(Z; \wh{r}_1) - \G_n a(Z; r_1^\ast)\|_F^2 \lesssim \|\wh{r}_1 - r_1^\ast\|_{L^2(P_{X_1})}^2 = o_\P(1)
\]
where the last equality follows by nuisance consistency. As in the proofs of Theorem~\ref{thm:normal_irregular} and Theorem~\ref{thm:normal_static_param}, conditionally applying Chebyshev's inequality and applying Vitali's Theorem (Proposition~\ref{prop:vitali}) gives 
\[
\|\G_n a(Z; \wh{r}_1) - \G_n a(Z; r_1^\ast)\|_{op} \leq \|\G_n a(Z; \wh{r}_1) - \G_n a(Z; r_1^\ast)\|_F = o_\P(1).
\]

Similarly, letting $h_1 = (\mu_1, r_1, q)$ and $\zeta_1^\beta(Z; \mu_1, q, \psi) := Y + \smax^\beta_\ell \psi^\top \phi_2(\ell, X_2) - \psi^\top \phi_2(A_2, X_2) - \mu_1(X_1) - q(X_1)$, we have 
\begin{align*}
\E\|\G_n \nu^\beta(Z; h_1, \psi_0) - \G_n \nu^\beta(Z; h_1', \psi_0)\|_2^2 \lesssim \E\|\nu^\beta(Z; h_1,\psi_0) - \nu^\beta(Z; h_1', \psi_0)\|_2^2. 
\end{align*}
We further have 
\begin{align*}
&\E\|\nu^\beta(Z; h_1,\psi_0) - \nu^\beta(Z; h_1', \psi_0)\|_2^2 \\
&\qquad = \E\Big\|\zeta_1^\beta(Z; \mu_1, q, \psi_0) (\phi_1(A_1, X_1) - r_1(X_1)) - \zeta_1^\beta(Z; \mu_1', q', \psi_0)(\phi_1(A_1, X_1) - r_1'(X_1))\Big\|_2^2 \\
&\qquad = \E\Big\|\zeta_1^\beta(Z; \mu_1, q, \psi_0) (\phi_1(A_1, X_1) - r_1(X_1)) \pm \zeta_1^\beta(Z; \mu_1, q, \psi_0)(\phi_1(A_1, X_1) - r_1'(X_1)) \\
&\qquad \qquad - \zeta_1^\beta(Z; \mu_1', q', \psi_0)(\phi_1(A_1, X_1) - r_1'(X_1))\Big\|_2^2 \\
&\qquad \lesssim \E\Big\|\zeta_1^\beta(Z; \mu_1, q, \psi_0)(r_1(X_1) - r_1'(X_1))\Big\|_2^2 \\
&\qquad \qquad + \E\Big\|\left\{\zeta_1^\beta(Z; \mu_1, q, \psi_0) - \zeta_1^\beta(Z; \mu_1', q', \psi_0)\right\}(\phi_1(A_1, X_1) - r_1'(X_1))\Big\|_2^2 \\
&\qquad \lesssim \E\left\|r_1(X_1) - r_1'(X_1)\right\|_2^2 + \E\left|\zeta_1^\beta(Z; \mu_1, q, \psi_0) - \zeta_1^\beta(Z; \mu_1', q', \psi_0)\right|^2,
\end{align*}
which follows in the same manner as the above arguments since $Y, r_1, r_1', q, q', r_1$, and $r_1'$ are assumed to be bounded. We have
\begin{align*}
\E\left|\zeta_1^\beta(Z; \mu_1, q, \psi_0) - \zeta_1^\beta(Z; \mu_1', q', \psi_0)\right|^2 &= \E\left|(\mu_1 - \mu_1')(X_1) + (q - q')(X_1)\right|^2 \\
&\lesssim \E\left|\mu_1(X_1) - \mu_1'(X_1)\right|^2 + \E\left|q(X_1) - q'(X_1)\right|^2,
\end{align*}
which follows from the parallelogram inequality. Since the outcome $Y$, all nuisances, and all nuisance estimates are bounded, this implies
\begin{align*}
\E_Z\left\|\G_n \nu^\beta(Z; \hat{h}_1, \psi_0) - \G_n \nu^\beta(Z; h_1^\beta, \psi_0)\right\|_2^2 &\lesssim \|\wh{\mu}_1 - \mu_1^\ast\|_{L^2(P_{X_1})}^2 + \|\wh{r}_1 - r_1^\ast\|_{L^2(P_{X_1})}^2 + \|\wh{q} - q^\beta\|_{L^2(P_{X_1})}^2 \\
&= o_\P(1)
\end{align*}
where the final line follows from nuisance consistency. Again, a conditional application of Chebyshev's inequality along with the bounded convergence theorem alongside Vitali's Theorem (Proposition~\ref{prop:vitali}) as in the proofs of Theorems~\ref{thm:normal_irregular} and \ref{thm:normal_static_param} yields 
\[
\left\|\G_n \nu^\beta(Z; \hat{h}_1, \psi_0) - \G_n \nu^\beta(Z; h_1^\beta, \psi_0)\right\|_2 = o_\P(1).
\]
This concludes the proof of stochastic equicontinuity.

\paragraph{Checking Theorem~\ref{thm:smooth_clt_generic}, Condition~\ref{cond:regularity}}
Since the outcomes $Y$ and all nuisances/nuisance estimates are bounded, it follows that there exists an absolute constant $D > 0$ such that
\[
    \|\theta^\beta\|_2, \sup_{1 \leq j, k \leq p}\left\|a^\beta_{j, k}(Z; g^\beta)\right\|_{L^{\infty}(P_Z)}, \sup_{1 \leq j \leq p}\left\|\nu^\beta_j(Z; g^\beta, h_0)\right\|_{L^{\infty}(P_Z)} \leq D.
\]
The above bounds on the $L^\infty$ norms directly imply the bounds on the $L^{2 + \epsilon}$ norms as well, by monotonicity of $L^p$ norms. The invertibility of $A(r_1)$ follows from the assumption that $\E\left[\Cov(\phi_1(A_1, X_1) \mid X_1)\right] \succ 0$. We now check mean-squared continuity of $A(r_1)$. Using an analogous argument to the one leveraged to show stochastic equicontinuity of $a(Z; r_1)$, we have
\begin{align*}
\|A(\wh{r}_1) - A(r_1^\ast)\|_{op} & \leq \|A(\wh{r}_1) - A(r_1^\ast)\|_{F} \\
&=\left\|\E\left[(\phi_1(A_1, X_1) - \wh{r}_1(X_1))^{\otimes 2} \right] - \E\left[(\phi_1(A_1, X_1) - r_1^\ast(X_1))^{\otimes 2}\right]\right\|_F \\
&\leq \left\|\E\left[(\phi_1(A_1, X_1) - \wh{r}_1(X_1))(\wh{r}_1(X_1) - r_1^\ast(X_1))^\top\right]\right\|_F \\
&\qquad + \left\|\E\left[(\phi_1(A_1, X_1) - r_1^\ast(X_1))(\wh{r}_1(X_1) - r_1^\ast(X_1))^\top\right]\right\|_F \\
&\leq \E\left\|(\phi_1(A_1, X_1) - \wh{r}_1(X_1))(\wh{r}_1(X_1) - r_1^\ast(X_1))^\top\right\|_F \\
&\qquad + \E\left\|(\phi_1(A_1, X_1) - r_1^\ast(X_1))(\wh{r}_1(X_1) - r_1^\ast(X_1))^\top\right\|_F \\
&=\E\left[\left\|\phi_1(A_1, X_1) - \wh{r}_1(X_1)\right\|_2\left\|\wh{r}_1(X_1) - r_1^\ast(X_1)\right\|_2\right] \\
&\qquad + \E\left[\left\|\phi_1(A_1, X_1) - r_1^\ast(X_1)\right\|_2\left\|\wh{r}_1(X_1) - \wh{r}_1^\ast(X_1)\right\|_2\right] \\
&\lesssim  \E\left\|\wh{r}_1(X_1) - r_1^\ast(X_1)\right\|_2 \\
&\leq \|\wh{r}_1 - r_1^\ast\|_{L^2(P_{X_1})},
\end{align*}
where the first inequality follows as $\|A\|_{op} \leq \|A\|_F$, the second inequality follows from the triangle inequality, the third follows from Jensen's inequality, and the second equality again follows from the equality $\|a b^\top\|_F = \|a\|_2 \|b\|_2$. The second to last inequality follows from the assumed boundedness of $\phi_1(A_1, X_1), r_1^\ast(X_1)$, and $\wh{r}_1$, and the final line follows from an application of Jensen's inequality again. This shows mean-squared continuity.

\paragraph{Checking Theorem~\ref{thm:smooth_clt_generic}, Condition~\ref{cond:h_conditions}}
First, we note that \textbf{Condition~\ref{cond:h_diff}} follows from the differentiability of the softmax function. To check mean-squared continuity (\textbf{Condition~\ref{cond:h_cont}}) of the Jacobian, observe that we have
\[
\partial_h \Psi_1^\beta(Z; g, \psi_0, \theta^\beta)  = (\phi_1(A_1, X_1) - r_1(X_1))\left\{\sum_{k = 1}^N \left(\partial_k \smax^\beta_\ell \psi_0^\top \phi_2(\ell, X_2)\right)\phi_2(k, X_2) - \phi_2(A_2, X_2)\right\}^\top.
\]
Using this, we can check mean-square continuity. Noting that the Jacobian above only depends on the nuisance $r_1$ and the structural parameter $\psi_0$, we can omit dependence on other parameters and write
\begin{align*}
\left\|\partial_\psi \Psi_1^\beta(Z; r_1, \psi_0) - \partial_\psi \Psi_1^\beta(Z; r_1', \psi_0)\right\|_{F} &= \left\|(r_1(X_1) - r_1'(X_1))\left\{\sum_{k = 1}^N\left( \partial_k \smax^\beta_\ell \psi_0^\top \phi_2(\ell, X_2)\right)\phi_2(k, X_2) - \phi_2(A_2, X_2)\right\}^\top\right\|_F\\
&\leq \sum_{k = 1}^N\left|\partial_k \smax^\beta_\ell \psi_0^\top \phi_2(\ell, X_2)\right|\left\|(r_1(X_1) - r_1'(X_1))\phi_2(k, X_2)\right\|_F \\
&\qquad + \left\|(r_1(X_1) - r_1'(X_1))\phi_2(A_2, X_2)^\top\right\|_F \\
&\leq \sum_{k = 1}^N\left|\partial_k \smax^\beta_\ell \psi_0^\top \phi_2(\ell, X_2)\right|\left\|r_1(X_1) - r_1'(X_1)\right\|_2 \left\|\phi_2(k, X_2)\right\|_2 \\
&\qquad + \left\|r_1(X_1) - r_1'(X_1)\right\|_2\left\|\phi_2(A_2, X_2)^\top\right\|_2 \\ 
&\lesssim \|r_1(X_1) - r_1'(X_1)\|_2,
\end{align*}
which follows because $\|\phi_2(k, X_2)\|_2$ is assumed to be bounded and the first partial derivatives of the softmax function are uniformly bounded by a constant that does not depend on $\beta$ (Lemma~\ref{lem:softmax} Point~\ref{pt:softmax_grad}). Thus, we have 
\begin{align*}
\E\left\|\partial_\psi \Psi_1^\beta(Z; r_1, \psi_0) - \partial_\psi \Psi_1^\beta(Z; r_1', \psi_0)\right\|_{F}^2 \lesssim \E\left\|r_1(X_1) - r_1'(X_1)\right\|_2^2,
\end{align*}
showing the desired continuity result.

We now check \textbf{Condition~\ref{cond:h_jacobian}}. The Jacobian of $\calV^\beta(h_1^\beta, \psi_0)$ is precisely the same as the Jacobian of $\E\left[\Psi^\beta_1(Z; g^\beta, \psi_0, \theta^\beta)\right]$, which we have computed above. Thus we see that
\begin{align*}
&\left\|\partial_\psi \calV^\beta(h_1^\beta, \psi_0) - J^\ast\right\|_{op} = \left\|\E\left[(\phi_1(A_1, X_1) - r_1^\ast(X_1))\left\{\sum_{k = 1}^N \left(\partial_k \smax^\beta_\ell \psi_0^\top \phi_2(\ell, X_2)\right)\phi_2(k, X) - \phi_2^\infty(X_2)\right\}^\top\right]\right\|_{op} \\
&\qquad \leq \E\left[\sum_{k \in \calM(X_2)} \left|\partial_k \smax^\beta_\ell \psi_0^\top \phi_2(\ell, X_2) - \frac{1}{|\calM(X_2)|}\right|\Big\|(\phi_1(A_1, X_1) - r_1^\ast(X_1))\phi_2(k, X_2)^\top\Big\|_{op}\right] \\
&\qquad \qquad + \E\left[\sum_{k \notin \calM(X_2)}\left|\partial_k \smax^\beta_\ell \psi_0^\top \phi_2(\ell, X_2)\right|\Big\|(\phi_1(A_1, X_1) - r_1^\ast(X_1))\phi_2(k, X_2)^\top\Big\|_{op}\right] \\
&\qquad \lesssim \E\left[\sum_{k \in \calM(X_2)}\left|\partial_k \smax^\beta_\ell \psi_0^\top \phi_2(\ell, X_2) - \frac{1}{|\calM(X_2)|}\right|\right] +\E\left[\sum_{k \notin \calM(X_2)}\left|\partial_k \smax^\beta_\ell \psi_0^\top \phi_2(\ell, X_2)\right|\right],
\end{align*}
where the first inequality follows from first pushing the operator norm into the expectation and then applying the triangle inequality and the second inequality follows from an application of $L^1/L^\infty$ Holder's inequality alongside the boundedness of $\phi_1, \phi_2,$ and $r_1^\ast$. Thus, we have
\begin{align*}
\lim_{n \rightarrow \infty}\left\|\partial_\psi \calV^\beta(h_1^\beta, \psi_0) - J^\ast\right\|_{op} &= \lim_{n \rightarrow \infty}\E\left[\sum_{k \in \calM(X_2)}\left|\partial_k \smax^\beta_\ell \psi_0^\top \phi_2(\ell, X_2) - \frac{1}{|\calM(X_2)|}\right|\right] \\
&\qquad + \lim_{n \rightarrow \infty}\E\left[\sum_{k \notin \calM(X_2)}\left|\partial_k \smax^\beta_\ell \psi_0^\top \phi_2(\ell, X_2)\right|\right] \\
&= \E\left[\lim_{n \rightarrow \infty}\sum_{k \in \calM(X_2)}\left|\partial_k \smax^\beta_\ell \psi_0^\top \phi_2(\ell, X_2) - \frac{1}{|\calM(X_2)|}\right|\right] \\
&\qquad + \E\left[\lim_{n \rightarrow \infty}\sum_{k \notin \calM(X_2)}\left|\partial_k \smax^\beta_\ell \psi_0^\top \phi_2(\ell, X_2)\right|\right] \\
&= 0,
\end{align*}
where we are able to exchange limit and expectation via the bounded convergence theorem and the final equality follows from the point-wise convergence of the derivative of the softmax function (Lemma~\ref{lem:softmax} Point~\ref{pt:softmax_grad_lim}). Thus, we have shown almost sure convergence of the Jacobian, which is strictly stronger than the convergence in probability result desired.

Next, we check \textbf{Condition~\ref{cond:h_hessian}}. The almost sure boundedness of the coordinates of $\partial_\psi \nu^\beta(z; h_1, \psi_0)$ for $h_1 \in [h_1^\ast, \wh{h}_1]$ follows from the fact that the gradient of the softmax function is bounded (Lemma~\ref{lem:softmax} Point~\ref{pt:softmax_grad}) alongside the boundedness of nuisances and nuisance estimates. For a given coordinate $j \in \{1, \dots, d\}$, we can directly compute and bound the Hessian in terms of the smoothing parameter $\beta > 0$. The Hessian $\partial_\psi^2 \nu_j^\beta(Z; h_1, \psi)$ of the $j$th coordinate of $\nu^\beta$ is given by
\[
\partial_\psi^2 \nu_j^\beta(Z; h_1, \psi) = (\phi_{1, j}(A_1, X_1) - r_1(X_1))\bfphi_2(X_2)\nabla_u^2 \smax^\beta_\ell \{u\}\vert_{u = \psi^\top \phi_2(\ell, X_2)} \bfphi_2(X_2)^\top
\]
where $\bfphi_2 : \calX_2 \rightarrow \R^{d \times N}$ is given by $\bfphi_2(x_2) := (\phi_2(1, x_2), \dots, \phi_2(N, x_2))$. The operator norm is thus bounded as
\begin{align*}
\left\|\partial_\psi^2 \nu_j^\beta(Z; h_1, \psi)\right\|_{op} &= \big|(\phi_{1, j}(A_1, X_1) - r_1(X_1))\big|\left\|\bfphi_2(X_2)\nabla_u^2 \smax^\beta_\ell\{ \psi^\top \phi_2(\ell, X_2) \}\bfphi_2(X_2)^\top\right\|_{op} \\
&\lesssim \left\|\bfphi_2(X_2)\nabla_u^2 \{\smax^\beta_\ell\{ \psi^\top\} \phi_2(\ell, X_2)\} \bfphi_2(X_2)^\top\right\|_{op} \\
&\leq \left\|\bfphi_2(X_2)\right\|_{op}^2\left\|\nabla_u^2 \smax^\beta_\ell\{ \psi^\top \phi_2(\ell, X_2)\}\right\|_{op} \\
&\lesssim \left\|\nabla_u^2 \smax^\beta_\ell \psi^\top \phi_2(\ell, X_2) \right\|_{op} \\
&= O(\beta_n)  \\
&= o(n^{1/2})
\end{align*}
The first inequality follows because $\phi_1$ is assumed to be bounded and $r_1 \in [r_1^\ast, \wh{r}_1]$, and $r_1^\ast$ and $\wh{r}_1$ are also both assumed to be bounded. The second inequality follows since $\|AB\|_{op} \leq \|A\|_{op}\|B\|_{op}$, the third inequality follows from the boundedness of $\phi_2$. The second to last line follows since the operator norm of the Hessian of the softmax function is bounded (up to an absolute constant) by  $\beta$ (Lemma~\ref{lem:softmax} Point~\ref{pt:softmax_hess}). The final line follows since we have assumed that $\beta_n = o(n^{1/2})$. This concludes our verification of Condition~\ref{cond:h_hessian}.

Lastly, we note that \textbf{Condition~\ref{cond:h_linearity}}, which concerns verifying  that the finite-dimensional nuisance parameters are asymptotically linear, follows from the first part of the theorem regarding the second stage structural parameters. This concludes the proof of the theorem.
\end{proof}

Showing asymptotic linearity (and hence normality) of the finite-dimensional structural parameters is the most difficult part of this section. Using Theorem~\ref{thm:normal_dynamic_param}, we can now show asymptotic linearity of the plug-in estimate for the value of the optimal treatment policy $V^\ast$. The following theorem provides this.

\begin{theorem}
\label{thm:val_dynamic_param}
Assume the same setup as Theorem~\ref{thm:normal_dynamic_param} and let $\wh{\psi}_n, \wh{\theta}_n$ denote respectively the estimates of the second and first stage structural parameters $\psi_0, \theta_0$. Define $\phi_1^\infty(x_1) := \frac{1}{|\calM_1(x_1)|}\sum_{\ell \in \calM_1(x_1)}\phi_1(\ell, x_1)$ and $\calM_1(x_1) := \arg\max_{\ell \in [N]}\theta_0^\top \phi_1(\ell, x_1) \subset [N]$. Then, defining the plug-in policy value estimate $\wh{V}_n$ as 
\[
\wh{V}_n := \P_n\Phi^\beta(Z; \wh{\psi}_n, \wh{\theta}_n),
\]
where 
\[
\Phi^\beta(Z; \psi, \theta) := \smax_\ell^\beta \theta^\top \phi_1(\ell, X_1) - \theta^\top \phi_1(A_1, X_1) + \smax^\beta_\ell \psi^\top \phi_2(\ell, X_2) - \psi^\top \phi_2(A_2, X_2) + Y,
\]
we have 
\begin{align*}
\sqrt{n}(\wh{V}_n - V^\ast) = \frac{1}{\sqrt{n}}\sum_{i = 1}^n \rho_V(Z_i) + o_\P(1),
\end{align*}
where the influence function $\rho_V$ is given by
\begin{align*}
\rho_V(Z) := \Phi^\ast(Z; \psi_0, \theta_0) - V^\ast + \E\left[\phi_2^\infty(X_2) - \phi_2(A_2, X_2)\right]^\top \rho_\psi(Z) + \E\left[\phi_1^\infty(X_1) - \phi_1(A_1, X_1)\right]^\top \rho_\theta(Z).
\end{align*}
Consequently, we have the following asymptotic normality result
\[
\sqrt{n}(\wh{V}_n - V^\ast) \Rightarrow \calN(0, \sigma_V^2) \quad \text{where }\sigma_V^2 := \Var[\rho_V(Z)].
\]
This implies that if $z_\delta$ denotes the $\delta$th quantile of the standard normal distribution and $\wh{\sigma}_n^2$ is a consistent estimator of $\sigma_V^2$, then 
\[
C_{1 - \delta} := \left[\wh{V}_n - \frac{\wh{\sigma}_n}{\sqrt{n}}z_{\delta/2}, \wh{V}_n + \frac{\wh{\sigma}_n}{\sqrt{n}}z_{\delta/2}\right]
\]
is an asymptotically valid $1 - \delta$ confidence interval for $V^\ast$.
\end{theorem}

\begin{proof}
The proof of asymptotic linearity of the estimate $\wh{V}$ is largely a simplified version of the proof of asymptotic linearity of $\wh{\theta}$ presented in Theorem~\ref{thm:normal_dynamic_param}. We still present a full proof for completeness, but present abbreviated arguments to avoid repetition. We first bound the bias from softmax approximation and then check the relevant conditions of Theorem~\ref{thm:smooth_clt_generic}. In particular, we need not check Conditions~\ref{cond:neyman}, \ref{cond:hessian}, \ref{cond:equi}, and \ref{cond:h_cont}.

\paragraph{Bounding the Bias}
We can first write 
\begin{align*}
\sqrt{n}(\wh{V} - V^\ast) &= \sqrt{n}(\wh{V} - V^\beta) + \sqrt{n}(V^\beta - V^\ast)
\end{align*}
and then further write
\begin{align*}
|V^\beta - V^\ast| &\leq  \E\left|\smax^\beta_\ell \psi_0^\top \phi_2(\ell, X_2) - \max_\ell \psi_0^\top \phi_2(\ell, X_2)\right| + \E\left|\smax^\beta_\ell \theta_0^\top \phi_1(\ell, X_1) - \max_\ell \theta_0^\top \phi_1(\ell, X_1)\right| \\
&\lesssim \sum_{k = 1}^N \E\left[\Delta_{2, k}\exp\{-\beta\Delta_{2, k}\}\right] + \sum_{k = 1}^N \E\left[\Delta_{1, k}\exp\{-\beta \Delta_{1, k}\}\right] \\
&\lesssim \left(\frac{1}{\beta}\right)^{1 + \delta} \\
&= o(n^{-1/2}),
\end{align*}
where again the inequality follows from an application of the triangle inequality alongside Jensen's, the second inequality follows from applying the first part of Lemma~\ref{lem:margin}, the next inequality follows from applying the second part of Lemma~\ref{lem:margin}, and the final line follows from the assumption that $\beta_n = \omega\left(n^{\frac{1}{2(1 + \delta)}}\right)$.

\paragraph{Checking Theorem~\ref{thm:smooth_clt_generic}, Condition~\ref{cond:score}}
Once again, by leveraging the point-wise convergence of the soft-max function (Lemma~\ref{lem:softmax} Point~\ref{pt:softmax_lim}), we have 
\begin{align*}
\lim_{n \rightarrow \infty}\left|\Phi^\beta(Z; \psi_0, \theta_0) - \Phi^\ast(Z; \psi_0, \theta_0)\right| &\leq \lim_{n\rightarrow \infty}\left|\smax^\beta_\ell \psi_0^\top\phi_2(\ell, X_2) - \max_\ell \psi_0^\top \phi_2(\ell, X_2) \right| \\
&\qquad + \lim_{n \rightarrow \infty}\left|\smax^\beta_\ell \theta_0^\top\phi_1(\ell, X_1) - \max_\ell \theta_0^\top \phi_1(\ell, X_1) \right| \\
&= 0.
\end{align*}

\paragraph{Checking Theorem~\ref{thm:smooth_clt_generic}, Condition~\ref{cond:regularity}}
We only need to argue boundedness, since the first part of this condition is not applicable as there is no infinite-dimensional nuisance component. But since we have assumed $\theta_0, \psi_0, \phi_1, \phi_2, Y$ and all nuisance/nuisance estimates are bounded, it follows that both $V^\beta$ and $\Phi^\beta$ are bounded almost surely. This in particular implies that the random variable $\Phi^\beta(Z; \psi_0, \theta_0)$ is also bounded by the same constant in $L^{2 + \epsilon}(P_Z)$ for any $\epsilon > 0$.

\paragraph{Checking Theorem~\ref{thm:smooth_clt_generic}, Condition~\ref{cond:h_conditions}}
First, \textbf{Condition~\ref{cond:h_diff}} is immediate as the softmax function is infinitely differentiable for any fixed $\beta > 0$. Likewise, there is nothing we need to check for \textbf{Condition~\ref{cond:h_cont}}, as there is no infinite-dimensional nuisance component.

We now check convergence of the Jacobian (\textbf{Condition~\ref{cond:h_jacobian}}). First we observe that we have 
\[
\partial_{\psi, \theta}\calV^\beta(\psi, \theta) = \begin{pmatrix} \partial_\psi \calV^\beta( \psi, \theta) \\ \partial_\theta \calV^\beta( \psi, \theta)
\end{pmatrix}
\]
Straightforward computation yields that 
\begin{align*}
\lim_{n \rightarrow \infty}\left\|\partial_\psi \calV^\beta( \psi_0, \theta_0) - \partial_\psi \calV^\ast(\psi_0, \theta_0)\right\|_2  &\leq  \lim_{n \rightarrow \infty}\E\left\|\sum_{k = 1}^N\left(\partial_k\smax^\beta_\ell \psi_0^\top \phi_2(\ell, X_2)\right)\phi_2(k, X_2) - \phi_2^\infty(X_2)\right\|_2.
\end{align*}
 We have
\begin{align*}
&\lim_{n \rightarrow \infty}\E\left\|\sum_{k = 1}^N\left(\partial_k\smax^\beta_\ell \psi_0^\top \phi_2(\ell, X_2)\right)\phi_2(k, X_2) - \phi_2^\infty(X_2)\right\|_2 \\
&\qquad \leq  \lim_{n \rightarrow \infty}\E\left[\sum_{k \in \calM(X_2)}\left|\partial_k\smax^\beta_\ell \psi_0^\top \phi_2(\ell, X_2) - \frac{1}{|\calM_2(X_2)|}\right|\left\|\phi_2(k, X_2)\right\|_2\right]  \\
&\qquad \qquad + \lim_{n \rightarrow \infty}\E\left[\sum_{k \notin \calM_2(X_2)}\left|\partial_k\smax^\beta_\ell \psi_0^\top \phi_2(\ell, X_2)\right|\left\|\phi_2(k, X_2)\right\|_2\right]\\
&\qquad =  \E\left[\lim_{n \rightarrow \infty}\sum_{k \in \calM(X_2)}\left|\partial_k\smax^\beta_\ell \psi_0^\top \phi_2(\ell, X_2) - \frac{1}{|\calM_2(X_2)|}\right|\left\|\phi_2(k, X_2)\right\|_2\right]  \\
&\qquad \qquad + \E\left[\lim_{n \rightarrow \infty}\sum_{k \notin \calM_2(X_2)}\left|\partial_k\smax^\beta_\ell \psi_0^\top \phi_2(\ell, X_2)\right|\left\|\phi_2(k, X_2)\right\|_2\right] \\
&= 0
\end{align*}
where the first inequality follows from the definition of $\phi_2^\infty$ and the triangle inequality, the first equality follows from the bounded convergence theorem, and the final equality follows from the point-wise convergence of $\partial_k\smax^\beta_\ell u \xrightarrow[\beta \rightarrow \infty]{} \fraks_k(u)$ (Lemma~\ref{lem:softmax} Point~\ref{pt:softmax_grad_lim}). It thus follows that $\partial_\psi \calV^\beta(\psi_0, \theta_0) \xrightarrow[n \rightarrow \infty]{}J_\psi^\ast$. An exactly identical argument yields that $\partial_\theta \calV^\beta(\psi_0, \theta_0) \xrightarrow[n \rightarrow \infty]{} J_\theta^\ast$.

We now check \textbf{Condition~\ref{cond:h_hessian}}. Again, the boundedness of the components of the Jacobian (i.e.\ gradient) follows from the assumption that $\phi_1, \phi_2$ are bounded with along with the fact that the gradient of the softmax function is bounded (Lemma~\ref{lem:softmax} Point~\ref{pt:softmax_grad}). The Hessian $\partial_{\psi, \theta}^2 \nu^\beta(Z; \psi, \theta)$ is given as
\begin{align*}
\partial_{\psi, \theta}^2\nu^\beta(Z; \psi, \theta) = \begin{pmatrix}
\partial_\psi^2 \nu^\beta(Z; \psi, \theta) & \mathbf{0} \\
\mathbf{0} & \partial_\theta^2\nu^\beta(Z; \psi, \theta).
\end{pmatrix}
\end{align*}
We just bound the operator norm of $\partial_\psi^2\nu^\beta(Z; \psi, \theta)$. Bounding $\|\partial_\theta^2 \nu^\beta(Z; \psi, \theta)\|_{op}$ is exactly analogous, and $\|\partial_{\psi, \theta}^2 \nu^\beta(Z; \psi, \theta)\|_{op} \leq \|\partial_\psi^2\nu^\beta(Z; \psi, \theta)\|_{op} + \|\partial_\theta^2 \nu^\beta(Z; \psi, \theta)\|_{op}$. We have 
\begin{align*}
\left\|\partial_\psi^2 \calV^\beta(\psi, \theta)\right\|_{op} &= \left\|\bfphi_2(X_2) \nabla_u^2\smax^\beta_\ell\{ \psi^\top \phi_2(\ell, X_2) \}\bfphi_2(X_2)^\top\right\|_{op} \\
&= o(n^{1/2}),
\end{align*}
which follows from the same computation/reasoning used in the validation of Condition~\ref{cond:h_hessian} in the proof of Theorem~\ref{thm:normal_dynamic_param}. This completes the point of Theorem~\ref{thm:val_dynamic_param}
\end{proof}

\section{A Smoothed Central Limit Theorem}
\label{app:clt}

In this section, we prove a generic central limit theorem (CLT) for parameters that are specified as zeros of sequence of smooth M-estimation problems. These scores can depend on both infinite-dimensional and finite-dimensional nuisance parameters that must be estimated from data. We require the scores to be Neyman orthogonal with respect to the infinite-dimensional/non-parametric nuisance functions, but relax orthogonality for finite-dimensional nuisance components, which can be estimated at parametric $O_\P(n^{-1/2})$ rates. In our applications (See Appendices~\ref{app:treatment} and \ref{app:irregular}), the degree of smoothness will be governed by the temperature parameter $\beta$ of the softmax function. However, the scores discussed in this appendix can be obtained via any abstract smoother, so long as the smoothed score obeys certain continuity and regularity assumptions. 

We are interested in estimating the population solution to the following $M$-estimation problem. Let $Z$ denote an abstract observation taking values in some measurable space $\calZ$, $\beta > 0$ denote some smoothing/indexing parameter, $g \in L^2(P_W)$ a square-integrable function depending on a subset of features $W \subset Z$, and $h \in \R^q$ a finite dimensional parameter. Our goal is to understand the behavior of the solution $\wh{\theta}$ to the following empirical M-estimation problem:
\begin{align*}
M^\beta_n(\wh{\theta}^\beta, \wh{g}, \wh{h}) &= 0 &\text{where} & &M^\beta_n(\theta, g, h) := \P_n m^\beta(Z; \theta, g, h)
\end{align*}
where, $m^\beta(Z; \theta, g, h)$ denotes some score depending on both finite and infinite-dimensional nuisances as well as the a smoothing parameter $\beta > 0$, and $\wh{g}$ and $\wh{h}$ represent estimated nuisance components. If $\beta > 0$ were fixed, under standard conditions, $\wh{\theta}$ would be asymptotically linear around the solution to the population estimating equations, i.e.\ the value $\theta^\beta$ specified via
\begin{align*}
 M^\beta(\theta^\beta, g^{\beta}, h_0) &=  0 &\text{where} & &M^\beta(\theta, g, h) &:= \E_Z[m^\beta(Z; \theta, g, h)],
\end{align*}
where $g^\beta$ denotes a true, unknown infinite-dimensional nuisance parameter that depends on the smoothing parameter $\beta > 0$ and $h_0$ denotes a true, finite-dimensional nuisance which is assumed to not depend on the particular index $\beta > 0$. However, we are not interested in the solution to the above population equations. Rather, assuming $m^\beta(Z; \theta^\beta, g^\beta, h_0)$ converges in probability to some appropriate limiting score $m^\ast(Z; \theta_0, g_0, h_0)$ with limiting nuisance $g_0 \in L^2(P_W)$ and true solution $0 = M^\ast(\theta_0, g_0, h_0) := \E\left[m^\ast(Z; \theta_0, g_0, h_0)\right]$, our goal is to understand under what conditions $\wh{\theta}$ is asymptotically linear about $\theta_0$. The main theorem of this appendix will detail how to select smoothing parameters $(\beta_n)_{n \geq 1}$ that tend towards infinity as $n$ grows large in order to guarantee said asymptotic linearity.

To simplify the regularity assumptions required for asymptotic normality, we focus on the case where $m^\beta(Z; \theta, g, h)$ is linear in $\theta$ for each $\beta$ in the following sense:
\begin{equation}
\label{eq:linear_score}
    m^\beta(Z; \theta, g, h) = a^\beta(Z;g)\theta + \nu^\beta(Z;g,h)
\end{equation}
where $a^\beta(Z;g)\in \mathbb{R}^{p\times p}$ is a $p\times p$ matrix and $\nu^\beta(Z;g,h)\in \mathbb{R}^p$ is a $p$-vector. We correspondingly define:
\begin{align*}
    A^\beta(g) &:= \E_Z[a^\beta(Z;g)] & 
    A_n^\beta(g) &:=  \P_n[a^\beta(Z;g)]\\
    \calV^\beta(g,h) &:= \E_Z[\nu^\beta(Z;g,h)] & 
    \calV_n^\beta(g,h) &:= \P_n[\nu^\beta(Z;g,h)].
\end{align*}
We now formally present the main theorem of this appendix.

\begin{theorem}
\label{thm:smooth_clt_generic}
Suppose that, for any $\beta > 0$, $m^\beta(Z; \theta, g, h)$ is a score as outlined in Equation~\eqref{eq:linear_score}. Let $h_0 \in \R^q$ denote the true finite-dimensional nuisance and, for each $\beta > 0$, let $g^\beta \in L^2(P_W)$ denote the true infinite-dimensional nuisance. Let $(\beta_n)_{n \geq 1}$ satisfy $\beta_n \rightarrow \infty$, let $(Z_n)_{n \geq 0}$ be a sequence of i.i.d.\ draws from some distribution $P_Z$, let $(\wh{g}_n)_{n \geq 1}$ and $(\wh{h}_n)_{n \geq 1}$ be sequences of random nuisance estimates, where we assume $\wh{g}_n$ is independent of $Z_1, \dots, Z_n$ for each $n \geq 1$, and define the sequence $(\wh{\theta}_n)_{n \geq 1}$ as the solution to the empirical estimation equation
\[
0 = M_n^\beta(\wh{\theta}_n, \wh{g}_n, \wh{h}_n) = A_n^\beta(\wh{g}_n)\wh{\theta}_n + \calV_n^\beta(\wh{g}, \wh{h}_n).
\]
Suppose the following hold:
\begin{enumerate}
    \item \textit{(Neyman Orthogonality)} For any $\beta > 0$, the score $m^\beta$ is Neyman orthogonal, i.e.\ for all $g \in L^2(P_W)$
    \[
    D_g \E_Z[m^\beta(Z; \theta^\beta, g^\beta, h_0)](g - g^{\beta}) = 0.
    \] \label{cond:neyman}
    \item \textit{(Nuisance Convergence)} The second-order errors shrink sufficiently quickly, i.e.\ $\forall \wb{g} \in [g^{\beta_n}, \wh{g}],$
    \[
    D_g^2\E_{Z}\left[m^\beta(Z; \theta^\beta, \wb{g}, h_0)\right](\wh{g}_n - g^{\beta_n}) = o_\P(n^{-1/2}). 
    \]
    Further, we assume $L^2$ nuisance consistency, i.e.\ $\|\wh{g}_n - g^{\beta_n}\|_{L^2(P_W)} = o_\P(1)$.
    \label{cond:hessian}
    \item \textit{(Convergence of Score)} The scores appropriately converge as $\beta$ grows large, i.e.
    \[
    m^{\beta}(Z; \theta^\beta, g^\beta, h_0) \xrightarrow[\beta \rightarrow \infty]{\P} m^\ast(Z; \theta_0, g_0, h_0)
    \]
    and 
    \[
    A^\beta(g^\beta) \xrightarrow[\beta \rightarrow \infty]{} A^\ast(g_0)
    \]
    for some limiting scores $m^\ast$, $A^\ast$, and some limiting nuisance $g_0 \in L^2(P_W)$.\label{cond:score}
    \item \textit{(Stochastic Equicontinuity)} We have 
    \[
    \left\|\G_n a^{\beta_n}(Z; \wh{g}_n) - \G_n a^{\beta_n}(Z; g^{\beta_n})\right\|_{op} = o_\P(1) \quad \text{and} \quad \left\|\G_n \nu^{\beta_n}(Z; \wh{g}_n, h_0) - \G_n \nu^{\beta_n}(Z; g^{\beta_n}, h_0)\right\|_2 = o_\P(1),
    \]
    where $\G_n := \sqrt{n}(\P_n - \E_Z)$.\label{cond:equi}
    \item \textit{(Regularity)} We assume the limiting score $A^\ast(g_0)$ is invertible, and that $A^\beta(g)$ is mean-squared continuous in $g$ in the sense that
    \[
    \|A^{\beta_n}(\wh{g}_n) - A^{\beta_n}(g^{\beta_n})\|_{op} \lesssim \|\wh{g}_n - g^{\beta_n}\|_{L^2(P_W)} \quad \forall n \geq 1
    \]
     almost surely, where $\lesssim$ indicates inequality up to an absolute constant that does not depend on $\beta > 0$. We also assume there exists $\epsilon > 0$ such that, for any $\beta > 0$
    \[
    \|\theta^\beta\|_2, \sup_{1 \leq j, k \leq p}\left\|a^\beta_{j, k}(Z; g^\beta)\right\|_{L^{2 + \epsilon}(P_Z)}, \sup_{1 \leq j \leq p}\left\|\nu^\beta_j(Z; g^\beta, h_0)\right\|_{L^{2 + \epsilon}(P_Z)} \leq D,
    \]
    where $D > 0$ is some absolute constant. Further, we assume $\left\|m^\ast(Z; \theta_0, g_0, h_0)\right\|_{L^{2 + \epsilon}(P_Z)} \leq D$ as well. \label{cond:regularity}
    \item \textit{(Influence of $h$)} We make the following assumptions about the relationship between $h$ and $m^\beta(Z; g, h)$.\label{cond:h_conditions}
    \begin{enumerate}
        \item \textit{(Continuity)} For any $g \in L^2(P_W)$ and $z \in \calZ$, $\nu^\beta(z;  g, h)$ is twice continuously differentiable in $h$.
        \label{cond:h_diff}
        \item \textit{(Mean-Squared Continuity of Jacobian)}The Jacobian $\partial_h m^\beta(Z;  g, h_0)$ is continuous in $g$ in the sense that
        \[
        \left\|\partial_h m^{\beta_n}(Z; \wh{g}_n, h_0) - \partial_h m^{\beta_n}(Z; g^{\beta_n}, h_0)\right\|_{L^2(P_Z)}^2 \lesssim \|\wh{g}_n - g^{\beta_n}\|_{L^2(P_W)}^2, \quad \forall n \geq 1
        \]
        almost surely, where $\lesssim$ indicates inequality up to an absolute constant that does not depend on $\beta > 0$.
        \label{cond:h_cont}
        \item \textit{(Convergent Jacobian)} There is some matrix $J^\ast \in \R^{p \times q}$ such that 
        \[
        \partial_h \calV^\beta(g^\beta, h_0) = J^\ast + o(1).
        \]
        \label{cond:h_jacobian}
        \item \textit{(Bounded Derivatives)} For any $1 \leq j, k \leq p$,  the $(j, k)$th component of the Jacobian of $\nu^\beta$ with respect to $h$ is uniformly bound, i.e.\ there is some constant $D > 0$ such that
        \[
        \sup_{\beta > 0, g \in [\wh{g}, g^\beta]}\left|\left(\partial_h \nu^\beta(Z; g, h_0)\right)_{j, k}\right| \leq D
        \]
        almost surely. Further, for each coordinate $1 \leq j \leq p$, the Hessian of $\nu^\beta_j$ with respect to $h$ is bounded as
        \[
        \sup_{h \in [\wh{h}, h_0], g \in [\wh{g}, g^\beta]}\left\|\nabla_h^2 \nu^{\beta_n}_j(Z; g, h)\right\|_{op} = o(n^{1/2})
        \]
        almost surely.
        \label{cond:h_hessian}
        \item \textit{(Asymptotic Linearity)} The sequence $(\wh{h}_n)$ is asymptotically linear around about $h_0$ with influence function $f_h(z)$ satisfying $\E f_h(Z) = 0$, i.e.
        \[
        \wh{h}_n - h_0 = \frac{1}{n}\sum_{i = 1}^n f_h(Z_i) + o_\P(n^{-1/2}).
        \]
        \label{cond:h_linearity}
    \end{enumerate}
\end{enumerate}
Then, we have the following asymptotic linearity result:
\[
\sqrt{n}(\wh{\theta}_n - \theta^{\beta_n}) = -\frac{1}{\sqrt{n}}\sum_{i = 1}^n \underbrace{A^\ast(g_0)^{-1}\left\{ m^\ast(Z_i; \theta_0, g_0, h_0)  + J^\ast f_h(Z_i)\right\}}_{=:\rho_\theta(Z_i)} + o_\P(1).
\]
In particular, this implies
\[
\sqrt{n}(\wh{\theta}_n - \theta^{\beta_n}) \Rightarrow \calN(0, \Sigma^\ast)
\]
whenever $\Sigma^\ast := \Cov[\rho_\theta(Z)]$ is positive definite.
\end{theorem}

Often, throughout our work, we are interested in performing inference on parameters of the form $\theta^\beta = \E[m^\beta(Z; g^\beta)]$, where $m^\beta$ is a score that depends only on an observation $Z \in \calZ$ and an infinite-dimensional nuisance component $g^\beta$. This setting can be captured a special case of Theorem~\ref{thm:smooth_clt_generic} where (1) there is no finite-dimensional nuisance component $h_0$, (2) $a^\beta \equiv I_d$, where $\theta \in \R^d$ and $I_d$ is the $d \times d$ identity matrix, and (3) $\nu^\beta(Z; g) = -m^\beta(Z; g)$. We actually use this simplified CLT to prove our main results (namely Theorem~\ref{thm:normal_irregular}). The general theorem is only used in semi-parametric settings. We provide a short proof of this simplified result later in this appendix. We believe this result result may be of broader applicability in the causal inference/semi-parametric inference literature.

\begin{theorem}
\label{thm:smooth_clt}
Suppose, for any $\beta > 0$, there is a score $m^\beta(Z; g)$. Define $\theta_\beta := \E m^\beta(Z; g^{\beta_n})$, where $g^\beta$ is some true nuisance parameter. Let $(\beta_n)_{n \geq 0}$ satisfy $\beta_n \rightarrow \infty$, let $(Z_n)_{n \geq 0}$ be a sequence of i.i.d.\ draws from some distribution $\P_Z$, let $(\wh{g}_n)$ be a sequence of random nuisance estimates independent of $(Z_n)$, and define
\[
\wh{\theta}_n := \P_n m^{\beta_n}(Z; \wh{g}_n).
\]
Suppose the following hold:
\begin{enumerate}
    \item \textit{(Neyman Orthogonality)} For any $\beta > 0$, the score $m^\beta$ is Neyman orthogonal, i.e. for all $g \in L^2(P_W)$
    \[
    D_g \E_Z[m^\beta(Z; g^\beta)](g - g^{\beta}) = 0.
    \]
   \label{cond:simp_neyman}
    \item \textit{(Second Order Gateaux Derivatives)} The second-order errors shrink sufficiently quickly, i.e.\ $\forall \wb{g} \in [g^{\beta_n}, \wh{g}]$
    \[
    D_g^2\E_{Z}\left[m^{\beta_n}(Z; \wb{g})\right](\wh{g}_n - g^{\beta_n}) = o_\P(n^{-1/2}). 
    \]
    \label{cond:simp_hessian}
    \item \textit{(Convergence of Score)} We have, as $\beta \rightarrow \infty$,
    \[
    m^{\beta}(Z; g^\beta) \xrightarrow[]{\P} m^\ast(Z; g_0)
    \]
    for some limiting score $m^\ast$ and some nuisance $g_0 \in L^2(P_W)$.\label{cond:simp_score}
    \item \textit{(Stochastic Equicontinuity)} We have 
    \[
    |\G_n m^{\beta_n}(Z; \wh{g}_n) - \G_n m^{\beta_n}(Z; g^{\beta_n})| = o_\P(1),
    \]
    where $\G_n := \sqrt{n}(\P_n - \E_Z)$.\label{cond:simp_equi}
    \item \textit{(Regularity)} There exists $D > 0$ and $\epsilon > 0$ such that, for any $\beta > 0$, we have
    \[
    \|m^\beta(Z; g^\beta)\|_{L^{2 + \epsilon}(P_Z)} \leq D.
    \]\label{cond:simp_reg}
    Further, we assume $\|m^\ast(Z; g_0)\|_{L^{2 + \epsilon}(P_Z)} \leq D$ as well.
\end{enumerate}

Then, we have the following asymptotic linearity result:
\[
\sqrt{n}(\wh{\theta}_n - \theta^{\beta_n}) = \frac{1}{\sqrt{n}}\sum_{m = 1}^n (m^\ast(Z; g_0) - \theta^\ast) + o_\P(1).
\]
In particular, this implies
\[
\sqrt{n}(\wh{\theta}_n - \theta^{\beta_n}) \Rightarrow \calN(0, \Sigma^\ast)
\]
where $\Sigma^\ast := \Var[m^\ast(Z; g_0)]$.
\end{theorem}

\subsection{Proof of Theorem~\ref{thm:smooth_clt_generic}}

Before proving the result, we state Vitali's theorem, a result that relates convergence of a sequence of random variables in $L^p$ to convergence in probability and uniform integrability.

\begin{prop}[\citet{bogachev2007measure}]
\label{prop:vitali}
Suppose $(X_n)_{n \geq 1}$ are a sequence of random variables and $X$ is a random variable in a measure space $(\Omega, \calF, P)$ such that $X_n \in L^p(P)$ for each $n$ and $X \in L^p(P)$. The following are equivalent.
\begin{enumerate}
    \item $\|X_n - X\|_{L^p(P)} \xrightarrow[n \rightarrow \infty]{} 0$, i.e.\ $(X_n)_{n \geq 0}$ converge in $L^p$ to $X$.
    \item $X_n \xrightarrow[n \rightarrow \infty]{\P} X$ and $(|X_n|^p)_{n \geq 1}$ is uniformly integrable.\footnote{We say a collection of random variables $(X_\alpha)_{\alpha \in A}$ is uniformly integrable if, for any $\delta > 0$, there exists a constant $K_\delta > 0$ such that
    \[
    \sup_{\alpha \in A}\E\left[|X_\alpha|\mathbbm{1}\left\{|X_\alpha| > K_\delta\right\}\right] \leq \delta.
    \]}
\end{enumerate}
\end{prop}

We also use the following result which can be used to relate uniformly bounded higher moments of a sequence of random variables/scores to uniform integrability.

\begin{prop}[\citet{poussin1915integrale,dellacherie2011probabilities}]
\label{prop:pous}
A sequence of integrable random variables $(X_n)_{n \geq 1}$ is uniformly integrable if and only if there is an increasing convex function $\Phi : \R_{\geq 0} \rightarrow \R_{\geq 0}$ such that
\[
\lim_{t \rightarrow \infty}\frac{\Phi(t)}{t} = \infty \quad \text{and} \quad \sup_n \E\left[\Phi(|X_n|)\right] < \infty.
\]

\end{prop}

With the above two lemmas along with the various assumptions posited above, we can prove Theorem~\ref{thm:smooth_clt_generic}.
\begin{proof}
Throughout the proof, we suppress dependence of various parameters on the sample size $n \geq 1$ to reduce clutter. By the linearity of the moment with respect to $\theta$, we have
\begin{align*}
    A^\beta(\wh{g}) \left(\wh{\theta} -\theta^\beta\right) =~&  M^\beta(\wh{\theta},  
    \wh{g}, h_0) - M^\beta(\theta^\beta,  \wh{g},h_0)\\
    =~& M^\beta(\theta^\beta, g^{\beta},h_0) - M^\beta(\theta^\beta, \wh{g},h_0)\\
    +~& M^\beta(\wh{\theta}, \wh{g},h_0) - M_n^\beta(\wh{\theta}, \wh{g},h_0)\\
    +~& M_n^\beta(\wh{\theta}, \wh{g},h_0) - M_n^\beta(\wh{\theta}, \wh{g},\wh{h})\\
    +~& M_n^\beta(\wh{\theta}, \wh{g},\wh{h}) - M^\beta(\theta^\beta, g^\beta,h_0)
\end{align*}
Note that by definitions of $\wh{\theta}$
 and $\theta^\beta$, we have $M^\beta(\theta^\beta, g_0^\beta,h_0) ~=~ 0$ and $M_n^\beta(\wh{\theta}, \wh{g},\wh{h}) ~=~ 0$.
Therefore, we can simplify the above expression to obtain \begin{align*}
    A^\beta(\wh{g})  \left(\wh{\theta} -\theta^\beta\right) 
    =~& \underbrace{M^\beta(\theta^\beta, g^{\beta},h_0) - M^\beta(\theta^\beta, \wh{g},h_0)}_{=:I_{1,n}} +  \underbrace{M^\beta(\wh{\theta}, \wh{g},h_0) - M_n^\beta(\wh{\theta}, \wh{g},h_0)}_{=:I_{2, n}}\\
    +~& \underbrace{M_n^\beta(\wh{\theta}, \wh{g},h_0) - M_n^\beta(\wh{\theta}, \wh{g},\wh{h})}_{ =:I_{3, n}}
\end{align*} 

In the following, we analyze the asymptotic behaviors of $I_{1,n}, I_{2,n}, I_{3,n}$ separately, starting with $I_{1,n}$. 

\paragraph{Asymptotic Behavior of $\mathbf{I_{1, n}}$.}
By performing a second order Taylor expansion with mean value theorem remainder and exploiting Neyman orthogonality (Condition~\ref{cond:neyman}), we obtain that:
\begin{align*}
     M^{\beta}(\theta^\beta, \wh{g},h_0) - M^\beta(\theta^\beta, g^{\beta},h_0) &= D_g M^\beta(\theta^\beta, g^\beta, h_0;\beta)(\wh{g} -g^\beta) + \frac{1}{2}D_g^2 M^{\beta}(\theta^\beta, \wb{g}, h_0)(\wh{g} - g^\beta) \\
    &= \frac{1}{2}D_g^2 M^{\beta}(\theta^\beta, \wb{g}, h_0)(\wh{g} - g^\beta) \\
    &= o_\P(n^{-1/2}),
\end{align*}
where the final equality follows from the convergence rates of nuisances (Condition~\ref{cond:hessian}). Thus, we have $I_{1, n} = o_\P(n^{-1/2})$.

\paragraph{Asymptotic Behavior of $\mathbf{I_{2, n}}$.}
We now move on to analyzing $I_{2,n}$. In this goal, define the scaled empirical process $G_n^\beta(\theta, g, h) := M^\beta(\theta, g,h) - M_n^\beta(\theta, g,h)$. Likewise, define the limiting scaled empirical process as $G^\ast_n(\theta, g, h) := M^\ast(\theta, g, h) - M_n^\ast(\theta, g, h)$, where $M^\ast(\theta, g, h) := \E_Z[m^\ast(Z; \theta, g, h)]$ and $M_n^\ast := \P_n m^\ast(Z; \theta, g, h)$ and $m^\ast$ is the limiting score defined above. $I_{2,n}$ can be re-expressed in terms of the empirical processes $G_n^\beta$ and $G^\ast_n$ as 
\begin{align*}
I_{2,n}=G_n^\beta\left(\wh{\theta}, \wh{g}, h_0\right) =~& G_n^\ast(\theta_0, g_0, h_0) 
+
\underbrace{\left(G_n^\beta(\theta^\beta, g^\beta, h_0) - G_n^\ast(\theta_0, g_0, h_0)\right)}_{I_{2, n}^{(1)}} \\
+~& 
\underbrace{\left(G_n^\beta(\theta^\beta, \wh{g}, h_0) - G_n^\beta(\theta^\beta, g^\beta, h_0)\right)}_{I_{2, n}^{(2)}} \\
+~& 
\underbrace{\left(G_n^\beta(\wh{\theta}, \wh{g}, h_0) - G_n^\beta(\theta^\beta, \wh{g}, h_0)\right)}_{I_{2, n}^{(3)}}.\label{last_term}
\end{align*}
\\
First, we note that, for each $n$, $I_{2, n}^{(1)}$ is precisely a sum of i.i.d.\ mean zero random variables:
\begin{align*}
    -I_{2, n}^{(1)} &= \frac{1}{n}\sum_{i=1}^{n}\left\{\left(m^\ast(Z_i;\theta_0,g_0,h_0)-m^\beta(Z_i;\theta^\beta,g^\beta,h_0)\right)-\E\left[m^\ast(Z_i;\theta_0,g_0,h_0)-m^\beta(Z_i;\theta^\beta,g^\beta,h_0)\right]\right\}
\end{align*}
Note then that we have
\begin{align*}&
    \left\|\frac{1}{n}\sum_{i=1}^{n}\left\{\left(m^\ast(Z_i;\theta_0,g_0,h_0)-m^\beta(Z_i;\theta^\beta,g^\beta,h_0)\right)-\E\left[m^\ast(Z_i;\theta_0,g_0,h_0)-m^\beta(Z_i;\theta^\beta,g^\beta,h_0)\right]\right\}\right\|_{L^2(P_Z)}\\
    &\qquad =\Var\left(\frac{1}{n}\sum_{i=1}^{n}\left\{m^\ast(Z_i;\theta_0,g_0,h_0)-m^\beta(Z_i;\theta^\beta,g^\beta,h_0)\right\}\right)^{1/2}\\
    &\qquad =\frac{1}{\sqrt{n}}\Var\Big(m^\ast(Z;\theta_0,g_0,h_0)-m^\beta(Z; \theta^\beta,g^\beta,h_0)\Big)^{1/2} \qquad \qquad \mbox{(by independence)}\\
    &\qquad \leq \frac{1}{\sqrt{n}}\left\|m^\ast(Z;\theta_0,g_0,h_0)-m^\beta(Z;\theta^\beta,g^\beta,h_0)\right\|_{L^2(P_Z)} \qquad \qquad (\Var[X] \leq \E X^2) \\
    &\qquad = o(n^{-1/2}).
\end{align*}
We now argue that our assumptions imply the final equality. Since we already know that $m^{\beta}(Z; \theta^\beta, g^\beta, h_0) - m^\ast(Z; \theta_0, g_0, h_0) = o_\P(1)$, it suffices to show that $\|m^{\beta}(Z; \theta^\beta, g^\beta, h_0) - m^\ast(Z; \theta_0, g_0, h_0)\|_2^2$ (implicitly indexed by sample size $n$) are uniformly integrable, as then Vitali's Theorem (Proposition~\ref{prop:vitali}) implies the final equality. To check uniform integrability, we apply Proposition~\ref{prop:pous} with $\Phi(t) := t^{1 + \frac{\epsilon}{2}}$ (where $\epsilon$ is as in Condition~\ref{cond:regularity}), which yields
\begin{align*}
\sup_{n \geq 1}\E\left[\Phi\left(\|m^{\beta}(Z; \theta^\beta, g^\beta, h_0) - m^\ast(Z; \theta_0, g_0, h_0)\|_2^2\right)\right] &= \sup_{n \geq 1}\E\left\|m^{\beta}(Z; \theta^\beta, g^\beta, h_0) - m^\ast(Z; \theta_0, g_0, h_0)\right\|^{2 + \epsilon}_2 \\
&\leq \sup_{n\geq 1} d^{\epsilon/2}\E\left\|m^{\beta}(Z; \theta^\beta, g^\beta, h_0) - m^\ast(Z; \theta_0, g_0, h_0)\right\|_{2 + \epsilon}^{2 + \epsilon} \\
&\leq d^{\epsilon/2}(2D)^{2 + \epsilon}
\end{align*}
where the final inequality follows from Condition~\ref{cond:regularity} and the second to last inequality follows from the inequality $\|x\|_q \leq d^{\frac{1}{q} - \frac{1}{p}} \|x\|_p$ for $x \in \R^d$ and $p > q \geq 1$. 
Thus, an application of Chebyshev's inequality yields the desired convergence in probability result for $I_{2, n}^{(1)}$, i.e.\ that
\begin{align*}
G_n^\beta(\theta^\beta, g^\beta, h_0) - G_n^\ast(\theta_0, g_0, h_0) =o_{\P}(n^{-1/2}).
\end{align*}

\noindent By the linearity of the moment $m^\beta$, the last term $I_{2, n}^{(3)}$ can be bounded as
\begin{equation}\label{tr}
    \begin{aligned}
    \left\|I_{2, n}^{(3)}\right\|_2 &= \left\|G_n^\beta(\wh{\theta}, \wh{g}, h_0) - G_n^\beta(\theta^\beta, \wh{g}, h_0)\right\|_2 = \left\|\left(A^\beta(\wh{g}) - A_n^\beta(\wh{g})\right)^\top(\wh{\theta}-\theta^\beta)\right\|_2 \\
    &\leq \left\|A^\beta(\wh{g}) - A^\beta_n(\wh{g})\right\|_{op}\|\wh{\theta} - \theta^\beta\|_2.
    \end{aligned}
\end{equation}

\noindent Therefore to successfully upper-bound \cref{tr} we need to show  $\left\|A^\beta( \wh{g}^{\beta}) - A_n^\beta(\wh{g}^{\beta})\right\|_{op} $ converges in probability to zero. Note that an application of the triangle inequality yields
\begin{equation}\label{tr2}
\begin{aligned}
   \left\|A^\beta(\wh{g}) - A_n^{\beta}(\wh{g})\right\|_{op} &\leq \left\|A^\beta(g^\beta)-A_n^{\beta}(g^\beta)\right\|_{op} \\
&\quad +  \left\|A^\beta(\wh{g}) - A^\beta(g^\beta) - \left(A_n^\beta(\wh{g}) - A_n^\beta(g^\beta)\right)\right\|_{op}
\end{aligned} 
\end{equation}
We can prove that each term on the right-hand side of \cref{tr2} is negligible. Stochastic equicontinuity (Condition~\ref{cond:equi}) yields that the second summand on the right hand side is $o_\P(n^{-1/2})$. To bound the first summand, note that we have
\begin{align*}
&\E\left[\left\|A^\beta(g^\beta) - A_n^\beta(g^\beta)\right\|_{op}^2\right] \leq \E\left[\left\|A^\beta(g^\beta) - A^\beta_n(g^\beta)\right\|_F^2\right] \\
&\qquad = \E\left[\sum_{j, k = 1}^p\left\{\P_n a^\beta_{j, k}(Z_i; g^\beta) - \E a^\beta_{j, k}(Z; g^\beta)\right\}^2\right] \\
&\qquad\leq p^2 \max_{j, k} \E\left[\left\{(\P_n - \E)a_{j, k}^\beta(Z_i; g^\beta)\right\}^2\right] \\
&\qquad =p^2\max_{j, k}\Var\left(\frac{1}{n}\sum_{i\le n}a_{j,k}^\beta(Z_i; g^\beta)\right) \\
&\qquad =\frac{p^2}{n} \max_{j, k}\Var\left(a_{j,k}^\beta(Z; g^\beta)\right) \\
&\qquad\leq  \frac{p^2}{n} \max_{j, k}\E\left[a_{j,k}^\beta(Z; g^\beta)^2\right] \\
&\qquad = O(1/n),
\end{align*}
where the last equality follows from the fact that 
\[
\sup_{j, k, \beta}\E\left[a^\beta_{j, k}(Z; g^\beta)^2\right] \leq \sup_{j, k, \beta} \left(\E\left[a^\beta_{j, k}(Z; g^\beta)^{2 + \epsilon}\right]\right)^{\frac{2}{2 + \epsilon}}\ \leq D < \infty.
\]
Again, this coupled with Chebyshev's inequality yields that $\left\|A^\beta(g^\beta) - A_n^\beta(g^\beta)\right\|_{op} = o_\P(1)$. Putting this all together, we get that
\[
\left\|I_{2, n}^{(3)}\right\|_2 = o_\P(1) \cdot \|\wh{\theta} - \theta^\beta\|_2 = o_\P\left(\|\wh{\theta} - \theta^\beta\|_2\right). 
\]
Now, we show $\|I_{2, n}^{(2)}\|_2 = o_\P(n^{-1/2})$. Now, expanding the definition of $I_{2, n}^{(2)}$ and applying the triangle inequality yields
\begin{align*}
    \left\|I_{2, n}^{(2)}\right\|_2 &= \left\|G_n^\beta(\theta^\beta, \wh{g}, h_0) - G_n^\beta(\theta^\beta, g^\beta, h_0)\right\|_2\\
    &\leq \left\|A^\beta(\wh{g}) - A^\beta(g^\beta) - \left(A_n^\beta(\wh{g}) - A_n^\beta(g^\beta)\right)\right\|_{op}\, \|\theta_0^\beta\|_2 \\
    &\qquad + \left\|\calV^\beta(\wh{g}^{\beta},h_0) - \calV^\beta(g^\beta, h_0) - \left(\calV_n^\beta(\wh{g}^{\beta},h_0) - \calV_n^\beta(g^\beta, h_0)\right)\right\|_2\\
    &= o_{\P}(n^{-1/2}),
\end{align*}
where the inequality follows due to stochastic equicontinuity of $A^\beta$ and $\calV^\beta$ (Condition~\ref{cond:equi}) and from the assumption
$\sup_{\beta>0}\|\theta_0^\beta\|_2\leq D < \infty$ (Condition~\ref{cond:regularity}). Altogether, the above analysis yields that
\[
I_{2,n} = G_n^\ast(\theta_0, g_0, h_0) + o_{\P}(n^{-1/2} + \|\wh{\theta}- \theta^\beta\|_2).
\]

\paragraph{Asymptotic Behavior of $\mathbf{I_{3, n}}$}

First, observe that we can perform a second order Taylor expansion with mean value remainder to obtain
\begin{align*}
I_{3, n} &:= M_n^\beta(\wh{\theta}, \wh{g}, h_0) - M_n^\beta(\wh{\theta}, \wh{g}, \wh{h}) \\
&= \partial_h M_n^\beta(\wh{\theta}, \wh{g}, h_0)(h_0 - \wh{h}) + \frac{1}{2}\partial_h^2 M_n^\beta(\wh{\theta}, \wh{g}, (\wt{h}_1, \dots, \wt{h}_q))(h_0 - \wh{h}, h_0 - \wh{h})\\
&= \partial_h M_n^\beta(\wh{\theta}, \wh{g}, h_0)(h_0 - \wh{h}) + o(n^{1/2}) \|h_0 - \wh{h}\|_2^2 \\
&= \partial_h M_n^\beta(\wh{\theta}, \wh{g}, h_0)(h_0 - \wh{h}) + o_\P(n^{-1/2}).
\end{align*}
The second equality follows from a second order Taylor expansion with mean-value theorem remainder. In more detail, we apply the mean-value theorem component-wise: $\wt{h}_1, \dots, \wt{h}_q \in [h_0, \wh{h}]$ and $\partial_h^2M_n^\beta(\wh{\theta}, \wh{g}, (\wt{h}_1, \dots, \wt{h}_q)) : \R^q \times \R^q \rightarrow \R^p$ is the quadratic form whose $k$th component (for $k \in [q]$) is given by $v^\top \partial_h^2 M_n^\beta(\wh{\theta}, \wh{g}, \wt{h_q}) v$ for any $v \in \R^q$. The third line follows from the fact we assumed the Hessian of each component has operator norm bounded above by $o(n^{1/2})$ uniformly in $g \in [\wh{g}, g^\beta]$ and $h \in [\wh{h}, h_0]$ (Condition~\ref{cond:h_hessian}). The final line follows since $\wh{h}$ is asymptotically linear about $h_0$ (Condition~\ref{cond:h_linearity}), which implies $\|\wh{h} - h_0\|_2^2 = O_\P(n^{-1})$. Next, since we have assumed $m^\beta(z; \theta, g, h) = a^\beta(z; g) \theta + \nu^\beta(z; g, h)$, we have the identity 
\[
\partial_h M_n^\beta(\wh{\theta}, \wh{g}, h_0)(h_0 - \wh{h}) = \partial_h \calV_n^\beta(\wh{g}, h_0)(h_0 - \wh{h}).
\]
Using this, we can rewrite $I_{3, n}$ as 
\begin{align*}
I_{3, n} &= \partial_h \calV^\beta_n(\wh{g}, h_0)(h_0 - \wh{h}) + o_\P(n^{-1/2}) \\
&= \Big(\underbrace{\left\{\partial_h \calV^\beta_n(\wh{g}, h_0) - \partial_h \calV^\beta_n(g^\beta, h_0)\right\}}_{=: I_{3, n}^{(1)}} + \underbrace{\left\{\partial_h \calV^\beta_n(g^\beta, h_0)- \partial_h\calV^\beta(g^\beta, h_0)\right\}}_{=: I_{3, n}^{(2)}}\Big)(h_0 - \wh{h})   \\
&\qquad + \partial_h\calV^\beta(g^\beta, h_0)(h_0 - \wh{h}) + o_\P(n^{-1/2}) \\
&= I_{3, n}^{(1)}(h_0 - \wh{h}) + I_{3, n}^{(2)}(h_0 - \wh{h}) + J^\ast (h_0 - \wh{h}) + o_\P(n^{-1/2}),
\end{align*}
where the second equality comes from adding and subtracting $\partial_h \calV^\beta_n(g^\beta, h_0)$ and $\partial_h \calV^\beta(g^\beta, h_0)$, and the final equality follows from the assumption that $\partial_h \calV^\beta(g^\beta, h_0) = J^\ast + o(1)$ (Condition~\ref{cond:h_jacobian}) and the fact that $\|h_0 -\wh{h}\|_2^2 = O_\P(n^{-1/2})$ (implied by asymptotic linearity, Condition~\ref{cond:h_linearity}). We aim to show asymptotic linearity of the form
\[
I_{3, n} = -\frac{1}{n}\sum_{i = 1}^n J^\ast f_h(Z_i) + o_\P(n^{-1/2}),
\]
and to do this it suffices to argue that $I_{3, n}^{(1)}, I_{3, n}^{(2)} = o_\P(1)$, since we know $\wh{h} - h_0 = n^{-1}\sum_i f_h(Z_i) + o_\P(1)$ (Condition~\ref{cond:h_linearity}). To show the first point, we leverage mean-squared continuity of $\partial_h \calV^\beta(g, h_0)$ in $g \in L^2(P_Z)$ (Condition~\ref{cond:h_cont}). We see we have
\begin{align*}
\E_{Z_1, \dots, Z_n}[\|I_{3, n}^{(1)}\|_{F}] &= \E_Z\left[\left\|\frac{1}{n}\sum_{i = 1}^n\left\{\partial_h\nu^\beta(Z_i; \wh{g}, h_0) - \partial_h\nu^\beta(Z_i; g^\beta, h_0)\right\}\right\|_F\right] \\
&\leq \E_Z\left[\frac{1}{n}\sum_{i = 1}^n\left\|\partial_h\nu^\beta(Z_i; \wh{g}, h_0) - \partial_h\nu^\beta(Z_i; g^\beta, h_0)\right\|_F\right] \\
&\leq \E_Z\left[\|\partial_h \nu^\beta(Z; \wh{g}, h_0) - \partial_h \nu^\beta(Z; g^\beta, h_0)\|_F^2\right]^{1/2} \\
&\lesssim \|\wh{g} - g^\beta\|_{L^2(P_W)} \land \sqrt{2p^2D}\\
&= o_\P(1),
\end{align*}
where the first inequality follows from a repeated application of the triangle inequality and $D$ is an absolute constant that exists because we have assumed a uniform $L^\infty$ bound of $D > 0$ on the components of the Jacobian for any $j, k \in [p]$ and $\beta > 0$ (Condition~\ref{cond:h_jacobian}). The final line follows from nuisance consistency, which is assumed in Condition~\ref{cond:hessian}.
Now, conditionally applying Chebyshev's inequality yields, for any $\epsilon >0$, that
\begin{align*}
\lim_{n \rightarrow \infty}\P\left(\left\|I_{3, n}^{(1)} \right\|_F \geq \epsilon\right)  = \lim_{n \rightarrow \infty}\E\left[\P_Z\left(\left\|I_{3, n}^{(1)} \right\|_F \geq \epsilon\right)\right]\lesssim \lim_{n \rightarrow \infty}\frac{1}{\epsilon}\E\left[\|\wh{g} - g^\beta\|_{L^2(P_X)} \land (2p^2D)\right] = 0
\end{align*}
where the final limit holds by Proposition~\ref{prop:vitali}, since the collection of random variables $\Big(\|\wh{g} - g^\beta\|_{L^2(P_X)} \land (2p^2 D)\Big)_{n \geq 1}$ is bounded and hence uniformly integrable. Thus, we have $I_{3, n}^{(1)} = o_\P(1)$. Next, note that 
\begin{align*}
\E\left\|I_{3, n}^{(2)}\right\|_F^2 &= \E\left[\left\|\frac{1}{n}\sum_{i = 1}^n \left\{\partial_h\nu^\beta(Z_i; g^\beta, h_0) - \E\partial_h\nu^\beta(Z_i; g^\beta, h_0)\right\}\right\|_F^2\right] \\
&= \frac{1}{n}\E\left[\left\|\partial_h\nu^\beta(Z; g^\beta, h_0) - \E\partial_h\nu^\beta(Z; g^\beta, h_0)\right\|_F^2\right] \\
&\lesssim \frac{1}{n},
\end{align*}
where the final inequality follows from the fact that the  second moments of the scores are uniformly bounded (which in turn follows from Condition~\ref{cond:regularity} after an application of Jensen's inequality). Thus, we have $I_{3, n}^{(2)} = o_\P(1)$ from an application of Chebyshev's inequality, unconditionally.

\paragraph{Concluding the proof.}
Thus far, we have shown that 
\begin{align*}
A^\beta(\wh{g})(\wh{\theta} - \theta^\beta) &= I_{1, n} + I_{2, n} + I_{3, n} \\
&= -\frac{1}{n}\sum_{i = 1}^n m^\ast(Z_i;\theta_0, g_0, h_0) - \frac{1}{n}\sum_{i = 1}
^n J^\ast f_h(Z_i) + o_\P(n^{-1/2}) + o_\P\left(\|\wh{\theta} - \theta^\beta\|_2\right).
\end{align*}
Since we have assumed that $A^\beta(g^\beta) = A^\ast(g_0) + o(1)$ (Condition~\ref{cond:score}) and since $A^\beta(g)$ is assumed to be continuous in $g$ (in the sense of Condition~\ref{cond:regularity}), we have 
\begin{align*}
A^\beta(\wh{g}) &= A^\ast(g_0) + (A^\beta(g^\beta) - A^\ast(g_0)) + (A^\beta(\wh{g}) - A^\beta(g^\beta)) \\
&= A^\ast(g_0) + o_\P(1),
\end{align*}
where the final equality follows from the consistency of $\wh{g}$ (Condition~\ref{cond:hessian}).
Therefore, we have the following equality
\begin{align*}
A^\ast(g_0) \left(\wh{\theta} - \theta^\beta\right) &= -\frac{1}{n}\sum_{i = 1}^n\left\{m^\ast(Z_i;\theta_0, g_0, h_0) + 
 J^\ast f_h(Z_i)\right\} + o_\P(n^{-1/2}) + o_\P\left(\|\wh{\theta} - \theta^\beta\|_2\right).
\end{align*}
Left multiplying both sides by $A^\ast(g_0)^{-1}$
(which exists by Assumption~\ref{cond:regularity}) yields
\begin{align*}
\wh{\theta} - \theta^\beta &= -\frac{1}{n}A^\ast(g_0)^{-1}\sum_{i = 1}^n\left\{m^\ast(Z_i;\theta_0, g_0, h_0) + J^\ast f_h(Z_i)\right\} + o_\P(n^{-1/2}) + o_\P\left(\|\wh{\theta} - \theta^\beta\|_2\right) \\
&= -\frac{1}{n}A^\ast(g_0)^{-1}\sum_{i = 1}^n\left\{m^\ast(Z_i;\theta_0, g_0, h_0) + J^\ast f_h(Z_i)\right\} + o_\P(n^{-1/2}).
\end{align*}
To justify the last line, we just need to show that $\|\wh{\theta} - \theta^\beta\|_2 = O_\P(n^{-1/2})$. To demonstrate this, we note that we can explicitly write 
\begin{align*}
\left\|\wh{\theta} - \theta^\beta\right\|_2 \leq \left\|\frac{1}{n}A^\ast(g_0)^{-1}\sum_{i = 1}^n \left\{m^\ast(Z_i; \theta_0, g_0, h_0) + J^\ast f_h(Z_i)\right\}\right\| + r_n + \epsilon_n \|\wh{\theta} - \theta^\beta\|_2,
\end{align*}
where $(\epsilon_n)_{n \geq 1}$ is a non-negative sequence satisfying $\epsilon_n \xrightarrow[n \rightarrow \infty]{\P} 0$ and $(r_n)_{n \geq 1}$ is a non-negative sequence satisfying $n^{1/2}r_n \xrightarrow[n \rightarrow \infty]{\P} 0$. Define the ``good'' event $G_n := \{\epsilon_n \leq 1/2\}$. Convergence in probability yields that $\lim_{n \rightarrow \infty}\P(G_n^c) = 0$. On $G_n$, we have
\begin{align*}
\left\|\wh{\theta} - \theta^\beta\right\|_2 &\leq 2\left(\frac{1}{n}\left\|A^\ast(g_0)^{-1}\sum_{i = 1}^n \left\{m^\ast(Z_i; \theta_0, g_0, h_0) + J^\ast f_h(Z_i)\right\}\right\|_2\right) + 2r_n \\
&= O\left(\frac{1}{n}\left\|A^\ast(g_0)^{-1}\sum_{i = 1}^n \left\{m^\ast(Z_i; \theta_0, g_0, h_0) + J^\ast f_h(Z_i)\right\}\right\|_2\right) + o_\P(n^{-1/2}) \\
&= O_\P\left(n^{-1/2}\right),
\end{align*}
where the final equality follows since $A^\ast(g_0)^{-1}\left\{m^\ast(Z_i, \theta_0, g_0, h_0) + J^\ast f_h(Z_i)\right\}$ are i.i.d.\ random vectors whose components have finite variance.
This thus completes the proof of Theorem~\ref{thm:smooth_clt_generic}
\end{proof}

\begin{proof}[Proof of Theorem~\ref{thm:smooth_clt}]
First, we note that we have the following decomposition.
\begin{align*}
\sqrt{n}(\wh{\theta} - \theta^{\beta}) &= \sqrt{n}\left\{\P_n m^{\beta}(Z; \wh{g}) - \E_Z m^{\beta}(Z; g^{\beta})\right\} \\
&= \sqrt{n}\left\{\P_n m^{\beta}(Z; \wh{g}) \pm \P_n m^{\beta}(Z; g^{\beta}) \pm \E_Z m^{\beta}(Z; \wh{g}) \pm \E_Z m^{\beta}(Z; g^{\beta}) - \E_Z m^{\beta}(Z; g^{\beta})\right\} \\
&= \sqrt{n}\left\{\P_n m^{\beta}(Z; \wh{g}) - \E m^{\beta}(Z; \wh{g}) - \P_n m^{\beta}(Z; g^{\beta}) + \E_Z m^{\beta}(Z; g^{\beta})\right\} \\
&\qquad\qquad + \sqrt{n}\left\{\P_n m^{\beta}(Z; g^{\beta}) - \E_Z m^{\beta}(Z; g^{\beta})\right\} + \sqrt{n}\left\{\E_Z m^{\beta}(Z; \wh{g}) - \E_Z m^{\beta}(Z; g^{\beta})\right\} \\
&= \underbrace{\G_n m^{\beta}(Z; \wh{g}) - \G_n m^{\beta}(Z; g^{\beta})}_{T_1} + \underbrace{\G_n m^{\beta}(Z; g^{\beta})}_{T_2} \\
&\qquad +\underbrace{\sqrt{n}\left\{\E_Z m^{\beta}(Z; \wh{g}) - \E_Z m^{\beta}(Z; g^{\beta})\right\}}_{T_3}.
\end{align*}
We handle each of the three terms separately. Handling the first term follows by assumption, as we have by Condition~\ref{cond:equi}
\[
T_1 = \G_n m^{\beta}(Z; \wh{g}) - \G_n m^{\beta}(Z; g^{\beta}) = o_\P(1).
\]
Handling $T_3$ follows from the assumptions regarding Neyman orthogonality and bounded second order errors. In particular, we have:
\begin{align*}
\frac{1}{\sqrt{n}}T_3 &= \E_Z m^{\beta}(Z; \wh{g}) - \E_Z m^{\beta}(Z; g^{\beta}) \\
&= D_{g}\E_Z[m^{\beta}(Z; g^{\beta})](\wh{g} - g^{\beta}) + \frac{1}{2}D^2_g\E_Z[m^{\beta}(Z; \wb{g})](\wh{g} - g^{\beta}) \\
&= \frac{1}{2}D^2_g\E_Z[m^{\beta}(Z; \wb{g})](\wh{g} - g^{\beta}) &(\text{Condition~\ref{cond:neyman}}) \\
&= o_\P(n^{-1/2}) &(\text{Condition~\ref{cond:hessian}}),
\end{align*}
where $\wb{g} \in [\wh{g}, g^\beta]$ exists by the mean value theorem. Multiplying everything by $\sqrt{n}$ yields that $T_3 = o_\P(1)$. 



The last thing to do is analyze $T_2$. 
We can rewrite $T_2$ as follows:
\[
T_2 = \G_n m^\ast(Z; g_0) + \left\{\G_n m^{\beta}(Z; g^{\beta}) - \G_n m^\ast(Z; g_0)\right\}.
\]
We need to show the last term is $o_\P(1)$. We can control this term by noting
\begin{align*}
\E|\G_n m^{\beta}(Z; g^{\beta}) - \G_n m^\ast(Z; g_0)|^2 &= \Var\left[\frac{1}{\sqrt{n}}\sum_{i = 1}^n\left\{m^\beta(Z_i; g^\beta) - m^\ast(Z_i; g_0)\right\}\right]\\
&=\Var\left[m^{\beta}(Z; g^{\beta}) - m^\ast(Z; g_0)\right] \\
&\leq \E\left|m^{\beta}(Z; g^{\beta}) - m^\ast(Z; g_0)\right|^2 \\
&= o(1).
\end{align*}

We now justify the final equality. In particular, we leverage the two propositions stated before the proof to show that 
\[
\lim_{n \rightarrow \infty}\E\left|m^{\beta}(Z; g^{\beta}) - m^\ast(Z; g_0)\right|^2 = 0.
\]
We want to this by applying Vitali's Theorem (Proposition~\ref{prop:vitali}). Since we already know that $m^{\beta_n}(Z; g^{\beta_n}) - m^\ast(Z; g_0) = o_\P(1)$ by Condition~\ref{cond:simp_score}, it suffices to check that the sequence $\big(\{m^{\beta}(Z; g^{\beta} - m^\ast(Z;g_0)\}^2\big)_{n \geq 1}$ is uniformly integrable. To see this, we take $\Phi(t) := t^{1 + \frac{\epsilon}{2}}$, where $\epsilon > 0$ is as in Condition~\ref{cond:simp_reg}. We have the bound
\begin{align*}
\sup_{n \geq 1}\E\left[\Phi\left(\left|m^{\beta}(Z; g^{\beta}) -m^\ast(Z; g_0)\right|\right)\right] &= \sup_{n \geq 1}\E\left[\left|m^{\beta}(Z; g^{\beta}) - m^\ast(Z; g_0)\right|^{2 + \epsilon}\right] \\
&\leq \sup_{n \geq 1} \left\{\left(\E\left|m^\beta(Z; g^\beta)\right|^{2 + \epsilon}\right)^{\frac{1}{2  + \epsilon}} + \left(\E\left|m^\ast(Z; g_0)\right|^{2 + \epsilon}\right)^{\frac{1}{2 + \epsilon}}\right\}^{2 + \epsilon}\\
&\leq (2D)^{2 + \epsilon},
\end{align*}
where the final inequality follows from the assumed bound $\left\|m^\beta(Z; g^\beta)\right\|_{L^{2 + \epsilon}(P_Z)}, \left\|m^\ast(Z; g_0)\right\|_{L^{2 + \epsilon}(P_Z)} \leq D$ (Condition~\ref{cond:score}). Thus the sequence in uniformly integrable by Proposition~\ref{prop:pous}
From this, we see that we have 
\[
\P\left(\left|\G_n m^{\beta}(Z; g^{\beta}) - \G_n m^\ast(Z; g_0)\right| \geq \delta\right) \leq \frac{\E\left|\G_n m^{\beta}(Z; g^{\beta}) - \G_n m^\ast(Z; g_0)\right|^2 }{\delta} = o(1)
\]
for any $\delta  > 0$. Thus, we have 
\[
T_2 = \G_n m^\ast(Z; g_0) + \{\G_n m^\beta(Z; g^\beta) - \G_n m^\ast(Z; g_0)\} = \frac{1}{\sqrt{n}}\sum_{i = 1}^n m^\ast(Z_i; g_0) + o_\P(1).
\]
Since we have already shown that $T_1, T_3 = o_\P(1)$, this completes the proof.
\end{proof}

\end{document}